\documentclass[11pt]{article}
\usepackage{etex}
\usepackage[paper=letterpaper,margin=1in]{geometry}
\newif\iflong
\newif\ifshort

\longtrue

\iflong
\else
\shorttrue
\fi
\usepackage{graphics}
\usepackage[demo]{graphicx}
\usepackage{caption}
\usepackage{subcaption}
\usepackage{epsfig}
\usepackage{arcs}
\usepackage[T1]{fontenc}
\usepackage{alltt}
\usepackage{amssymb,amsfonts,epsf,url}
\usepackage{graphicx}
\usepackage{pstricks,pstricks-add}
\usepackage{pst-node}
\usepackage{pst-coil}
\usepackage{amsmath}
\usepackage{amsthm}
\usepackage{verbatim}
\usepackage{xcolor}
\usepackage[textsize=footnotesize]{todonotes}
\usepackage{etex}
\usepackage{graphics}
\usepackage{epsfig}
\usepackage{arcs}
\usepackage[T1]{fontenc}
\usepackage{alltt}
\usepackage{amssymb,amsfonts,epsf,url}
\usepackage{graphicx}
\usepackage{pstricks,pstricks-add}
\usepackage{pst-node}
\usepackage{pst-coil}
\usepackage{amsmath}
\usepackage{amsthm}
\usepackage{verbatim}
\usepackage{xcolor}
\usepackage[textsize=footnotesize]{todonotes}
\usepackage{fancyhdr}

\newtheorem{claim}{Claim}
\newtheorem{theorem}{Theorem}[section]
\newtheorem{lemma}[theorem]{Lemma}
\newtheorem{corollary}[theorem]{Corollary}

\newtheorem{assumption}[theorem]{Assumption}
\newtheorem{definition}[theorem]{Definition}

\newtheorem{observation}[theorem]{Observation}
\newtheorem{rem}[theorem]{Remark}

\newcommand{\app}{$\mathrm{(\spadesuit)}$}
\newcommand{\appno}{$\mathrm{\spadesuit}$}

\newcommand{\paramproblem}[3]{\noindent {\sc #1}
\\
{\bf Given:} #2\\
{\bf Question:} #3}

\newcommand{\Oh}{{\mathcal O}}
\newcommand{\nat}{\mathbb{N}}

\newcommand{\XP}{\mbox{\sf XP}}
\newcommand{\Pol}{\mbox{\sf P}}

\newcommand{\NP}{\mbox{{\sf NP}}}
\newcommand{\FPT}{\mbox{{\sf FPT}}}
\newcommand{\W}{\mbox{{\sf W}}}
\def\eg{{\em e.g.}}

\def\ie{{\em i.e.}}
\def\etal{{\em et al.}}

\newcommand{\Col}{\chi}
\newcommand{\PPP}{\mathcal{P}}
\newcommand{\VVV}{\mathcal{V}}
\newcommand{\MMM}{\mathcal{M}}
\newcommand{\DDD}{\mathcal{D}}
\newcommand{\TTT}{\mathcal{T}}
\newcommand{\RRR}{\mathcal{R}}
\newcommand{\SSS}{\mathcal{S}}
\newcommand{\WWW}{\mathcal{W}}
\psset{unit=1pt}

\newcommand{\AND}{{\sc and}}
\newcommand{\OR}{{\sc or}}
\newcommand{\mor}{{\sc Colored Path}}
\newcommand{\cmor}{{\sc Colored Path-Con}}
\newcommand{\bcmor}{{\sc Bounded-length Colored Path-Con}}
\newcommand{\bicmor}{{\sc Bounded-intersection Colored Path-Con}}
\newcommand{\gmor}{{\sc Obstacle Removal}}
\newcommand{\gcmor}{{\sc Connected Obstacle Removal}}

\date{}
\ifshort
\title{How to navigate through obstacles?}
\fi

\iflong
\title{How to navigate through obstacles?}
\fi

\author{Eduard Eiben\thanks{Algorithms and Complexity Group, TU Wien, Austria \& Department of Informatics, University of Bergen, Bergen, Norway.  
	Email: {\tt eduard.eiben@uib.no}} \and Iyad Kanj\thanks{School of Computing, DePaul University, Chicago, USA. Email: {\tt ikanj@cs.depaul.edu}}}

\begin{document}

\maketitle

\begin{abstract}
Given a set of obstacles and two points in the plane, is there a path between the two points that does not cross more than $k$ different obstacles?
\iflong Equivalently, can we remove $k$ obstacles so that there is an obstacle-free path between the two designated points? \fi
This is a fundamental problem that has undergone a tremendous amount of work by researchers in various areas, including computational geometry, graph theory, wireless computing, and motion planning.
It is known to be \NP-hard, even when the obstacles are very simple geometric shapes (\eg, unit-length line segments). The problem can be formulated and generalized into the following graph problem: Given a planar graph $G$ whose vertices are colored by color sets, two designated vertices $s, t \in V(G)$, and $k \in \nat$, is there an $s$-$t$ path in $G$ that uses at most $k$ colors? If each obstacle is connected, the resulting graph  satisfies the color-connectivity property, namely that each color induces a connected subgraph.

We study the complexity and design algorithms for the above graph problem with an eye on its geometric applications.
\iflong We prove a set of hardness results, among which a result showing that the color-connectivity property is crucial for any hope for fixed-parameter tractable (\FPT) algorithms (even for various restrictions and parameterizations of the problem), as without it, the problem is \W[SAT]-hard parameterized by $k$. Previous results only implied that the problem is \W[2]-hard. A corollary of the aforementioned result is that, unless \W[2] $=$ \FPT{}, the problem cannot be approximated in \FPT{} time to within a factor that is a function of $k$. \fi
\ifshort We prove a set of hardness results, among which a result showing that the color-connectivity property is crucial for any hope for fixed-parameter tractable (\FPT) algorithms, as without it, the problem is \W[SAT]-hard parameterized by $k$. Previous results only implied that the problem is \W[2]-hard. A corollary of this result is that, unless \W[2] $=$ \FPT{}, the problem cannot be approximated in \FPT{} time to within a factor that is a function of $k$. \fi~By describing a generic plane embedding of the graph instances, we show that our hardness results translate to the geometric instances of the problem.

\ifshort
We then focus on graphs satisfying the color-connectivity property. By exploiting the planarity of the graph and the connectivity of the colors, we develop topological results that allow us to prove that, for any vertex $v$, there exists a set of paths whose cardinality is upper bounded by a function of $k$, that ``represents'' the valid $s$-$t$ paths containing subsets of colors from $v$. We employ these structural results to design an \FPT{} algorithm for the problem parameterized by both $k$ and the treewidth of the graph, and extend this result further to obtain an \FPT{} algorithm for the parameterization by both $k$ and the length of the path. The latter result generalizes and explains previous \FPT{} results for various obstacle shapes, such as unit disks and fat regions.\fi

\iflong
We then focus on graphs satisfying the color-connectivity property. By exploiting the planarity of the graph and the connectivity of the colors, we develop topological results that reveal rich structural properties of the problem. These results allow us to prove that, for any vertex $v$ in the graph, there exists a set of paths whose cardinality is upper bounded by some function of $k$, that ``represents'' the valid $s$-$t$ paths containing subsets of colors from $v$. We employ these structural results to design an \FPT{} algorithm for the problem parameterized by both $k$ and the treewidth of the graph, and extend this result further to obtain an \FPT{} algorithm for the parameterization by both $k$ and the length of the path. The latter result directly implies and explains previous \FPT{} results for various obstacle shapes, such as unit disks and fat regions.\fi
 \end{abstract}
 
\pagenumbering{roman}
\pagenumbering{arabic}

\section{Introduction}
\label{sec:intro}
We consider the following problem: Given a set of obstacles and two designated points in the plane, is there a path between the two points that does not cross more than $k$ obstacles?
Equivalently, can we remove $k$ obstacles so that there is an obstacle-free path between the two designated points?
We refer to this problem as \gmor{}, and to its restriction to instances in which each obstacle is connected as \gcmor{}.

By considering the auxiliary plane graph that is the dual of the plane subdivision determined by the obstacles\footnote{We assume that the regions formed by the obstacles can be computed in polynomial time. \iflong We do not assume that the obstacles contain their interiors.\fi If the intersection of two obstacles is not a 2-D region, we can thicken the borders of the obstacles without changing the sets of obstacles they intersect, so that their intersection becomes a 2-D region.}, \gmor{} was formulated and generalized into the following graph problem, referred to as \mor{} (see Figure~\ref{fig:tree} for illustrations):

\paramproblem{\mor} {A planar graph $G$; a set of colors $C$; $\chi: V \longrightarrow 2^{C}$; two designated vertices $s, t \in V(G)$; and $k \in \nat$}{Does there exist an $s$-$t$ path in $G$ that uses at most $k$ colors?} \\

Denote by \cmor{} the restriction of \mor{} to instances in which each color induces a connected subgraph of $G$.

\begin{figure}[htbp]

 \begin{subfigure}{.48\textwidth}
 \begin{center}
\includegraphics[width=6cm]{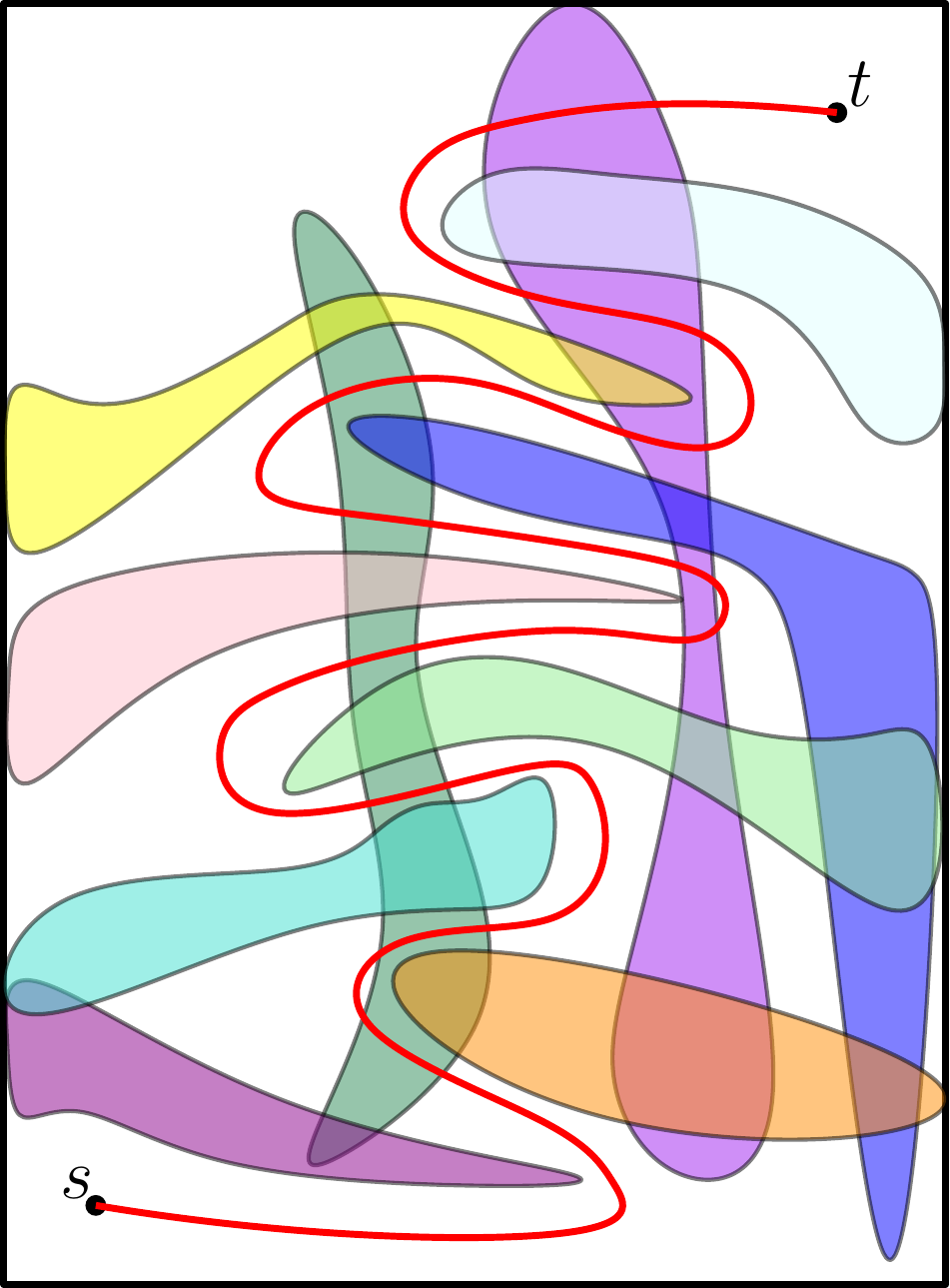}
\end{center}
\end{subfigure}
 \begin{subfigure}{.48\textwidth}
 \begin{center}
\includegraphics[width=6cm]{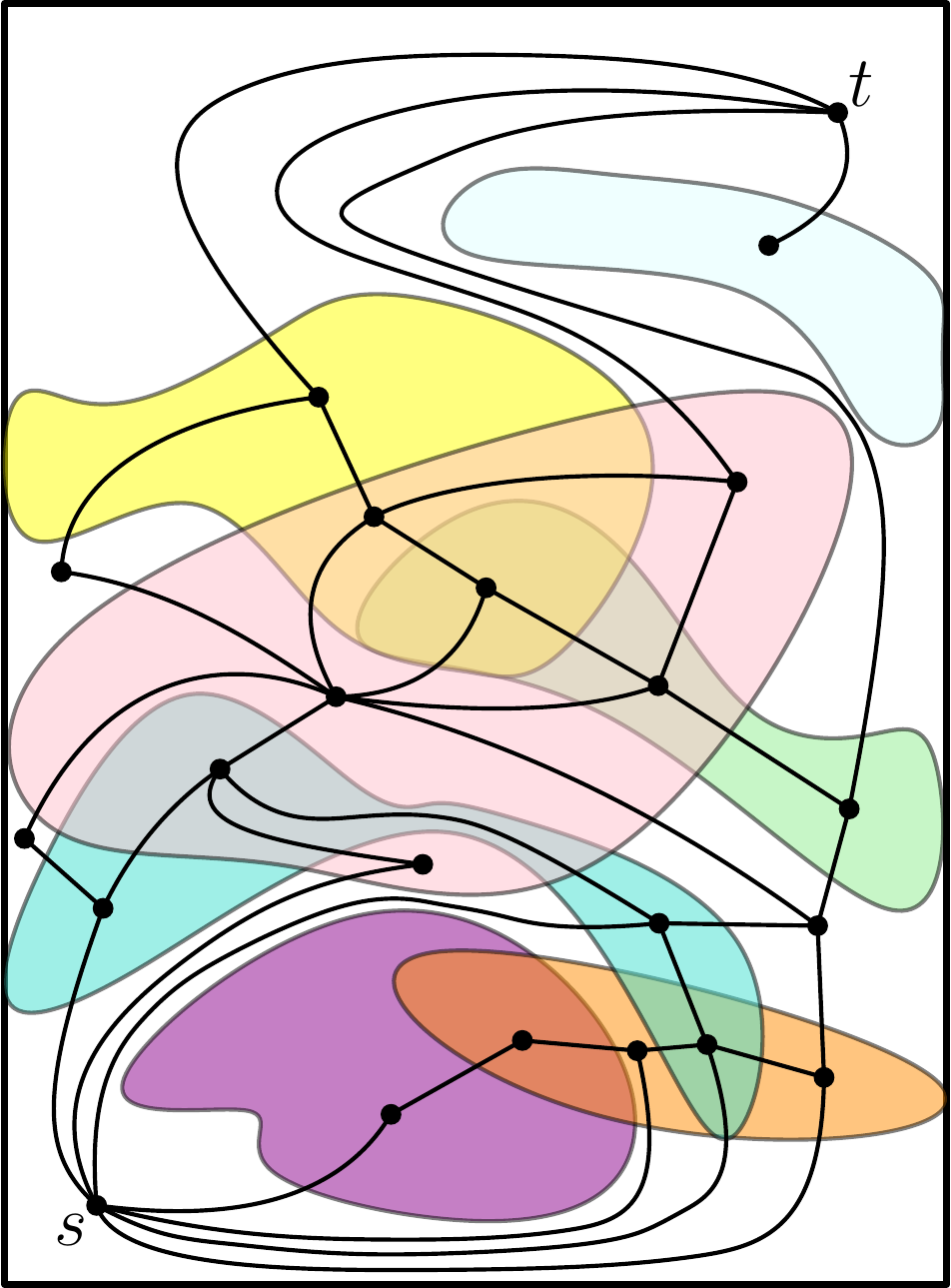}
 \end{center}
\end{subfigure}

\caption{Illustration of instances of the problem under consideration drawn within a bounding box. The left figure shows an instance in which the optimal path crosses two obstacles, zigzagging between the other obstacles. The right figure shows an instance and its auxiliary plane graph.}
\label{fig:tree}
\end{figure}

 As we discuss next, \gcmor{} and \mor{} are fundamental problems that have undergone a tremendous amount of work, albeit under different names and contexts, by researchers in various areas, including computational geometry, graph theory, wireless computing, and motion planning.

\subsection{Related Work}
\label{subsec:related}

In motion planning, the goal is generally to move a robot from a starting position to a final position, while avoiding collision with a set of obstacles. This is usually referred to as the {\em piano-mover}'s problem. \gmor{} is a variant of the piano-mover's problem, in which the obstacles are in the plane and the robot is represented as a point. Since determining if there is an obstacle-free path for the robot in this case is solvable in polynomial time, if no such path exists, it is natural to seek a path that intersects as few obstacles as possible. Motivated by planning applications, \gcmor{} and \mor{} were studied under the name {\sc Minimum Constraint Removal}~\cite{aaai,lavalle,popov,hauser}. \gcmor{} has also been studied extensively, motivated by applications in wireless computing, under the name {\sc Barrier Coverage} or {\sc Barrier Resilience}~\cite{alt,kirkpatrick3,korman,kumar,kirkpatrick1,yangphd}. In such applications, we are given a field covered by sensors (usually simple shapes such as unit disks), and the goal is to compute a minimum set of sensors that need to fail before an entity can move undetected between two given sites.

Kumar \etal~\cite{kumar} were the first to study \gcmor{}. They showed that for unit-disk obstacles in some restricted setting, the problem can be solved in polynomial time. The complexity of the general case for unit-disk obstacles remains open. Several works showed the \NP-hardness of the problem, even when the obstacles are very simple geometric shapes such as line segments (\eg, see~\cite{alt,kirkpatrick1,yangphd}). The complexity of the problem when each obstacle intersects a constant number of other obstacles is open~\cite{lavalle,hauser}.

Bereg and Kirkpatrick~\cite{kirkpatrick3} designed approximation algorithms when the obstacles are unit disks by showing that the length, referred to as the \emph{thickness}~\cite{kirkpatrick3} (\ie, number of regions visited), of a shortest path that crosses $k$ disks is at most $3k$; this follows from the fact that a shortest path does not cross a disk more than a constant number of times.

Korman \etal~\cite{korman} showed that \gcmor{} is \FPT{} parameterized by $k$ for unit-disk obstacles, and extended this result to similar-size fat-region obstacles with a constant \emph{overlapping number}, which is the maximum number of obstacles having nonempty intersection. Their result draws the observation, which was also used in~\cite{kirkpatrick3}, that for unit-disk (and fat-region) obstacles, the length of an optimal path can be upper bounded by a linear function of the number of obstacles crossed (\ie, the parameter). This observation was then exploited by a branching phase that decomposes the path into subpaths in (simpler) restricted regions, enabling a similar approach to that of Kumar \etal~\cite{kumar}.

Motivated by its applications to networking, among other areas, the problem of computing a minimum-colored path in a graph received considerable attention (\eg, see~\cite{carr,varma}). The problem was shown to be \NP-hard in several works~\cite{carr,kirkpatrick2,hauser,varma}.\iflong \footnote{We note that some works consider the edge-colored version of the problem, but for all purposes considered in this paper the two versions are equivalent.}\fi~Most of the \NP-hardness reductions start from {\sc Set Cover}, and result in instances of \mor{} (\ie, planar graphs), as was also observed by~\cite{kirkpatrick3}. These reductions are \FPT-reductions, implying the \W[2]-hardness of \mor{}. Moreover, these reductions imply that, unless \Pol{} $=$ \NP, the minimization version of \mor{} cannot be approximated to within a factor of $c \lg{n}$, for any constant $c < 1$. Hauser~\cite{hauser}, and Gorbenko and Popov~\cite{popov}, implemented exact and heuristic algorithms for the problem on general graphs.
Very recently, Eiben \etal~\cite{aaai} designed exact and heuristic algorithms for \mor{} and \gmor{}, and proved computational lower bounds on their subexponential-time complexity, assuming the Exponential Time Hypothesis.

\iflong Finally, we mention that there is a related problem that is solvable in polynomial time, which has received considerable attention~\cite{dannychen,hersh2,hersh1}, where the goal is to find a shortest path w.r.t.~the Euclidean length between two given points in the plane that intersects at most $k$ obstacles. The \mor{} problem also falls into the category of many computationally-hard problems on colored graphs, where the objective is to compute a graph structure satisfying certain (desired) properties that uses the minimum number of colors (\eg, see~\cite{fgk}).\fi

\subsection{Our Results and Techniques}
\label{subsec:contributions}
We study the complexity and parameterized complexity of \mor{} and \cmor{}, eyeing the implications on their geometric counterparts \gmor{} and \gcmor{}, respectively. \iflong We do not treat the problem on general graphs because, as we point out in Remark~\ref{rem:nonplanar}, this problem is computationally very hard, even when restricted to graphs satisfying the color-connectivity property. \fi \ifshort The hardness results we obtain are too long and technical to be included/sketched given the space limit, and are delegated to the full paper in the appendix, and so are the complete proofs and details for the text marked with \app. \fi

\iflong Our first set of hardness results show that both problems are \NP-hard, even when restricted to graphs of small outerplanarity and pathwidth, and that it is unlikely that they can be solved in subexponential time:

\begin{itemize}
\item[(i)] \mor{} is~\NP-complete, even for outerplanar graphs of pathwidth at most 2 and in which every vertex contains at most one color (Theorem~\ref{thm:pathwidthnphardness}).

\item[(ii)] \cmor{} is \NP-complete even for 2-outerplanar graphs of pathwidth at most 3 (Corollary~\ref{cor:connectedpathwidthnphardness}).

\item[(iii)] Unless ETH fails, \cmor{} (and hence \mor{}) is not solvable in subexponential time, even for 2-outerplanar graphs of pathwidth at most 3 and in which each color appears at most 4 times (Corollary~\ref{cor:pathwidthnphardness}).
\end{itemize}

The reduction used to prove (i) produces instances of \mor{} that can be realized as geometric instances of \gmor{} whose overlapping number is at most 2. Thus, this hardness result extends to the aforementioned restriction of \gmor{}. The same reduction is then modified to yield (ii) and (iii) for \cmor{}; this reduction produces instances of \mor{} that can be realized as geometric instances of \gcmor{} whose overlapping number is at most 4, again showing that the hardness results extend to these restrictions of \gcmor{}.

We then study the parameterized complexity of \mor{} and \cmor{}.\fi~Clearly, \mor{} is in the parameterized class \XP. We show that the color-connectivity property is crucial for any hope for an \FPT-algorithm, since even very restricted instances and combined parameterizations of \mor{} are \W[1]-complete:

\begin{itemize}
\iflong \item[(iv)] \fi \ifshort \item[(i)] \fi \mor{}, restricted to instances of pathwidth at most 4, and in which each vertex contains at most one color and each color appears on at most 2 vertices, is \W\rm{[}1\rm{]}-complete parameterized by $k$ \iflong (Theorem~\ref{thm:wcompleteoccurences})\fi \ifshort(\appno{} Theorem~3.8)\fi.
\iflong \item[(v)] \fi \ifshort \item[(ii)] \fi \mor{}, parameterized by both $k$ and the length of the sought path $\ell$, is \W\rm{[}1\rm{]}-complete \iflong (Theorem~\ref{thm:wcompletelength})\fi \ifshort (\appno{} Theorem~3.7)\fi.
\end{itemize}

Without restrictions, the problem sits high in the parameterized complexity hierarchy:

\begin{itemize}
\iflong \item[(vi)] \fi \ifshort \item[(iii)] \fi \mor{}, parameterized by $k$, is \W[SAT]-hard \iflong (Theorem~\ref{thm:wsathard}) \fi \ifshort (\appno{} Theorem~3.10) \fi and is in \W\rm{[}\Pol\rm{]} \iflong (Theorem~\ref{thm:inwp})\fi \ifshort(\appno{} Theorem~3.9)\fi.
\end{itemize}

A corollary of \iflong (vi) \fi \ifshort (iii) \fi is that, unless \W[2] $=$ \FPT{}, \mor{} cannot be approximated in \FPT{} time to within a factor that is a function of $k$ \iflong (Corollary~\ref{cor:apx})\fi \ifshort (\appno{} Corollary~3.12)\fi.

By producing a generic construction \iflong (Remark~\ref{rem:gadget}) \fi \ifshort (\appno{} Remark~3.4) \fi that can be used to realize any graph instance of \mor{} as a geometric instance of \gmor{}, the hardness results in  \iflong (iv)--(vi)\fi \ifshort (i)--(iii)\fi, and the inapproximability result discussed above, translate to \gmor{}. This geometric realization may slightly increase the overlapping number by at most 2. Previously, \mor{} was only known to be \W[2]-hard, via the standard reduction from {\sc Set Cover}~\cite{carr,hauser,varma}. Our results refine the parameterized complexity and approximability of \mor{} and \gmor{}.

As noted in \iflong Remark~\ref{rem:nonplanar}\fi \ifshort Remark~3.13 \app{}\fi, the color-connectivity property without planarity is hopeless: We can
tradeoff planarity for color-connectivity by adding a single vertex that serves as a color-connector, thus establishing the \W[SAT]-hardness of the problem on apex graphs.

The above hardness results show that we can focus our attention on \cmor{}. We show the following algorithmic result:

\begin{itemize}
\iflong \item[(vii)] \fi \ifshort \item[(iv)] \fi \cmor{}, parameterized by both $k$ and the treewidth $\omega$ of the input graph, is \FPT{} (Theorem~\ref{thm:treewidth}).
\end{itemize}
We remark that bounding the treewidth does not make \cmor{} much easier, as we show in this paper that \cmor{} is \NP-hard even for 2-outerplanar graphs of pathwidth at most 3 \iflong (Corollary~\ref{cor:connectedpathwidthnphardness})\fi \ifshort (\appno{} Corollary~3.2)\fi.

The folklore dynamic programming approach based on tree decomposition, used for the
{\sc Hamiltonian Path/Cycle} problems, does not work for \cmor{} to prove the result in \iflong (vii) \fi \ifshort (iv) \fi for the following reasons. As opposed to the {\sc Hamiltonian Path/Cycle} problems, where it is sufficient to keep track of how the path/cycle interacts with each bag in the tree decomposition, this is not sufficient in the case of \cmor{} because we also need to keep track of which color sets are used on both sides of the bag. Although (by color connectivity) any subset of colors appearing on both sides of a bag must appear on vertices in the bag as well, there can be too many such subsets (up to $|C|^k$, where $C$ is the set of colors), and certainly we cannot afford to enumerate all of them if we seek an \FPT{} algorithm. To overcome this issue, we develop in Section~\ref{sec:structural} topological structural results that exploit the planarity of the graph and the connectivity of the colors to show the following.  For any vertex $w \in V(G)$, and for any pair of vertices $u, v \in V(G)$, the set of (valid) $u$-$v$ paths in $G-w$ that use colors appearing on vertices in the face of $G-w$ containing $w$ can be ``represented'' by a minimal set of paths whose cardinality is a function of $k$.

In Section~\ref{sec:algo}, we extend the notion of a minimal set of paths w.r.t.~a single vertex to a ``representative set'' of paths w.r.t.~a specific bag, and a specific enumerated configuration for the bag, in a tree decomposition of the graph. This enables us to use the upper bound on the size of a minimal set of paths, derived in Section~\ref{sec:structural}, to upper bound the size of a representative set of paths
w.r.t.~a bag and a configuration. This, in turn, yields an upper bound on the size of the table stored at a bag, in the dynamic programming algorithm, by a function of both $k$ and the treewidth of the graph, thus yielding the desired result.


In Section~\ref{sec:extension}, we extend the \FPT{} result for \cmor{} in \iflong (vii) \fi \ifshort (iv) \fi w.r.t.~the parameters $k$ and $\omega$, to the parameterization by both $k$ and the length $\ell$ of the path:

\begin{itemize}
\iflong \item[(viii)] \fi \ifshort \item[(v)] \fi \cmor{}, and hence \gcmor{},  parameterized by both $k$ and $\ell$ is \FPT{} (Theorem~\ref{thm:ext_treewidth}).	
\end{itemize}

The dependency on both $\ell$ and $k$ is \emph{essential} for the result in~\iflong (viii) \fi \ifshort (v)\fi: If we parameterize only by $k$, or only by $\ell$, then the problem becomes \W[1]-hard \iflong (Theorem~\ref{thm:bcmor_hard_k} and Theorem~\ref{thm:bcmor_hard_ell})\fi \ifshort (\appno{} Theorem~6.1 and Theorem~6.2)\fi. By \ifshort \app{}  Remark~3.4\fi \iflong Remark~\ref{rem:gadget}\fi, these two results translate to \gcmor{}.

The result in \iflong (viii) \fi \ifshort (v) \fi generalizes and explains Korman \etal's results~\cite{korman} showing that \gcmor{} is \FPT{} parameterized by $k$ for unit-disk obstacles, which they also generalized to similar-size fat-region obstacles with bounded overlapping number. Their results exploit the obstacle shape to upper bound the length of the path by a linear function of $k$, and then use branching to reduce the problems to a simpler setting. Our result directly implies that, regardless of the (connected) obstacle shapes, as long as the path length is upper bounded by some function of $k$ (Corollary~\ref{cor:boundedlengthaapp}), the problem is \FPT{}. The \FPT{} result in \iflong (viii) \fi \ifshort (v) \fi also implies that:

\begin{itemize}
\iflong \item[(ix)] \fi \ifshort \item[(vi)] \fi  For any computable function $h$, \cmor{} restricted to instances in which each color appears on at most $h(k)$ vertices, is \FPT{} parameterized by $k$ (Corollary~\ref{cor:boundedintersection}).
\end{itemize}

Result \iflong (ix) \fi \ifshort (vi) \fi has applications to \gcmor{}, in particular, to the interesting case when the obstacles are convex polygons, each intersecting a constant number of other polygons. The question about the complexity of this problem was posed in~\cite{lavalle,hauser}, and remains open. The result in \iflong (ix) \fi \ifshort (vi) \fi implies that this problem is \FPT{} (\ifshort \appno{} Theorem~6.19\fi \iflong Theorem~\ref{thm:boundedintersectiongeometric}\fi).

We finally mention that it remains open whether \cmor{} and \gcmor{} are \FPT{} parameterized by $k$ only.

\section{Preliminaries}
\label{sec:prelim}
\ifshort We assume familiarity with graph theory and parameterized complexity. We refer the reader to the full paper in the appendix and the standard books~\cite{diestel,dfbook}. \fi

\iflong
We assume familiarity with the basic notations and terminologies in graph theory and parameterized complexity. We refer the reader to \ifshort the full paper in the appendix and to \fi the standard books~\cite{diestel,dfbook} for more information on these subjects.

\paragraph{{\bf Graphs.}} All graphs in this paper are simple (\ie, loop-less and with no multiple edges).
Let $G$ be an undirected graph. For an edge $e=uv$ in $G$, {\em contracting} $e$ means removing the two vertices $u$ and $v$ from $G$, replacing them with a new vertex $w$, and
for every vertex $y$ in the neighborhood of $v$ or $u$ in $G$, adding an edge $wy$ in the new graph, not allowing multiple edges.
Given a vertex-set $S\subseteq V(G)$, {\em contracting} $S$ means contracting the edges between the vertices in $S$ to obtain a single vertex at the end.

A graph is {\it planar} if it can be drawn in the plane without edge
intersections (except at the endpoints). An \emph{apex graph} is a graph in which the removal of a single vertex results in a planar graph.  A {\it plane graph} has a
fixed drawing. Each maximal connected region of the plane minus the
drawing is an open set; these are the {\em faces}. One is unbounded,
called the {\em outer face}. An {\em outerplane graph} is a plane
graph for which every vertex is incident to the outer face; and {\em
outerplanar graph} is a graph that has such a plane embedding. An \emph{$i$-outerplane graph} (resp.~\emph{$i$-outerplaner graph}), for $i > 1$, is defined inductively as a graph such that the removal of its outer face results in an \emph{$(i-1)$-outerplane graph} (resp.~\emph{$(i-1)$-outerplaner graph}) graph.

Let $S$ be a set of points in the plane, and let $C_1, C_2$ be two non self-intersecting curves that meet $S$ in precisely their common endpoints $a$ and $b$.
We say that $C_1$ and $C_2$ are \emph{isotopic} w.r.t.~$S$ (also known as \emph{homotopic rel.~boundary}) if there is a continuous deformation from $C_1$ to $C_2$ through curves between $a$ and $b$ such that no intermediate curve in this deformation meets a vertex of $S$ in its interior.

Let $W_1=(u_1, \ldots, u_p)$ and $W_2=(v_1, \ldots, v_q)$, $p, q \in \nat$, be two walks such that $u_p =v_1$. Define the \emph{gluing} operation $\circ$ that when applied to $W_1$ and $W_2$ produces that walk
$W_1 \circ W_2 = (u_1, \ldots, u_p, v_2, \ldots, v_q)$.

For a graph $G$ and two vertices $u, v \in V(G)$, we denote by $d_G(u, v)$ the \emph{distance} between $u$ and $v$ in $G$, which the length of a shortest path between $u$ and $v$ in $G$.

\paragraph{{\bf Treewidth, Pathwidth and Tree Decomposition.}}
\begin{definition} \rm
Let $G=(V, E)$ be a graph. A {\it tree decomposition} of $G$ is
a pair $({\cal V}, {\cal T})$ where ${\cal V}$ is a collection of
subsets of $V$ such that $\bigcup_{X_i \in {\cal V}} = V$,
and ${\cal T}$ is a rooted tree whose node set is $\cal V$, such that:

\begin{enumerate}
\item For every edge $\{u, v\} \in E$, there is an $X_i \in {\cal V}$, such
 that $\{u, v\} \subseteq  X_i$; and

\item for all $X_i, X_j, X_k \in {\cal V}$, if the node $X_j$ lies on the
   path between the nodes $X_i$ and $X_k$ in the tree ${\cal T}$, then
   $X_i \cap X_k \subseteq X_j$.

The {\it width} of the tree decomposition $({\cal V}, {\cal T})$ is
defined to be $\max\{|X_i| \mid X_i \in {\cal V} \} - 1$. The
{\it treewidth} of the graph $G$ is the minimum width over all
tree decompositions of $G$.
\end{enumerate}
\end{definition}

A \emph{path decomposition} of a graph $G$ is a tree decomposition $({\cal V}, {\cal T})$ of $G$, where $\TTT$ is a path.
The {\it pathwidth} of a graph $G$ is the minimum width over all
path decompositions of $G$.

A tree decomposition $({\cal V}, {\cal T})$ is {\it nice} if it satisfies the following conditions:

\begin{enumerate}

\item Each node in the tree ${\cal T}$ has at most two children.

\item If a node $X_i$ has two children $X_j$ and $X_k$ in the tree ${\cal T}$,
   then $X_i=X_j=X_k$; in this case node $X_i$ is called a {\em join node}.

\item If a node $X_i$ has only one child $X_j$ in the tree ${\cal T}$, then either
   $|X_i|=|X_j|+1$ and $X_j \subset X_i$, and in this case $X_i$ is called an {\em insert node}; or $|X_i|=|X_j|-1$ and
   $X_i\subset X_j$, and in this case $X_i$ is called a {\em forget node}.

\item If $X_i$ is a leaf node or the root, then $X_i=\emptyset$.
\end{enumerate}

\paragraph{{\bf Boolean Circuits and Parameterized Complexity.}}

A {\it circuit} is a directed acyclic graph. The vertices of
indegree $0$ are called the (input) {\it variables}, and are labeled either
by {\it positive literals} $x_i$ or by {\it negative literals}
$\overline{x}_i$. The vertices of indegree larger than $0$ are
called the {\it gates} and are labeled with Boolean operators
{\sc and} or {\sc or}. A special gate of outdegree $0$ is
designated as the {\it output} gate. We do not allow {\sc not}
gates in the above circuit model, since by De Morgan's
laws, a general circuit can be effectively converted into the
above circuit model. A circuit is said to be {\it monotone}
if all its input literals are positive. The {\it depth} of a circuit is the maximum
distance from an input variable to the output gate of the circuit.
A circuit represents a Boolean function in a natural way. The size of a circuit $C$, denoted $|C|$, is the size of the underlying graph (\ie, number of vertices and edges).
An {\em occurrence} of a literal in $C$ is an edge from the literal to a gate in $C$. Therefore, the total number of occurrences of the literals in $C$ is the number of outgoing edges from the literals in $C$ to its gates.

We say that a truth assignment $\tau$ to the variables
of a circuit $C$ {\it satisfies} a gate $g$ in $C$ if $\tau$ makes the gate $g$ have
value $1$, and that $\tau$ {\it satisfies the circuit $C$} if $\tau$
satisfies the output gate of $C$. A circuit $C$ is {\it satisfiable}
if there is a truth assignment to the input variables of $C$ that
satisfies $C$. The {\it weight} of an assignment $\tau$ is the number
of variables assigned value $1$ by $\tau$.

A {\it parameterized problem} $Q$ is a subset of $\Omega^* \times
\mathbb{N}$, where $\Omega$ is a fixed alphabet. Each instance of the
parameterized problem $Q$ is a pair $(x, k)$, where $k \in \nat$ is called the {\it
parameter}. We say that the parameterized problem $Q$ is
{\it fixed-parameter tractable} (\FPT)~\cite{dfbook}, if there is a
(parameterized) algorithm, also called an {\em \FPT-algorithm},  that decides whether an input $(x, k)$
is a member of $Q$ in time $f(k) \cdot |x|^{O(1)}$, where $f$ is a computable function.  Let \FPT{} denote the class of all fixed-parameter
tractable parameterized problems.

A parameterized problem $Q$
is {\it \FPT-reducible} to a parameterized problem $Q'$ if there is
an algorithm, called an \emph{\FPT-reduction}, that transforms each instance $(x, k)$ of $Q$
into an instance $(x', k')$ of
$Q'$ in time $f(k)\cdot |x|^{O(1)}$, such that $k' \leq g(k)$ and $(x, k) \in Q$ if and
only if $(x', k') \in Q'$, where $f$ and $g$ are computable
functions. By \emph{\FPT-time} we denote time of the form $f(k)\cdot |x|^{O(1)}$, where $f$ is a computable function and $|x|$ is the input instance size.

Based on the notion of \FPT-reducibility, a hierarchy of
parameterized complexity, {\it the \W-hierarchy} $\bigcup_{t
\geq 0} \W[t]$, where $\W[t] \subseteq \W[t+1]$ for all $t \geq 0$, has
been introduced, in which the $0$-th level \W[0] is the class {\it
\FPT}. The hardness and completeness have been defined for each level
\W[$i$] of the \W-hierarchy for $i \geq 1$ \cite{dfbook}. It is
commonly believed that $\W[1] \neq \FPT$ (see \cite{dfbook}). The
\W[1]-hardness has served as the main working hypothesis of fixed-parameter
intractability.

The class \W[SAT] contains all parameterized problems that are \FPT-reducible to the weighted satisfiability of
Boolean formulas.  It contains the classes \W[t], for every $t \geq 0$. Boolean formulas can be represented (in polynomial time) by Boolean circuits that are in the {\em normalized} form (see~\cite{dfbook}).
In the normalized form every (nonvariable) gate has outdegree at most 1, and the gates are structured into alternating levels of {\sc or}s-of-{\sc and}s-of-{\sc or}s.... Therefore, the underlying undirected graph of the circuit with the input variables removed is a tree; the input variables can be connected to any gate in the circuit, including the output gate.
The class \W[\Pol] contains all parameterized problems that are \FPT-reducible to the weighted satisfiability of
Boolean circuits of polynomial size, and contains the class \W[SAT].

Let $O$ be a parameterized minimization problem, and $\rho :\nat \rightarrow \mathbb{R}_{\ge 1}$ a computable function such that $\rho(k)\ge 1$ for every $k\ge 1$. A decision algorithm $\mathbb{A}$ is an \emph{\FPT{} cost approximation algorithm} for $O$ with approximation ratio $\rho$~\cite{ChenGG06}, if for every input $(x,k)\in \Sigma^*\times \nat$, 
its output satisfies the following:
\begin{itemize}
	\item If $k\le OPT(x)$, then $\mathbb{A}$ rejects $(x,k)$, and
	\item if $k \ge \rho(OPT(x))\cdot OPT(x)$, then $\mathbb{A}$ accepts $(x,k)$.
\end{itemize}

Furthermore, $\mathbb{A}$ runs in \FPT{}-time.

The {\em Exponential Time Hypothesis} (ETH) states that the satisfiability of {\sc $k$-cnf} Boolean formulas, where $k \geq 3$, is not decidable in subexponential-time $\Oh(2^{o(n)})$, where $n$ is the number of variables in the formula. ETH has become a standard hypothesis in complexity theory for proving hardness results that is closely related to the computational intractability of a large class of well-known NP-hard problems, measured from a number of different angles,
such as subexponential-time complexity, fixed-parameter tractability, and
approximation.

The asymptotic notation $\Oh^{*}$ suppresses a polynomial factor in the input length.

\paragraph{\textbf{\textsc{\mor{}}} and \textbf{\textsc{\cmor{}}}.}
\fi

For a set $S$, we denote by $2^{S}$ the power set of $S$.
Let $G=(V, E)$ be a graph, let $C \subset \mathbb{N}$ be a finite set of colors, and let $\chi: V \longrightarrow 2^{C}$. A vertex $v$ in $V$ is \emph{empty} if
$\chi(v) = \emptyset$. A color $c$ \emph{appears on}, or is \emph{contained in}, a subset $S$ of vertices if $c \in \bigcup_{v \in S} \chi(v)$. For two vertices $u, v \in V(G)$, $\ell \in \mathbb{N}$, a $u$-$v$ path $P=(u=v_0, \ldots, v_r=v)$ in $G$ is \emph{$\ell$-valid} if $|\bigcup_{i=0}^{r} \chi(v_i)| \leq \ell$; that is, if the total number of colors appearing on the vertices of $P$ is at most $\ell$. A color $c \in C$ is {\em connected} in $G$, or simply {\em connected}, if $\bigcup_{c \in C(v)} \{v\}$ induces a connected subgraph of $G$. The graph $G$ is {\em color-connected}, if for every $c \in C$, $c$ is connected in $G$.




For an instance $(G, C, \chi, s, t, k)$ of \mor{} or \cmor{}, if $s$ and $t$ are nonempty vertices, we can remove their colors and decrement $k$ by $|\chi(s) \cup \chi(t)|$ because their colors appear on every $s$-$t$ path. If afterwards $k$ becomes negative, then there is no $k$-valid $s$-$t$ path in $G$. Moreover, if $s$ and $t$ are adjacent, then the path $(s, t)$ is a path with the minimum number of colors among all $s$-$t$ paths in $G$. Therefore, we will assume:

\begin{assumption}
\label{ass:sat}
For an instance $(G, C, \chi, s, t, k)$ of \mor{} or \cmor{}, we can assume that $s$ and $t$ are nonadjacent empty vertices.
\end{assumption}

\iflong

\begin{definition}\rm
\label{def:colorcontraction}
Let $s, t$ be two designated vertices in $G$, and let $x, y$ be two adjacent vertices in $G$ such that $\chi(x) = \chi(y)$. We define the following operation to $x$ and $y$, referred to as a \emph{color contraction} operation, that results in a  graph $G'$, a color function $\chi'$, and two designated vertices $s', t'$ in $G'$, obtained as follows:
 \begin{itemize}
 \item $G'$ is the graph obtained from $G$ by contracting the edge $xy$, which results in a new vertex $z$;
 \item $s'=s$ (resp.~$t'=t$) if $s \notin \{x, y\}$ (resp.~$t \notin \{x, y\}$), and $s'=z$ (resp.~$t'=z$) otherwise; and
 \item $\chi': V(G') \longrightarrow 2^{C}$ is the function defined as $\chi'(w) = \chi(w)$ if $w \neq z$, and $\chi'(z)=\chi(x)=\chi(y)$.
 \end{itemize}
$G$ is \emph{irreducible} if there does not exist two vertices in $G$ to which the color contraction operation is applicable.
\end{definition}

\begin{lemma}
\label{lem:contract}
Let $G$ be a color-connected plane graph, $C$ a color set, $\chi: V \longrightarrow 2^{C}$, $s, t \in V(G)$, and $k \in \mathbb{N}$. Suppose that the color contraction operation is applied to two vertices in $G$ to obtain $G'$, $\chi'$, $s', t'$, as described in Definition~\ref{def:colorcontraction}. Then $G'$ is a color-connected plane graph, and there is a $k$-valid $s$-$t$ path in $G$ if and only if there is a $k$-valid $s'$-$t'$ path in $G'$.
\end{lemma}

\begin{proof}
Let $x$ and $y$ be the two adjacent vertices in $G$ to which the color contraction operation is applied, and let $z$ be the new vertex resulting from this contraction. It is clear that after the contraction operation the obtained graph $G'$ is a plane color-connected graph.

Suppose that there is a $k$-valid $s$-$t$ path in $G$, and let $P=(s=v_0, \ldots, v_r=t)$ be such a path. We can assume that $P$ is an induced path. If no vertex in $\{x, y\}$ is on $P$, then $P'=P$ is a $k$-valid $s'$-$t'$ path in $G'$. If exactly one vertex in $\{x, y\}$, say $x$, is on $P$, then since the color set of every vertex other than $x$ on $P$ is the same before and after the contraction operation, and since $\chi'(z) = \chi(x)$, the path $P'$ obtained from $P$ by replacing $x$ with $z$ is a $k$-valid $s'$-$t'$ in $G'$. (Note that if $x=s$ then $s'=z$, and replacing $x$ with $z$ on $P$ is obsolete in this case.) Finally, if both $x$ and $y$ are on $P$, then since $P$ is induced, $x$ and $y$ must appear consecutively on $P$. Without loss of generality, assume $x=v_i$ and $y=v_{i+1}$, for some $i \in \{0, \ldots, r-1\}$. Since the color set of every vertex other than $x$ and $y$ on $P$ is the same before and after the operation, and since $\chi'(z)=\chi(x)=\chi(y)$, the path $P'=(s'=v_0, \ldots, v_{i-1}, z, v_{i+1}, \ldots, t=v_r)$ is a $k$-valid $s'$-$t'$ path in $G'$.

Conversely, suppose that there is a $k$-valid $s'$-$t'$ path in $G'$, and let $P'=(s'=v_0', \ldots, v_p'=t')$, where $p > 0$, be such a path. If $z$ does not appear on $P'$ then $P'$ is a $k$-valid $s$-$t$ path in $G$. Otherwise, $z=v_i'$ for some $i \in \{0, \ldots, p\}$. If $i=0$ and $P'$ consists only of vertex $z$, then since $\chi(x)=\chi(y)=\chi'(z)$, either $s=t$, and in which case there is a trivial $k$-valid $s$-$t$ path in $G$, or $s \neq t$, and in this case $P=(x, y)$ is a $k$-valid $s$-$t$ path in $G$. Otherwise, when $i=0$ we must have $s=x$ or $s=y$, $v'_i \in G$ for $i \in [p]$, and $t'=t$; without loss of generality, assume that $s=x$. Since $z$ is adjacent to $v_1'$, either $x=s$ or $y$ (or both) is adjacent to $v_1'$. Since $\chi(x)=\chi(y)=\chi'(z)$, if $x$ is adjacent to $v_1'$ then $P=(s=x, v_1', \ldots, v_p'=t')$ is a $k$-valid $s$-$t$ path in $G$, and if $y$ is adjacent to $v_1'$ then $P=(s=x, y, v_1', \ldots, v_p'=t')$ is a $k$-valid $s$-$t$ path in $G$. The case is similar if $i=p$. Suppose now that $i \neq 0$ and $i \neq p$. If $x$ (resp.~$y$) is adjacent to both $v'_{i-1}$ and $v'_{i-1}$, then the path $P=(s, v_1', \ldots, v'_{i-1}, x, v'_{i+1}, \ldots, v_p'=t)$ (resp.~$P=(s, v_1', \ldots, v'_{i-1}, y, v'_{i+1}, \ldots, v_p'=t)$ is a $k$-valid $s$-$t$-path in $G$; otherwise, one vertex in $\{x, y\}$, say $x$, must be adjacent to $v'_{i-1}$, and the other vertex $y$ must be adjacent to $v'_{i+1}$. In this case the path $P=(s, v_1', \ldots, v'_{i-1}, x, y, v'_{i+1}, \ldots, v_p'=t)$ is a $k$-valid $s$-$t$-path in $G$.
\end{proof}
\fi

\iflong

\section{Hardness Results}
\label{sec:hardness}

In this section, we study the complexity and the parameterized complexity of \mor{} and \cmor{} and their geometric counterparts \gmor{} and \gcmor{}. We start by showing that both problems are \NP-hard, even when restricted to graphs of small outerplanarity and pathwidth.
\begin{figure}[htbp]
 \begin{subfigure}{.6\textwidth}
\resizebox{.7\textwidth}{1.4in}{\includegraphics[width=7.5cm]{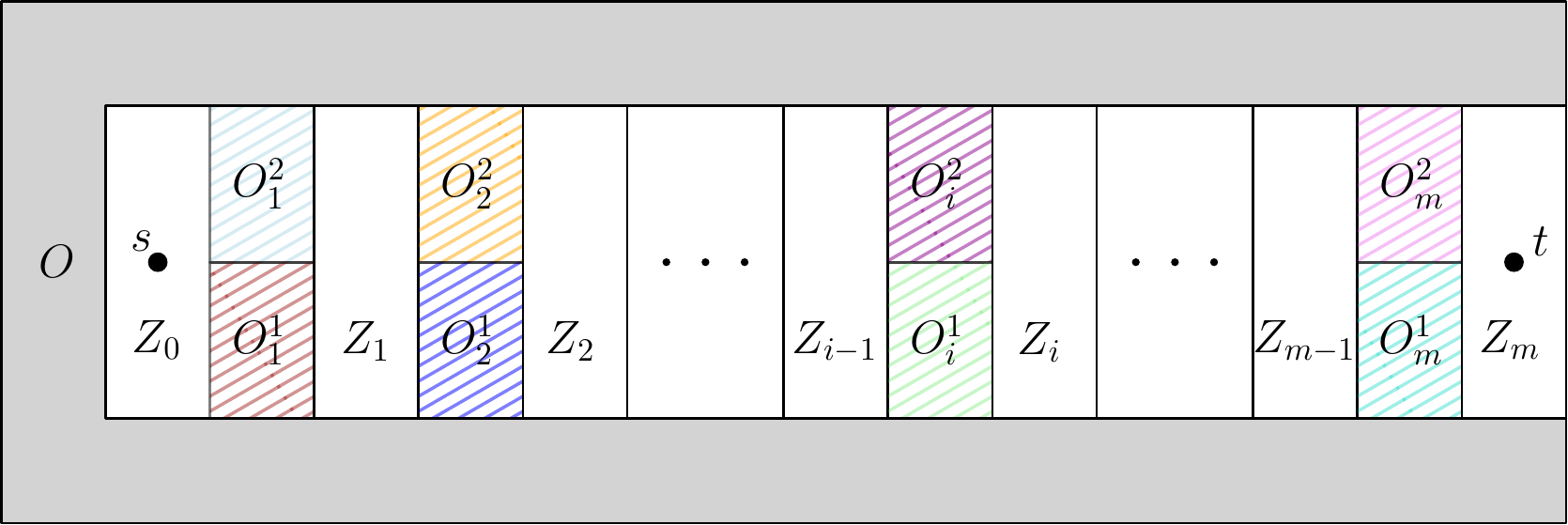}}
\end{subfigure}
\hspace*{-0.5cm}
\begin{subfigure}{.6\textwidth}
\resizebox{.7\textwidth}{1.4in}{\includegraphics[width=7.5cm]{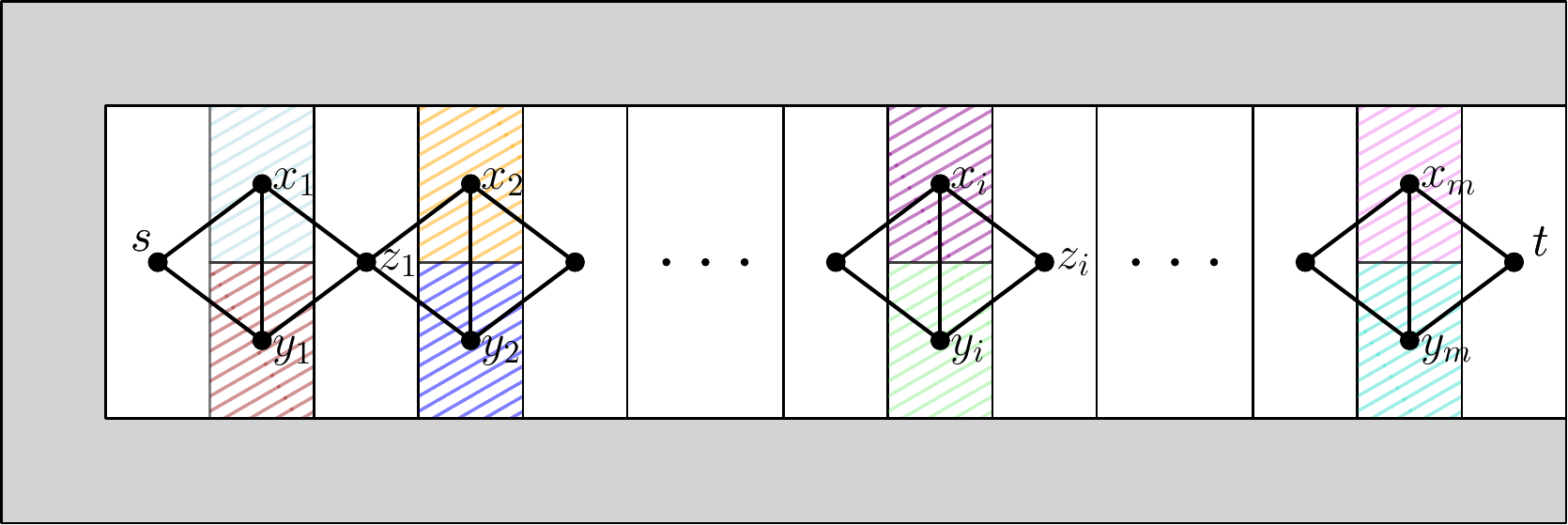}}
\end{subfigure}
\caption{Illustration of the construction in the proof of Theorem~\ref{thm:pathwidthnphardness}. The left figure shows the geometric instance of \mor{}, and the right figure the graph associated with it.}
\label{fig:npreduction}
\end{figure}

\begin{theorem}
\label{thm:pathwidthnphardness}
\mor{}, restricted to outerplanar graphs of pathwidth at most 2 and in which every vertex contains at most one color, is~\NP-complete.
\end{theorem}

\begin{proof}
It is clear that \mor{} is in \NP. To show its~\NP-hardness, we reduce from the \NP-hard problem {\sc Vertex Cover}~\cite{gj}. Let $(G, k)$ be an instance of {\sc Vertex Cover}, where $V(G) = \{v_1, \ldots, v_n\}$ and $E(G) = \{e_1, \ldots, e_m\}$. In the rest of the proof, when we write $e=uw$ for an edge $e$ in $E(G)$, we assume that $u=v_i$ and $w=v_j$ such that $i < j$ (\ie, the vertex of smaller index always appears first). Although not necessary for the proof, we first describe an instance $I$ of \gmor{} whose associated graph is the desired instance of \mor{}. The regions of $I$ are $O \cup \{Z_0, \ldots, Z_m\} \cup \bigcup_{i=1}^{m} \{O_{i}^{1}, O_{i}^{2}\}$, depicted in Figure~\ref{fig:npreduction} (left figure). The obstacles of $I$ are defined as follows. For each vertex $v_j \in V(G)$, the obstacle corresponding to $v_j$ is the polygon whose boundary is the boundary of the region formed by the union of $O$, each $O_{i}^{1}$ such that $e_i =v_jv_q$, and each $O_{i}^{2}$ such $e_i =v_pv_j$. More formally, the obstacle corresponding to $v_j$ is $\partial(O \cup \bigcup_{e_i=v_jv_q} O_{i}^{1} \cup \bigcup_{e_i=v_pv_j} O_{i}^{2})$. The graph associated with $I$, $G_I$, is defined as follow. Each (empty) region $Z_i$, $i=0, \ldots, m$, corresponds to a vertex $z_i \in V(G_I)$, where $Z_0$ corresponds to $s$ and $Z_m$ to $t$. Each region $O_{i}^{1}$, $i \in [m]$, corresponds to a vertex $y_i$, and each region $O_{i}^{2}$, $i \in [m]$, corresponds to a vertex $x_i$. The set of edges $E(G_I)$ is $E(G_I)= \{z_{i-1}x_i, z_{i-1}y_i, x_iy_i, z_{i}x_i, z_{i}y_i \mid i \in [m]\}$. The color function $\chi: V(G_I) \longrightarrow 2^{C}$, where $C=[n]$, is defined as follows: $\chi(z_i)= \emptyset$, for $i=0, \ldots, m$; $\chi(x_i) = \{j\}$ and $\chi(y_i) = \{p\}$, where $e_i=v_pv_j$, for $i \in [m]$. This completes the construction of $G_I$; see Figure~\ref{fig:npreduction} (right figure) for illustration. It is easy to see that $G_I$ is outerplanar and has pathwidth at most 2.

Define the reduction from {\sc Vertex Cover} to \mor{} that takes an instance $(G, k)$ to the instance $(G_I, C, \chi, s, t, k)$. Clearly, this reduction is polynomial-time computable. Suppose that $Q$, where $|Q|=r \leq k$, is a vertex cover of $G$. Consider the $s$-$t$ path $P=(s, w_1, z_1, , \ldots, w_m, z_m)$ in $G_I$, where $w_i =y_i$ if edge $e_i=v_pv_q$ is covered by $v_p$, and $w_i=x_i$ otherwise, for $i \in [m]$. Clearly this is a $k$-valid $s$-$t$ path in $G_I$ since each edge $e_i$ is covered by a vertex in $Q$, each $w_i$ is colored by the index of one of the vertices in $Q$, and each vertex in $G_I$ (and hence each $w_i$) contains at most one color. Conversely, suppose that $P$ is a $k$-valid $s$-$t$ path in $G_I$. By construction of $G_I$, $P$ has to contain at least one vertex from $\{x_i, y_i\}$, for each $i \in [m]$. If $P$ contains both $x_i$ and $y_i$, for some $i \in [m]$, then clearly, from the construction of $P$, $P$ must contain either $(z_{i-1}, x_i, y_i, z_i)$ or $(z_{i-1}, y_i, x_i, z_i)$, as a subpath, and we can shortcut this subpath by removing one of $x_i, y_i$, to obtain another $k$-valid $s$-$t$ path in $G_I$. Therefore, without loss of generality, we may assume that $P$ contains exactly one vertex $w_i$ from $\{x_i, y_i\}$, for $i \in [m]$.
Now define the set of vertices $Q$ in $G$ as the vertices in $G$ whose indices are the colors appearing on (the $w_i$'s in) $P$. More formally, define $Q=\{v_p \mid w_i=x_i \in P \wedge e_i=v_qv_p \} \cup \{v_p \mid w_i=y_i \in P \wedge e_i=v_pv_q\}$. Since $P$ is a $k$-valid path in $G_I$, the total number of colors appearing on $\{w_1, \ldots, w_m\}$ is at most $k$. Notice that the color of each of $x_i, y_i$ is the index of a vertex in $G$ that covers edge $e_i$. It follows that the set $Q$ of vertices in $G$, that are the indices of the colors on $P$, form a $k$-vertex cover of $G$.
\end{proof}

\begin{figure}[htbp]
\begin{center}
\includegraphics[width=12cm]{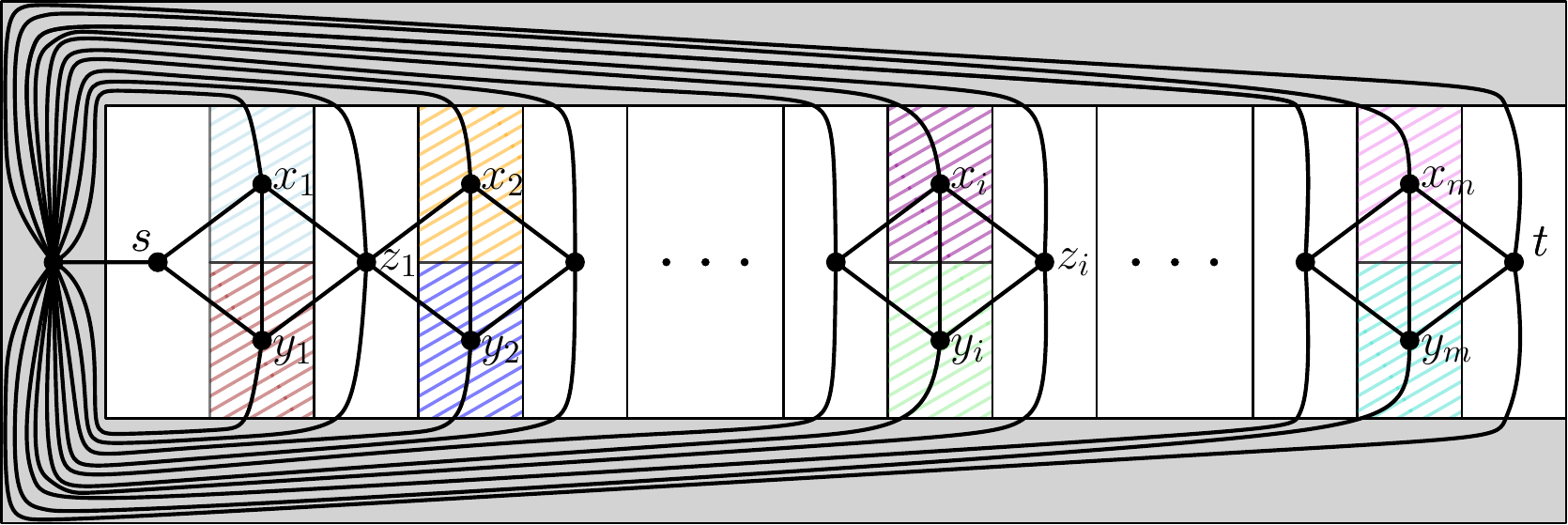}
\end{center}
\caption{Illustration of the proof of Corollary~\ref{cor:connectedpathwidthnphardness}.}
\label{fig:npreductionconnected}
\end{figure}

\begin{corollary}
\label{cor:connectedpathwidthnphardness}
\cmor{}, restricted to 2-outerplanar graphs of pathwidth at most 3, is \NP-complete.
\end{corollary}

\begin{proof}
This follows directly from the \NP-hardness reduction in the proof of Theorem~\ref{thm:pathwidthnphardness} by observing the following. The graph $G_I$ resulting from the reduction is outerplanar. We can add a new vertex to the outer face of $G_I$ (see Figure~\ref{fig:npreductionconnected}\ifshort~in Section~\ref{sec:figures}\fi) containing all colors that appear on $G_I$, and add edges between the new vertex and all vertices in $G_I$. The obtained graph is color-connected and has pathwidth at most 3.
\end{proof}

Assuming ETH, the following corollary rules out the existence of subexponential-time algorithms for \cmor{} (and hence for \mor{}), even for restrictions of the problem to graphs of small outerplanarity, pathwidth, and maximum number of occurrences of each color:

\begin{corollary}
\label{cor:pathwidthnphardness}
Unless ETH fails, \cmor{}, restricted to 2-outerplanar graphs of pathwidth at most 3 and in which each color appears at most 4 times, is not solvable in $\Oh(2^{o(n)})$ time, where $n$ is the number of vertices in the graph.
\end{corollary}

\begin{proof}
It is well known, and follows from~\cite{js} and the standard reduction from {\sc Independent Set} to {\sc Vertex Cover}, that unless ETH fails, {\sc Vertex Cover}, restricted to graphs of maximum degree at most 3, denoted VC-3, is not solvable in subexponential time. Starting from an instance of VC-3 with $n$ vertices, and observing that the reduction in the proof of Theorem~\ref{thm:pathwidthnphardness} results in an instance of \cmor{} whose number of vertices is $O(n)$, of pathwidth at most 3, and in which each color appears at most 4 times, proves the result.
\end{proof}

\begin{figure}[htbp]
\begin{center}
\includegraphics[width=8.5cm]{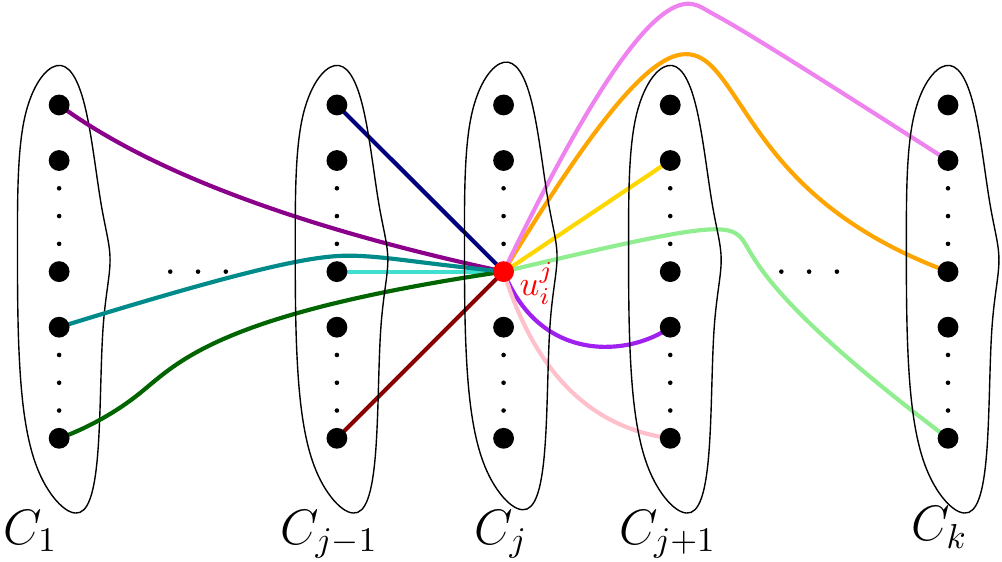}
\end{center}
\begin{center}
\includegraphics[width=8.5cm]{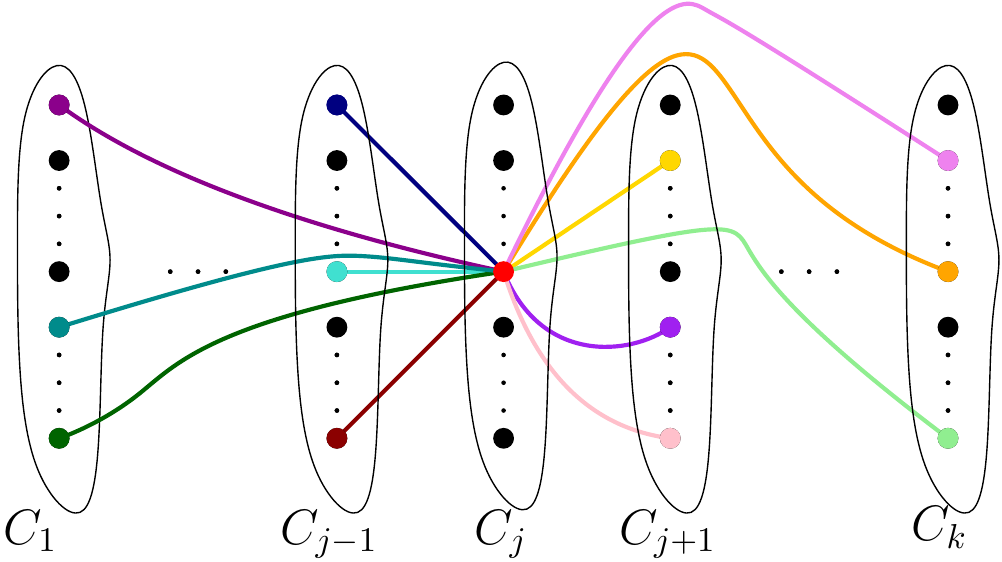}
\end{center}
\begin{center}
\includegraphics[width=8.5cm]{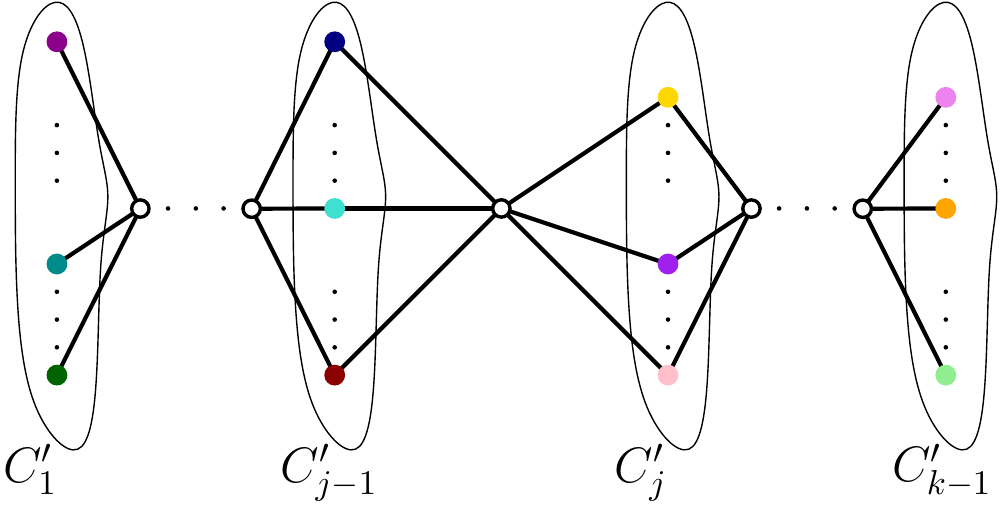}
\end{center}
\caption{Illustration of the construction of the gadget $G_{i, j}$ in the proof of Lemma~\ref{lem:w1hard}.}
\label{fig:w1hard_gadget}
\end{figure}

\begin{figure}[htbp]
\begin{center}
\includegraphics[width=13cm]{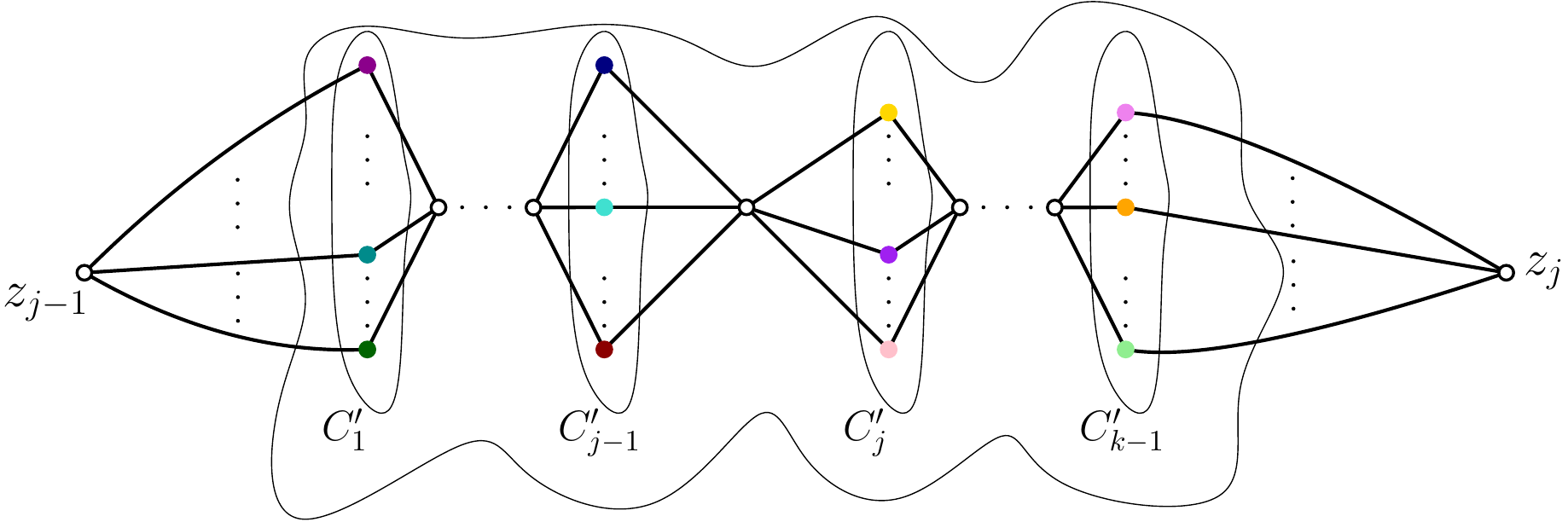}
\end{center}
\begin{center}
\includegraphics[width=13cm]{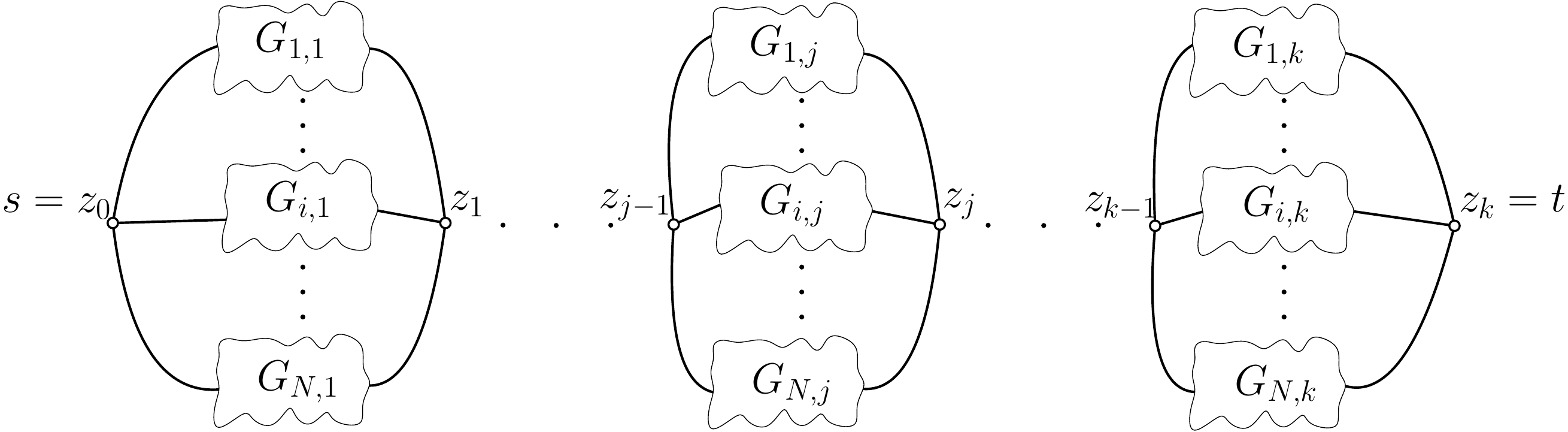}
\end{center}
\caption{Illustration of the construction of $G'$ in the proof of Lemma~\ref{lem:w1hard}.}
\label{fig:w1_hardness_whole}
\end{figure}

Next, we shift our attention to studying the parameterized complexity of \mor{} and \cmor{}. The reduction from {\sc Set Cover} showing the \NP-hardness of \mor{}, given in several works~\cite{carr,hauser,varma}, is in fact an \FPT-reduction implying the \W[2]-hardness of \mor{}. We will strengthen this result, and show in the remainder of this section that \mor{} is \W[SAT]-hard. We will also prove the membership of the problem in \W[\Pol], which adds a natural \W[SAT]-hard problem to this class. The \W[SAT]-hardness result shows that the problem is hopeless in terms of it having \FPT-algorithms. We start by showing that the problem remains \W[1]-hard, even when restricted to instances of small pathwidth (and hence small treewidth) and maximum number of occurrences of each color. We then show that the problem remains \W[1]-hard even when parameterized by both $k$ and the length of the sought path.

\begin{rem}
\label{rem:gadget}
Before we prove our parameterized hardness results for \mor{}, we remark that we can obtain equivalent hardness results for \gmor{} using the following generic realization of instances of \mor{} as instances of \gmor{}. Given an instance $(G,C,\chi,s,t,k)$ of \mor{}, we define an equivalent instance of \gmor{} as follows. We start by fixing a straight-line plane embedding $\Pi$ of $G$, which always exists by F\'ary's theorem~\cite{Fary48}. Moreover, we can compute such an embedding in linear time~\cite{deFraysseixHubertPachPollack88}. We define the starting and finishing positions of the path as the images of vertices $s$ and $t$ under $\Pi$, respectively. To force the path to go along the edges of $G$, we correspond with every edge a ``corridor'' by putting on both sides of the image of the edge $k+1$ trapezoids as shown in Figure~\ref{fig:graphToGeom}. The only possible way to move along the vertices of the graph $G$ without intersecting more than $k$ obstacles is to move within these corridors. Finally, for each color $c\in C$ and every vertex $v\in V(G)$ such that $c\in \chi(v)$, we create a rectangle around the image of the vertex $v$ under $\Pi$ that intersects all the trapezoids corresponding to the edges incident to $v$. We define the obstacle corresponding to the color $c$ in the geometric instance to be the union of these rectangles. This disallows the use of less than $k+1$ trapezoid obstacles to go through a vertex $v$ of $G$ without intersecting all the obstacles representing the color set $\chi(v)$.
Note that the only thing that is affected by this geometric realization is the number of obstacles that overlap at a region, which corresponds to the number of colors on the vertex  in the graph that corresponds to the region; this number might increase by at most 2.
\end{rem}

\begin{figure}[htbp]
	\begin{center}
		\includegraphics[width=13cm]{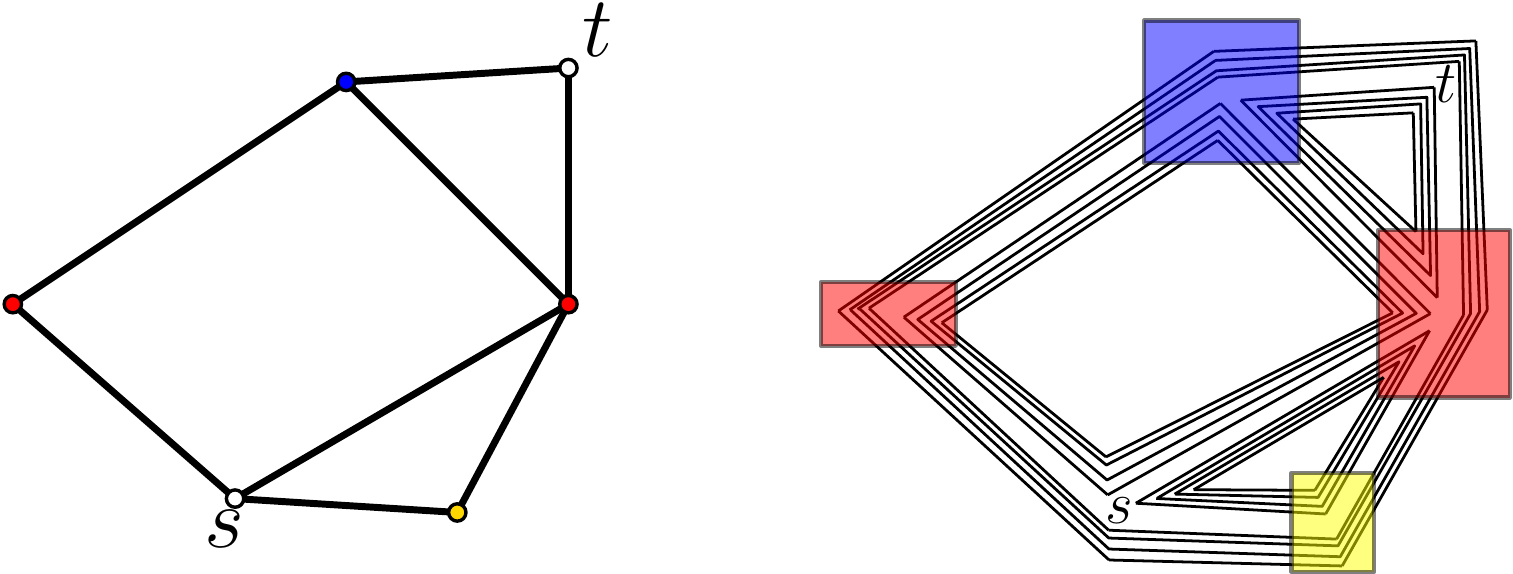}
	\end{center}
\caption{Illustration of the realization of an instance of \mor{} as an instance of \gmor{}.}
	\label{fig:graphToGeom}
\end{figure}

\begin{lemma}
\label{lem:w1hard}
\mor{}, restricted to instances of pathwidth at most 4 and in which each vertex contains at most one color and each color appears on at most 2 vertices, is \W\rm{[}1\rm{]}-hard parameterized by $k$.
\end{lemma}

\begin{proof}
We reduce from the \W[1]-hard problem \textsc{Multi-Colored Clique}~\cite{Hartung2013}.
Let $(G, k)$ be an instance of \textsc{Multi-Colored Clique}, where $V(G)$ is partitioned into the color classes $C_1, \ldots, C_k$. Let $C_j = \{u_{i}^{j} \mid i \in [|C_j|]\}$. We describe how to construct an instance $(G', C', \chi', s, t, k')$ of \mor{}. For an edge $e \in G$, associate a distinct color $c_e$, and define $C' = \{c_e \mid e \in E(G)\}$.
To simplify the description of the construction, we start by defining a gadget that will serve as a building block for this construction.

For a vertex $u_{i}^{j}$ in color class $C_j$, we define the gadget $G_{i,j}$ as follows. Create a copy of each color class $C_{j'}$, $j' \neq j$, and remove from each $C_{j'}$ all copies of vertices that are not neighbors of $u_{i}^{j}$ in $G$. Let the resulting copies of the color classes be $C'_1, \ldots, C'_{k-1}$. We define the color of a copy $v'$ of a neighbor $v$ of $u_{i}^{j}$ as $\chi'(v') = \{c_e\}$, where $e=u_{i}^{j}v$. Next, we introduce $k-2$ empty vertices $y_{r}$, $r \in [k-2]$. For $r \in [k-2]$, we connect all vertices in $C'_r$ to $y_r$, and connect $y_r$ to all vertices in $C'_{r+1}$. This completes the construction of gadget $G_{i,j}$; we refer to $C'_1$ and $C'_{k-1}$ as the first and last color classes in gadget $G_{i,j}$, respectively. See Figure~\ref{fig:w1hard_gadget} \ifshort in Section~\ref{sec:figures} \fi for illustration of $G_{i,j}$. Observe that every path from a vertex in $C'_1$ to a vertex in $C'_{k-1}$ contains exactly one vertex from each $C'_r$, $r \in [k-1]$, and contains all vertices $y_{r}$, $r \in [k-2]$. Therefore, any such path contains the colors of exactly $k-1$ distinct edges that are incident to $u_{i}^{j}$.

We finish the construction of $G'$ by introducing $k+1$ new empty vertices $z_0, \ldots, z_{k}$, and connecting them as follows. For each color class $C_j$, $j \in [k]$, and each vertex $u_{i}^{j} \in C_j$, we create the gadget $G_{i,j}$, connect $z_{j-1}$ to each vertex in the first color class of $G_{i,j}$, and connect each vertex in the last color class of $G_{i,j}$ to $z_j$. Let $G'$ be the resulting graph. Finally, we set $s=z_0$, $t=z_k$, and $k'=\binom{k}{2}$. See Figure~\ref{fig:w1_hardness_whole} \ifshort in Section~\ref{sec:figures} \fi for illustration. This completes the construction of the instance $(G', C', \chi', s, t, k')$ of \mor{}. Observe that each vertex in $G'$ contains at most one color, and that each color $c_e$ of an edge $e=u_{i}^{j}u_{i'}^{j'}$ in $G$, appears on exactly two vertices in $G'$: the copy of $u_{i'}^{j'}$ in the gadget $G_{i,j}$ of $u_{i}^{j}$, and the copy of $u_{i}^{j}$ in the gadget $G_{i',j'}$ of $u_{i'}^{j'}$.

Clearly, the reduction that takes an instance $(G, k)$ of \textsc{Multi-Colored Clique} and produces the instance $(G', C', \chi', s, t, k')$ of \mor{} is computable in \FPT-time. To show its correctness, suppose that
$(G, k)$ is a yes-instance of \textsc{Multi-Colored Clique}, and let $Q$ be a $k$-clique in $G$. Then $Q$ contains a vertex from each $C_j$, for $j \in [k]$. For a vertex $u_{i}^{j} \in Q$, let $G_{i,j}$ be its gadget, and define the path $P_j$ as follows. In each color class in $G_{i,j}$, pick the unique vertex that is a copy of a neighbor of $u_{i}^{j}$ in $Q$; define $P_j$ to be the path in $G_{i,j}$ induced by the picked vertices, plus the empty vertices $y_{r}$, $r \in [k-2]$, that appear in $G_{i,j}$. Finally, define $P$ to be the $s$-$t$ path in $G'$ whose edges are: the (unique) edge between $z_{r-1}$ and an endpoint of $P_r$, $P_r$, and the (unique) edge between an endpoint of $P_r$ and $z_r$, for $r \in [k]$. To show that $P$ is $k'$-valid, observe that all the nonempty vertices in $P$ are vertices whose color is the color of an edge between two vertices in $Q$. This shows that the number of colors that appear on $P$ is at most $k'=\binom{k}{2}$, and hence, $P$ is $k'$-valid. It follows that $(G', C', \chi', s, t, k')$ is a yes-instance of \mor{}.

Conversely, suppose that $P$ is a $k'$-valid $s$-$t$ path in $G'$. Then $P'$ must start at $s$, visit the gadgets of exactly $k$ vertices $u_{i_j}^{j} \in C_j$, for $j \in [k], i_j \in [|C_j|]$, and end at $t$. We claim that
$Q= \{u_{i_j}^{j} \mid j \in [k]\}$ is a clique in $G$. Recall that the subpath of $P$ that traverses a gadget $G_{i, j}$ of $u_{i_j}^{j}$ contains the colors of exactly $k-1$ edges that are incident to $u_{i_j}^{j}$.
Therefore, the total number of occurrences of colors (counting multiplicities) on $P$ is precisely $(k-1)k$. Since $P$ is $\binom{k}{2}$-valid, and each color $c_e$ of an edge $e$ in $G$ appears exactly twice in $G'$, it follows that each color that appears on $P$ appears exactly twice on $P$. This is only possible if the gadgets corresponding to the two endpoints of the edge are traversed by $P$, and hence, both endpoints of the edge are in $Q$. Therefore, $P$ contains the colors of $k'=\binom{k}{2}$ edges, whose both endpoints are in $Q$. Since $|Q|=k$, it follows that $Q$ is a $k$-clique in $G$, and that $(G, k)$ is a yes-instance of \textsc{Multi-Colored Clique}.
\end{proof}

\begin{lemma}
\label{lem:inw1}
\mor{}, parameterized by both $k$ and the length of the path $\ell$, is in \W\rm{[}1\rm{]}.
\end{lemma}

\begin{proof}
To prove membership in \W\rm{[}1{]}, we use the characterization of the class \W[1] given by Chen \etal~\cite{yijia}:

\begin{quote}
A parameterized problem $Q$ is in \W[1] if and only if there is a computable function $h$ and a nondeterministic \FPT{} algorithm $\mathbb{A}$ for a nondeterministic-RAM machine deciding $Q$, such that, for each instance $(x, k')$ of $Q$ ($k'$ is the parameter), all nondeterministic steps of $\mathbb{A}$ take place during the last $h(k')$ steps of the computation.
\end{quote}

Therefore, to show that \mor{}~is in \W[1], it suffices to exhibit such a nondeterministic \FPT{} algorithm $\mathbb{A}$. $\mathbb{A}$ works as follows: It guesses a set $C'$ of $k$ colors and guesses a sequence of $\ell-1$ internal vertices $v_1, \ldots, v_{\ell-1}$ of the path. Then it verifies that $(s=v_0, v_1, \ldots, v_{\ell-1}, v_{\ell}=t)$ is a path in $G$, and that $\chi(v_i) \subseteq C'$, for $i =0, \ldots, \ell$. It is not difficult to see that this verification can be implemented in $h(k, \ell)$ steps, where $h$ is a computable function.
\end{proof}

By Lemma~\ref{lem:contract}, we can assume that in an instance of \mor{}, no two adjacent vertices are empty. With this assumption in mind, if the instance satisfies that each vertex contains at most one color and that each color appears on at most 2 vertices, then any $k$-valid $s$-$t$ path has length at most $4k+1$. It follows from Lemma~\ref{lem:w1hard} and Lemma~\ref{lem:inw1} that:
\begin{theorem}
	\label{thm:wcompletelength}
	\mor{}, parameterized by both $k$ and the length of the path $\ell$, is \W\rm{[}1\rm{]}-complete.
\end{theorem}

\begin{theorem}
	\label{thm:wcompleteoccurences}
	\mor{}, restricted to instances of pathwidth at most 4 and in which each vertex contains at most one color and each color appears on at most 2 vertices, is \W\rm{[}1\rm{]}-complete parameterized by $k$.
\end{theorem}

Next, we show that \mor{} sits high up in the parameterized complexity hierarchy. We start by showing its membership in \W\rm{[}\Pol\rm{]}:

\begin{theorem}
\label{thm:inwp}
\mor{}, parameterized by $k$, is in \W\rm{[}\Pol\rm{]}.
\end{theorem}

\begin{proof}
We give an \FPT-reduction from \mor{} to {\sc Weighted Boolean Circuit Satisfiability} (WBCS) on polynomial size (monotone) circuits. Given an instance $(G, C, \chi, s, t, k)$ of \mor{}, we construct an instance
$(B, k)$ of WBCS, where $B$ is a circuit whose output gate is an {\sc or}-gate, as follows. By Assumption~\ref{ass:sat}, we can assume that $s$ and $t$ are nonadjacent empty vertices. By Lemma~\ref{lem:contract}, we can also assume that no two adjacent vertices are empty. For each color $c \in C$, we create a variable $x_c$; those are the input variables to $B$. In addition to the output gate, $B$ contains $n=|V(G)|$ layers of gates, where each layer, except the first, consists of two rows of gates, $U_i, L_i$, for $i =2, \ldots, n$, and the first layer consists of one row $L_1$ of gates. The layers of $B$ are defined as follows.

Each gate in $L_1$ is an {\sc and}-gate $g_v$ that corresponds to a neighbor $v$ of $s$; the input to $g_v$ is the set of input variables corresponding to the colors in $\chi(v)$. Suppose that row $L_i$ in layer $i$, $i \geq 1$, has been defined, and we describe how $U_{i+1}$ and $L_{i+1}$ are defined. For every vertex $v \in V(G)$ with a neighbor $u$ such that $u$ has a corresponding \AND{}-gate $g_{u}^{2}$ in $L_i$, we create an \OR{}-gate $g_{v}^{1}$ in $U_{i+1}$ and an \AND{}-gate $g_{v}^{2}$ in $L_{i+1}$ corresponding to $v$; we connect the output of each \AND{}-gate $g_{u}^{2}$ in $L_i$ corresponding to neighbor $u$ of $v$ to the input of  \OR{}-gate $g_{v}^{1}$ in $U_{i+1}$, and connect the output of the \OR{}-gate $g_{v}^{1}$ and each input variable $x_c$ such that $c \in \chi(v)$ to the \AND-gate{} $g_{v}^{2}$ in $L_{i+1}$. If $v=t$, then we connect the output of the \AND-gate{} $g_{v}^{2}$ to the output gate of the circuit. This completes the description of $B$. Clearly, the reduction that takes $(G, C, \chi, s, t, k)$ to $(B, k)$ runs in polynomial time, and hence in \FPT-time. Next, we prove its correctness.

First observe that the only gates in $B$ that are connected to its output gate are the \AND-gates{} that correspond to $t$. Second, every gate in $B$ corresponds to a vertex that is reachable from $s$ in $G$. Moreover, for every \AND-gate{} $g$ corresponding to a vertex $v$, and every $s$-$v$ path in $G$, the truth assignment that assigns 1 to the variables corresponding to the colors of this path satisfies $g$.

Suppose now that $(G, C, \chi, s, t, k)$ is a yes-instance of \mor{}. Then there is an $s$-$t$ $k$-valid path $P$ in $G$. Based on the above observations, the assignment that assigns $x_c=1$ if and only if $c \in \chi(P)$ is a satisfying assignment to $B$ of weight at most $k$. Conversely, suppose that $B$ has a satisfying assignment $\tau$ of weight at most $k$. Then there is an \AND-gate{} $g$ corresponding to $t$ that is satisfied by $\tau$, and there is a path in $B$ from a gate corresponding to neighbor of $s$ in $L_1$ to $g$, all of whose gates are satisfied by $\tau$. It is easy to verify that this path in $B$ corresponds to an $s$-$t$ path all of whose colors correspond to the input variables assigned 1 by $\tau$, and hence this path is $k$-valid.
\end{proof}

\begin{figure}[htbp]
\begin{center}
\includegraphics[width=6cm]{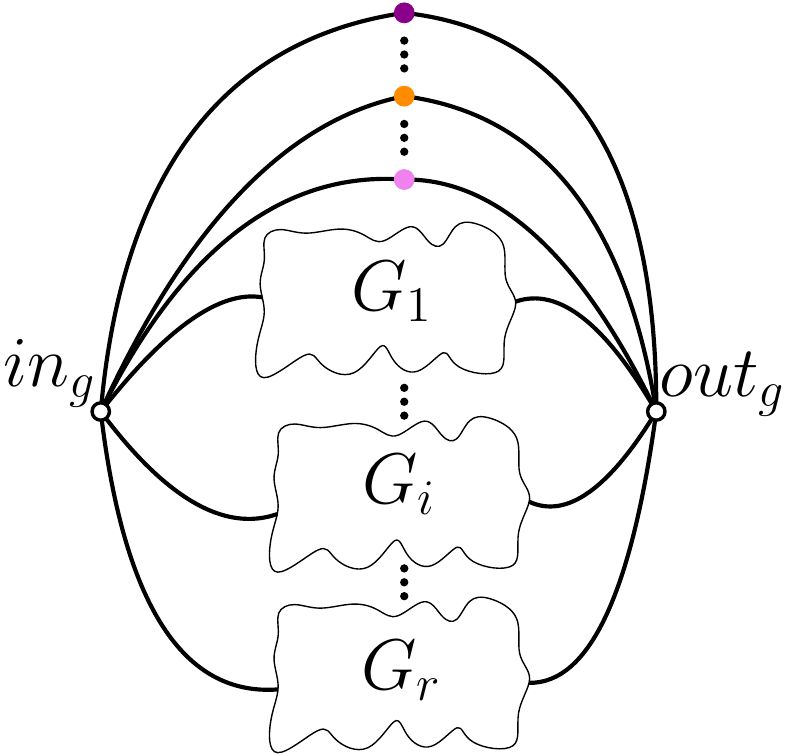}
\end{center}
\begin{center}
\includegraphics[width=16cm]{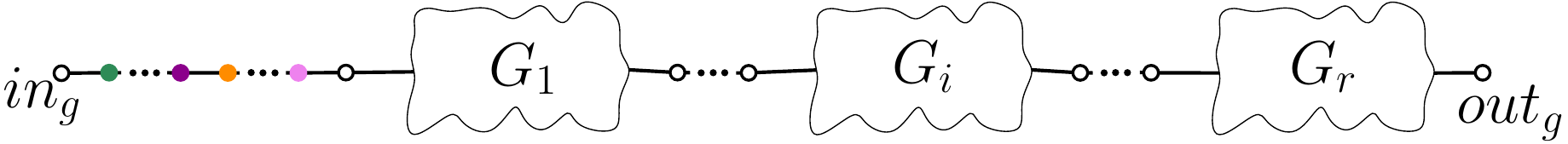}
\end{center}
\caption{Illustrations of the construction of the gadgets for an \OR-gate{} (top) and an \AND-gate{} (bottom) in the proof of Theorem~\ref{thm:wsathard}.}
\label{fig:wsat}
\end{figure}

\begin{theorem}
\label{thm:wsathard}
\mor{}, parameterized by $k$, is \W\rm{[SAT]}-hard.
\end{theorem}

\begin{proof}
We give an \FPT-reduction from the \W[SAT]-complete problem {\sc Monotone Weighted Boolean Formulas Satisfiability} (M-WSAT)~\cite{dfbook}.

Recall that a Boolean formula corresponds to a circuit in the normalized form. Therefore, we can assume that the input instance of M-WSAT is $(B, k)$, where $B$ is a monotone Boolean circuit in which each (non-variable) gate has fan-out at most 1, and the gates of $B$ are structured into alternating levels of {\sc or}s-of-{\sc and}s-of-{\sc or}s.
We construct an instance $(G, C, \chi, s, t, k)$ of \mor{} as follows.

First, we let $C=[n]$, where color $i$ will represent input variable $x_i$ in $B$. We define $G$ from $B$ by defining a gadget for each gate in $B$ recursively, starting the recursive definition at the output gate of $B$.
For a gate $g$ in $B$, its gadget is defined by distinguishing the type of $g$ as follows.

If $g$ is an \AND-gate{}, let $g_1, \ldots, g_r$ be the \OR-gates{}, and $x_{i_1}, \ldots, x_{i_p}$ be the input variables that feed into $g$. The gadget of $g$ is defined as follows. First, create two empty vertices
$in_g$ and $out_g$, which will serve as the ``entry'' and ``exit'' vertices of the gadget for $g$, respectively. For each $x_{i_j}$, $j \in [p]$, create a vertex $v_j$ colored with color $i_j$ and an entry vertex $v_0$ and an exit vertex $v_{p+1}$; form a path $G_0$ consisting of the vertices $v_0, v_1, \ldots, v_p, v_{p+1}$. For each \OR-gate{} $g_i$, $i \in [r]$, recursively construct the gadget $G_i$ for $g_i$. Connect all these gadgets $G_0, \ldots, G_r$ serially in arbitrary order, starting by identifying $in_g$ with the entry vertex of the first gadget, the exit vertex of the first gadget with the entry of the second, ..., and the exit vertex of the last gadget with $out_g$. See Figure~\ref{fig:wsat} (bottom) \ifshort in Section~\ref{sec:figures} \fi for illustration.

If $g$ is an \OR-gate{}, let $g_1, \ldots, g_r$ be the \AND-gates{}, and $x_{i_1}, \ldots, x_{i_p}$ be the input variables that feed into $g$. The gadget of $g$ is defined as follows. First, create two empty vertices
$in_g$ and $out_g$, which will serve as the ``entry'' and ``exit'' vertices of the gadget for $g$, respectively. For each $x_{i_j}$, $j \in [p]$, create a vertex $v_j$ colored with color $i_j$, and connect each $v_j$ to $in_g$ and $out_g$. For each \AND-gate{} $g_i$, $i \in [r]$, recursively construct the gadget $G_i$ for $g_i$. Connect all these gadgets $G_1, \ldots, G_r$ in parallel by identifying all the entry vertices of $G_1, \ldots, G_r$ with $in_g$ and all their exit vertices with $out_g$. This complete the description of $G$. It is not difficult to see that since $B$ with its input variables removed is a tree, the above construction runs in polynomial time and results in a planar graph $G$. See Figure~\ref{fig:wsat} (top) \ifshort in Section~\ref{sec:figures} \fi for illustration.

Finally, set $s$ and $t$ to be the entry and exit vertices of the gadget corresponding to the output gate of $B$. Clearly, the reduction that takes $(B, k)$ and produces $(G, C, \chi, s, t, k)$ runs in \FPT-time. Next, we prove its correctness.

We will prove the following statement: For any gate $g$ in $B$, and any assignment $\tau$ to $B$ that assigns variables $x_{i_1}, \ldots, x_{i_p}$ the value 1, and all other variables the value 0, $\tau$ satisfies $g$ if and only if there is a path $P$ in $G$ from the entry vertex to the exit vertex of the gadget corresponding to $g$ such that $P$ uses a subset of the colors $\{i_1, \ldots, i_p\}$. Clearly, proving the aforementioned statement implies that there is a $k$-valid $s$-$t$ path in $G$ if and only if there is an assignment of weight at most $k$ that satisfies the output gate of $B$, and hence satisfies $B$.

We prove the above statement by induction on the depth of the gate $g$ in $B$. The base case is when $g$ has depth 1. In this case the input to $g$ consists only of input variables. Suppose first that $g$ is an \OR-gate{}, and let $\tau$ be an assignment that assigns exactly variables $x_{i_1}, \ldots, x_{i_p}$ the value 1. Then $\tau$ satisfies $g$ if and only if $x_{i_j}$ is an input variable to $g$, for some $j \in [p]$, which is true if and only if there is a path from the entry vertex of the gadget for $g$ to its exit vertex that uses color $i_j$. Suppose now that $g$ is an \AND-gate{}, and let $\tau$ be an assignment that assigns exactly variables $x_{i_1}, \ldots, x_{i_p}$ the value 1. Then $\tau$ satisfies $g$ if and only if the input variables to $g$ form a subset $S$ of $\{x_{i_1}, \ldots, x_{i_p}\}$; let $\eta(S)$ be the indices of the variables in $S$. Since the gadget for $g$ consists of a path $P$ between the entry and exit vertices of the gadget for $g$ such that $\chi(P) = \eta(S)$, the statement follows.

Suppose, by the inductive hypothesis, that the statement we are proving is true for any gate $g$ of depth $1 \leq a < \ell$, and let $g$ be a gate of depth $\ell$. Let $x_{j_1}, \ldots, x_{j_q}$ be the input variables to $g$, and $g_1, \ldots, g_r$ be the input gates to $g$. We again distinguish two cases based on the type of $g$. \\

\noindent {\bf Gate $g$ is an \OR-gate{}.} Let $\tau$ be an assignment that assigns exactly variables $x_{i_1}, \ldots, x_{i_p}$ the value 1. Suppose first that $\tau$ satisfies $g$. Then either $\tau$ satisfies an input variable $x_{j_z}$, $z \in [q]$, or $\tau$ satisfies an input \AND-gate{} $g_y$, $y \in [r]$. If $\tau$ satisfies $x_{j_z}$ then there is a path between the entry and exit vertices of the gadget for $g$ that uses color $j_z$. Otherwise, $\tau$ satisfies $g_y$, $y \in [r]$, and by the inductive hypothesis applied to $g_y$, there is a path $P_y$ between the entry and exit vertices of the gadget for $g_y$ such that $\chi(P_y) \subseteq \{i_1, \ldots, i_p\}$. From the way the gadget for $g$ was constructed, it follows that $P_y$ is also a path between the entry and exit vertices of the gadget for $g$. To prove the converse, suppose that there is a path $P_g$ between the entry and exit vertices of the gadget for $g$ that uses a subset of colors in $\{i_1, \ldots, i_p\}$. Either $P_g$ is a path whose only internal vertex corresponds to an input variable, and in such case the input variable is in $\{x_{i_1}, \ldots, x_{i_p}\}$, and $g$ is satisfied; or $P_g$ is a path between the entry and exit vertices of the gadget for an \AND-gate{} $g_y$ that feeds into $g$, and by the inductive hypothesis, $\tau$ satisfies $g_y$ and also $g$. \\

\noindent {\bf Gate $g$ is an \AND-gate{}.} Let $\tau$ be an assignment that assigns exactly variables $x_{i_1}, \ldots, x_{i_p}$ the value 1. Suppose first that $\tau$ satisfies $g$. Then $\tau$ assigns 1 to every input variable $x_{j_z}$ to $g$, $z \in [q]$. Hence, there is a path $P$ between the entry and exit vertices of the gadget corresponding to $x_{j_1}, \ldots, x_{j_q}$ such that $\chi(P) \subseteq  \{i_1, \ldots, i_p\}$. Assignment $\tau$ also satisfies each \OR-gate{} $g_y$, where $y \in [r]$. By the inductive hypothesis, there is a path $P_y$ between the entry and exit vertices of the gadget for $g_y$ such that $\chi(P_y) \subseteq  \{i_1, \ldots, i_p\}$.
From the construction of $g$, it follows that the path between the entry and exit vertices of the gadget for $g$, which is $P_g=P \circ P_1 \circ \cdots \circ P_r$,  satisfies $\chi(P_g) \subseteq \{i_1, \ldots, i_p\}$.
Conversely, suppose that there is a path $P_g$ between the entry and exit vertices of the gadget for $g$ such that $\chi(P_g) \subseteq \{i_1, \ldots, i_p\}$. Then $P_g$ can be decomposed into a subpath $P$ that traverses the vertices corresponding to $x_{j_1}, \ldots, x_{j_q}$, and subpaths $P_1, \ldots, P_r$, where $P_y$ is a subpath between the entry and exit vertices of the gadget for $g_y$. Since $P$ traverses the vertices corresponding to $x_{j_1}, \ldots, x_{j_q}$, it follows that $\{x_{j_1}, \ldots, x_{j_q}\} \subseteq \{x_{i_1}, \ldots, x_{i_p}\}$. Since $P_y$, $y \in [r]$, is a subpath between the entry and exit vertices of the gadget for $g_y$, by the inductive hypothesis, it follows that $\tau$ satisfies $g_y$. It follows that $\tau$ assigns 1 to all input variables to $g$ and satisfies all the input \OR-gates{} to $g$, and hence, $\tau$ satisfies $g$.
\end{proof}

As it turns out, we can even exclude \FPT{} cost approximation algorithms for \mor{}.  We first need the following theorem:

\begin{theorem}[Corollary 5 of~\cite{Marx13}]
	Unless \FPT $=$ {\rm \W[2]}, {\sc Monotone Weighted Boolean Circuit Satisfiability} for circuits with depth $4$ is not \FPT{} cost approximable.
\end{theorem}

\begin{corollary}
\label{cor:apx}
	Unless \FPT $=$ {\rm \W[2]}, \mor{} parameterized by $k$ is not \FPT{} cost approximable.
\end{corollary}

\begin{proof}
	We reduce from {\sc Monotone Weighted Boolean Circuit Satisfiability} for circuits with depth $4$. Let $C$ be a monotone circuit of depth at most $4$. We can transform $C$ into a monotone circuit in the normalized form  (see the proof of the Normalization Theorem in Section 23.2.2 of~\cite{dfbook} for more details). Since $C$ has depth at most $4$, this procedure terminates in polynomial time, results in a polynomial blow-up of the instance size, and preserves the solutions and their sizes. Finally, it is not hard to see that the reduction in the proof of Theorem~\ref{thm:wsathard} preserves all the solutions and their sizes.
\end{proof}

\begin{rem}
\label{rem:nonplanar}
A noteworthy remark that we close this section with, is to comment on the role that planarity plays in the parameterized complexity of \cmor. If one drops the planarity requirement on the instances of \cmor{} (\ie, considers \cmor{} on general graphs), then it follows from the proof of Theorem~\ref{thm:wsathard} that the resulting problem is \W{\rm [SAT]}-hard. This can be seen by adding a single vertex containing all colors, that serves as a ``color-connector,'' to the instance of \mor{} produced by the \FPT-reduction; this modification results in an instance of the connected obstacle removal problem on apex graphs, establishing the \W{\rm [SAT]}-hardness of this problem on apex graphs.
\end{rem}
\fi

\section{Structural Results}
\label{sec:structural}

Let $G$ be a color-connected plane graph, $C$ a set of colors, and $\chi: V \longrightarrow 2^{C}$. In this section, we present structural results that are the cornerstone of the \FPT-algorithm for \cmor{} presented in the next section. We start by giving an intuitive description of the plan for this section.

As mentioned in Section~\ref{sec:intro}, the main issue facing a dynamic programming algorithm based on tree decomposition, is how to upper bound, by a function of $k$ and the treewidth, the number of $k$-valid paths between (any) two vertices $u$ and $v$ that use color sets contained in a certain bag. As it turns out, this number cannot be upper bounded as desired. Instead, we ``represent'' those paths using a minimal set $\PPP$ of $k$-valid $u$-$v$ paths, in the sense that any $k$-valid $u$-$v$ path can be replaced by a path from $\PPP$ that is not ``worse'' than it. To do so, it suffices to represent the $k$-valid $u$-$v$ paths that use color sets contained in a third vertex $w$, by a set whose cardinality is a function of $k$. This will enable us to extend the notion of a minimal set of $k$-valid $u$-$v$ paths w.r.t.~a single vertex to a representative set for the whole bag, which is the key ingredient of the dynamic programming \FPT-algorithm---based on tree decomposition---in the next section.

As it turns out, the paths that matter are those that use ``external'' colors w.r.t.~$w$ (defined below), since those colors have the potential of appearing on both sides of a bag containing $w$. Therefore, the ultimate goal of this section is to define a notion of a minimal set $\PPP$ of $k$-valid $u$-$v$ paths with respect to $w$ (Definition~\ref{def:minimalpaths}), and to upper bound $|\PPP|$ by a function of $k$. Upper bounding $|\PPP|$ by a function of $k$ turns out to be quite challenging, and requires ideas and topological results that will be discussed later in this section.

Throughout this section, we shall assume that $G$ is color-connected. We start with the following simple observation\ifshort~that holds because of this assumption\fi:

\begin{observation}
\label{obs:colorconnectivity}
Let $x, y \in V(G)$ be such that there exists a color $c \in {\cal C}$ that appears on both $x$ and $y$. Then any $x$-$y$ vertex-separator in $G$ contains a vertex on which $c$ appears.
\end{observation}
\iflong
\begin{proof}
This follows because color $c$ is connected.
\end{proof}
\fi

Let $G'$ be a plane graph, let $w \in V(G')$, and let $f$ be the face in $G'-w$ such that $w$ is interior to $f$; we call $f$ the \emph{external face} w.r.t.~$w$ in $G'$, and the vertices incident to $f$ \emph{external vertices} w.r.t.~$w$ in $G'$. A color $c \in C$ is an \emph{external color} w.r.t.~$w$ in $G'$, or simply \emph{external} to $w$ in $G'$, if $c$ appears on an external vertex w.r.t.~$w$ in $G'$; otherwise, $c$ is \emph{internal} to $w$ in $G'$. The following observation is easy to see:

\begin{observation}
\label{obs:externalsubgraphs}
Let $G$ be a color-connected graph, and let $w \in V(G)$. Let $H$ be any subgraph of $G-w$. If $c$ is an external color to $w$ in $G-w$ and $c$ appears on some vertex in $H$, then $c$ is an external color to $w$ in $H$. This also implies that the set of internal colors to $w$ in $H$ is a subset of the set of internal colors to $w$ in $G-w$.
\end{observation}

\ifshort

\begin{definition}\rm
\label{def:colorcontraction}
Let $s, t$ be two designated vertices in $G$, and let $x, y$ be two adjacent vertices in $G$ such that $\chi(x) = \chi(y)$. Define the following operation to $x$ and $y$, referred to as a \emph{color contraction} operation, that results in a  graph $G'$, a color function $\chi'$, and two designated vertices $s', t'$ in $G'$, obtained as follows:
 \begin{itemize}
 \item $G'$ is obtained from $G$ by contracting the edge $xy$, which results in a new vertex $z$;
 \item $s'=s$ (resp.~$t'=t$) if $s \notin \{x, y\}$ (resp.~$t \notin \{x, y\}$), and $s'=z$ (resp.~$t'=z$) otherwise; and
 \item $\chi': V(G') \longrightarrow 2^{C}$ is defined as $\chi'(w) = \chi(w)$ if $w \neq z$, and $\chi'(z)=\chi(x)=\chi(y)$.
 \end{itemize}
$G$ is \emph{irreducible} if there does not exist two vertices in $G$ to which the color contraction operation is applicable.
\end{definition}

\begin{lemma} [\appno{} Lemma~2.4]
\label{lem:contract}
Let $G$ be a color-connected plane graph, $C$ a color set, $\chi: V \longrightarrow 2^{C}$, $s, t \in V(G)$, and $k \in \mathbb{N}$. Suppose that the color contraction operation is applied to two vertices in $G$ to obtain $G'$, $\chi'$, $s', t'$, as described in Definition~\ref{def:colorcontraction}. Then $G'$ is a color-connected plane graph, and there is a $k$-valid $s$-$t$ path in $G$ if and only if there is a $k$-valid $s'$-$t'$ path in $G'$.
\end{lemma}
\fi

\iflong
\begin{definition} \rm
\label{def:operation}
Let $P=(w_1, \ldots, w_r)$ be a path in a graph $G$, and let $x, y \in V(G)$. Suppose that we apply the color contraction operation to $x$ and $y$, and let $z$ be the new vertex resulting from this contraction. We define an operation, denoted $\Lambda_{xy}$, that when applied to path $P$ results in another path $\Lambda_{xy}(P)$ defined as follows:
\begin{itemize}
\item[1.] If $\{x, y\} \cap \{w_1, \ldots, w_r\} = \emptyset$ then $\Lambda_{xy}(P)=P$.
\item[2.] If $\{x, y\} \cap \{w_1, \ldots, w_r\} = \{w_i\}$, where $i \in [r]$, then $\Lambda_{xy}(P) = (w_1, \ldots, w_{i-1}, z, w_{i+1}, \ldots, w_r)$.
\item[3.] If $\{x, y\} \cap \{w_1, \ldots, w_r\} = \{w_i, w_j\}$, where $i < j$, then
$\Lambda_{xy}(P) = (w_1, \ldots, w_{i-1}, z, w_{j+1}, \ldots, w_r)$.
\end{itemize}
For a set of paths ${\cal P}$, we define $\Lambda_{xy}({\cal P}) = \{\Lambda_{xy}(P) \mid P \in {\cal P}\}$.
\end{definition}
\fi

\begin{definition}\rm
\label{def:minimalpaths}
Let $u, v, w \in V(G)$. A set ${\cal P}$ of $k$-valid $u$-$v$ paths in $G-w$ is said to be {\em minimal} w.r.t.~$w$ if:

\begin{itemize}
\item[(i)] There does not exist two paths $P_1, P_2 \in {\cal P}$ such that $\chi(P_1) \cap \chi(w) = \chi(P_2) \cap \chi(w)$;
\item[(ii)] there does not exist two paths $P_1, P_2 \in {\cal P}$ such that $\chi(P_1) \subseteq \chi(P_2)$; and
\item[(iii)] for any $P \in {\cal P}$, there does not exist a $u$-$v$ path $P'$ in $G-w$ such that $\chi(P') \subsetneq \chi(P)$.
\end{itemize}

Clearly, for any $u, v, w \in V(G)$, a minimal set of $k$-valid $u$-$v$ paths in $G-w$ exists.
\end{definition}

\ifshort
\begin{observation}[\appno{} Observation~4.5]
\label{obs:internalpath}
Let $u, v, w \in V(G)$. Any set of $u$-$v$ paths that is minimal w.r.t.~$w$ contains at most one path whose vertices contain \emph{only} internal colors w.r.t.~$w$ in $G-w$.
\end{observation}
\fi

\iflong
\begin{observation}
\label{obs:internalpath}
Let $u, v, w \in V(G)$. Any set of $u$-$v$ paths that is minimal w.r.t.~$w$ contains at most one path whose vertices contain only internal colors w.r.t.~$w$ in $G-w$.
\end{observation}

\begin{proof}
Since the external face $f$ of $w$ in $G-w$ is a Jordan curve that separates $w$ from any vertex in $G-w$ that is not incident to $f$, by Observation~\ref{obs:colorconnectivity}, any color that appears both on $w$ and on a vertex in $G-w$ must appear on a vertex incident to $f$, and hence, must be external to $w$ by definition. Therefore, any path $P$ containing only internal colors to $w$ satisfies $\chi(P) \cap \chi(w) = \emptyset$. The observation now follows from property (i) in Definition~\ref{def:minimalpaths}.
\end{proof}
\fi

\iflong
\begin{lemma}
\label{lem:contractionminimal}
Let $u, v, w \in V(G)$, and let ${\cal P}$ be a minimal set of $k$-valid $u$-$v$ paths in $G-w$.
Suppose that we apply the color contraction operation to an edge $xy \in G-w$, and let $G', \chi'$ be the graph and color function obtained from the contraction operation, respectively. Let ${\cal P'} = \Lambda_{xy}({\cal P})$. Then ${\cal P'}$ is a minimal set of $k$-valid paths w.r.t~$w$ in $G'$.
\end{lemma}

\begin{proof}
	Let $H'$ be the subgraph of $G'-w$ induced by the edges of the paths in ${\cal P'}$, and denote by $z$ the new vertex obtained from the contraction of the edge $xy$. We start by showing the following claim:
\begin{claim}\label{claim:minimality}
	For every $P\in \PPP$, it holds that $\chi(\Lambda_{xy}(P)) = \chi(P)$.
\end{claim}

Let $P=(u=w_1, \ldots, w_r=v)$. Since $\chi'(z)=\chi(x)=\chi(y)$, it follows from Definition~\ref{def:operation} that if  $|\{x, y\} \cap \{w_1, \ldots, w_r\}|\le 1$, then $\chi(\Lambda_{xy}(P)) = \chi(P)$. Now assume that $\{x, y\} \cap \{w_1, \ldots, w_r\} = \{w_i, w_j\}$, where $i < j$, and suppose to get a contradiction that $\chi(\Lambda_{xy}(P)) \neq \chi(P)$. Since $\Lambda_{xy}(P) = (w_1, \ldots, w_{i-1}, z, w_{j+1}, \ldots, w_r)$, it follows that $\chi(\Lambda_{xy}(P)) \subsetneq \chi(P)$. However, $G-w$ contains the $u$-$v$ path $P' = (w_1, \ldots, w_{i-1}, w_i, w_j, w_{j+1}, \ldots, w_r)$, which satisfies $\chi(P')=\chi(\Lambda_{xy}(P))\subsetneq \chi(P)$; this, together with $P\in \PPP$, contradicts the minimality of $\PPP$.

We now proceed to verify that $\PPP'$ is indeed minimal w.r.t.~$w$. Properties $(i)$ and $(ii)$ in Definition~\ref{def:minimalpaths} follow directly from Claim~\ref{claim:minimality} and the minimality of $\PPP$. To prove that property $(iii)$ holds, assume that there is a path $P'\in \PPP'$, and a path $Q'$ in $G'-w$ between the endpoints of $P'$ such that $\chi(Q') \subsetneq \chi(P')$. Let $P$ be the path in $\PPP$ such that $\Lambda_{xy}(P)=P'$.  It is straightforward to verify that $G-w$ contains a $u$-$v$ path $Q$ that is either identical to $Q'$, or obtained from $Q'$ by replacing $z$ by either a single vertex $x$ or $y$, or by the pair $x,y$. Clearly, $\chi(Q) = \chi(Q')$. Since $\chi(Q') \subsetneq \chi(P') = \chi(P)$ by Claim~\ref{claim:minimality}, it follows that $\chi(Q) \subsetneq \chi(P)$, contradicting the minimality of $\PPP$. It follows that Property $(iii)$ holds, and the proof is complete.
\end{proof}
\fi

To derive an upper bound on the cardinality of a minimal set $\PPP$ of $k$-valid $u$-$v$ paths w.r.t.~a vertex $w$, we select a maximal set $\MMM$ of color-disjoint paths in $\PPP$. We first upper bound $|\MMM|$ by a function of $k$, which requires developing several results of topological nature. The key ingredient for upper bounding $|\MMM|$ is showing that the subgraph induced by the paths in $\MMM$ has a $u$-$v$ vertex-separator of cardinality $O(k)$ (Lemma~\ref{lem:separator}). We then upper bound $|\MMM|$ (Lemma~\ref{lem:aux}) by upper bounding the number of different traces of the paths of $\MMM$ on this small separator, and inducting on both sides of the separator. Finally, we show (Theorem~\ref{thm:main}) that $|\PPP|$ is upper bounded by a function of $|\MMM|$, which proves the desired upper bound on $|\PPP|$. We proceed to the details.


For the rest of this section, let $u, v, w \in V(G)$, and let ${\cal P}$ be a set of minimal $k$-valid $u$-$v$ paths in $G-w$.  Let ${\cal M}$ be a set of minimal $k$-valid color-disjoint $u$-$v$ paths in $G-w$, and let $M$ be the subgraph of $G-w$ induced by the edges of the paths in ${\cal M}$.

\ifshort
\begin{observation}[\appno{} Observation~4.7]
\label{obs:externalcolor}
If $P \in \MMM$ contains a color $c$ that is external to $w$ in $M$, then $c$ appears on a vertex in $P$ that is incident to the external face to $w$ in $M$.
\end{observation}
\fi

\iflong
\begin{observation}
\label{obs:externalcolor}
If $P \in \MMM$ contains a color $c$ that is external to $w$ in $M$, then $c$ appears on a vertex in $P$ that is incident to the external face to $w$ in $M$.
\end{observation}

\begin{proof}
By definition, $c$ appears on a vertex $x$ incident to the external face w.r.t.~$w$ in $M$. Since the paths in $\MMM$ are pairwise color-disjoint and $c$ appears on $P$, it follows that $x$ is a vertex of $P$.
\end{proof}
\fi

\iflong
\begin{lemma}
\label{lem:generalpaths}
Let $G'$ be a plane graph, and let $x, y,z \in V(G')$. Let $x_1, \ldots, x_r$, $r \geq 3$, be the neighbors of $x$ in counterclockwise order. Suppose that, for each $i \in [r]$, there exists an $x$-$y$ path $P_i$ containing $x_i$ such that $P_i$ does not contain $z$ and does not contain any $x_j$, $j \in [r]$ and $j \neq i$. Then there exist two paths $P_i, P_j$, $i, j \in [r]$ and $i \neq j$, such that the two paths $P_i, P_j$ induce a Jordan curve separating $\{x_1, \ldots, x_r\} \setminus \{x_i, x_j\}$ from $z$.
\end{lemma}

\begin{proof}
The proof is by induction on $r \geq 3$. The base case is when $r=3$. Consider the faces induced by the two paths $P_1$ and $P_2$ in the embedding. If $z$ and $x_3$ are in two separate faces, then clearly $P_1$ and $P_2$
induce a Jordan curve separating $x_3$ from $z$, and we are done. Therefore, we can assume that $z$ and $x_3$ are in the same face induced by $P_1$ and $P_2$. Since $P_1$ does not contain $x_2$, we can continuously deform $P_1$ into an isotopic non self-intersecting curve $P'_1$ w.r.t.~$x_3, x_2, z$, that includes $xx_1$, intersects edges $xx_2$ and $xx_3$ only at $x$, and intersects $P_2$ only at $x$ and $y$. Similarly,
if $P_2$ and $P_3$ do not separate $z$ from $x_1$, then $z$ and $x_1$ are in the same face induced by $P_2$ and $P_3$ and we can define a curve $P'_3$ that is isotopic to $P_3$ w.r.t.~$x_2, x_1, z$, and such that
$P'_3$ contains $xx_3$, intersects $xx_2$ and $xx_1$ only at $x$, and intersects $P_2$ only at $x$ and $y$. Now if $z$ and $x_2$ are in different faces induced by $P'_1$ and $P'_3$, then $P'_1$ and $P'_3$ separate $z$ from $x_2$, and since $P_1$ is isotopic to $P'_1$ w.r.t.~$z$ and $x_2$, and $P'_3$ is isotopic to $P_3$ w.r.t.~$z$ and $x_2$, it follows that $P_1$ and $P_3$ induce a Jordan curve that separates $x_2$ from $z$. Assume now that $z$ and $x_2$ are in the same face $f$ induced by $P'_1$ and $P'_3$. Since $P_2$ intersects with each of $P'_1$ and $P'_3$ precisely at $x$ and $y$, it follows that $P_2$ splits $f$ into two faces $f_1, f_2$, where $xx_2, xx_1$ are two consecutive edges on the boundary of $f_1$ and $xx_2, xx_3$ are two consecutive edges on the boundary of $f_2$. Then, $z$ must be interior to exactly one of the two faces $f_1, f_2$. If $z$ is interior to $f_1$,
let $f'_1$ be the face induced by $P'_1$ and $P_2$ and containing $z$. Then $f'_1$ contains $f_1$, and does not contain $x_3$ (because $P'_1$ intersects $xx_3$ only at $x$). Therefore, $f'_1$, and hence, $P'_1$ and $P_2$ induce a Jordan curve that separates $z$ from $x_3$. It follows that $P_1$, which is isotopic to $P'_1$ w.r.t.~$x_2, x_3, z$, and $P_2$ induce a Jordan curve that separates $z$ from $x_3$. Similarly, if $z$ is interior to $f_2$, then $P'_3, P_2$ induce a Jordan curve that separates $z$ from $x_1$, and hence, $P_3$ and $P_2$ induce a Jordan curve that separates $z$ from $x_1$.

Assume inductively that the statement of the lemma is true for any $3 \leq \ell < r$. By the inductive hypothesis applied to $x_1, \ldots, x_{r-1}$, there exist two paths $P_i, P_j$, $i, j \in [r-1]$ and $i \neq j$, such that the two paths $P_i, P_j$ induce a Jordan curve separating $\{x_1, \ldots, x_{r-1}\} \setminus \{x_i, x_j\}$ from $z$. If $x_r$ and $z$ are not in the same face induced by $P_i, P_j$, then $P_i, P_j$ separate $x_r$ from $z$ as well, and we are done. Assume now that $z$ and $x_r$ are in the same face $f$ induced by $P_i, P_j$. Since $P_i, P_j$ separate $z$ from $\{x_1, \ldots, x_{r-1}\} \setminus \{x_i, x_j\}$, none of $\{x_1, \ldots, x_{r-1}\} \setminus \{x_i, x_j\}$ is interior to $f$, and hence, $x_r$ is the only neighbor of $x$ between $x_i$ and $x_j$ w.r.t.~the rotation system of $G'$, which implies w.l.o.g. that $x_1=x_i$ and $x_{r-1}=x_j$. By the inductive hypothesis applied to $x_1, x_{r-1}, x_r$ there are two paths in $P_1, P_{r-1}, P_r$ that induce a Jordan curve that separates $z$ from one of $x_1, x_{r-1}, x_r$. Since $P_1$ and $P_{r-1}$ do not separate $x_r$ from $z$, one of these two path must be $P_r$; assume, w.l.o.g., that the two paths are $P_1$ and $P_r$. Since $x_1$ and $x_r$ are consecutive neighbors in the rotation system, and since $P_1, P_r$ do not contain any of $x_2, \ldots, x_{r-1}$, it follows that $x_2, \ldots, x_{r-1}$ are in the same face induced by $P_1, P_r$, and this face does not contain $z$ because $P_1, P_r$ separate $z$ from $x_{r-1}$. It follows that $P_1, P_r$ induce a Jordan curve that separates $z$
from $x_2, \ldots, x_{r-1}$. This completes the inductive proof.
\end{proof}

\begin{lemma}
\label{lem:paths}
Let $G'$ be a plane graph with a face $f$, and let $u, v \in V(G')$. Let $u_1, \ldots, u_r$, $r \geq 3$, be the neighbors of $u$. Suppose that, for each $i \in [r]$, there exists a $u$-$v$ path $P_i$ in $G'$ containing $u_i$ and a vertex incident to $f$ different from $v$, and such that $P_i$ does not contain any $u_j$, $j \in [r]$, $j \neq i$. Then there exist two paths $P_i, P_j$, $i, j \in [r]$, $i \neq j$, such that $V(P_i) \cup V(P_j) - \{v\}$ is a vertex-separator separating $\{u_1, \ldots, u_r\} \setminus \{u_i, u_j\}$ from $v$.
\end{lemma}

\begin{proof}
Create a new vertex $y$ interior to $f$. Each path $P_i$, $i\in [r]$, contains a vertex $y_i$ incident to $f$ and different from $v$; we define a new path $P'_i$ from $u$ to $y$, consisting of the prefix of $P_i$ up to $y_i$, and extending this prefix by adding a new edge between $y_i$ and the new vertex $y$. Note that we can extend the rotation system of $G'$ in a straightforward manner to obtain a rotation system for the plane graph resulting from adding $y$ and the edges $y_iy$ to $G'$, $i \in [r]$. Since $v$ is the endpoint of $P_i$ and $v \neq y_i$, it follows that $v$ is not contained in $P'_i$, for $i \in [r]$. By Lemma~\ref{lem:generalpaths}, there exist two paths $P'_i, P'_j$, $i, j \in [r]$, and $i\neq j$, such that the two paths $P'_i, P'_j$ induce a Jordan curve separating $\{u_1, \ldots, u_r\} \setminus \{u_i, u_j\}$ from $v$ in $G'+y$. It follows that $V(P'_i) \cup V(P'_j)$ is a vertex-separator separating $\{u_1, \ldots, u_r\} \setminus \{u_i, u_j\}$ from $v$ in $G'+y$, and hence, $V(P'_i) \cup V(P'_j) -\{y\} \subseteq V(P_i) \cup V(P_j) - \{v\}$ is a vertex-separator separating $\{u_1, \ldots, u_r\} \setminus \{u_i, u_j\}$ from $v$ in $G'$.
\end{proof}
\fi

\ifshort
\begin{lemma}[\appno{} Lemma~4.9]
\label{lem:paths}
Let $G'$ be a plane graph with a face $f$, let $u, v \in V(G')$, and let $u_1, \ldots, u_r$, $r \geq 3$, be the neighbors of $u$. Suppose that, for each $i \in [r]$, there exists a $u$-$v$ path $P_i$ containing $u_i$ and a vertex incident to $f$ different from $v$, and such that $P_i$ does not contain any $u_j$, $j \in [r]$, $j \neq i$. Then there exist two paths $P_i, P_j$, $i, j \in [r]$, $i \neq j$, such that $V(P_i) \cup V(P_j) - \{v\}$ is a vertex-separator separating $\{u_1, \ldots, u_r\} \setminus \{u_i, u_j\}$ from $v$.
\end{lemma}

\begin{lemma}[\appno{} Lemma~4.10]
\label{lem:externalpaths}
Let $x, y$ be two vertices in an irreducible subgraph $G'$ of $G$, and let $f$ be a face in $G'$. Then there are at most two color-disjoint $x$-$y$ paths in $G'$ that contain only colors that appear on $f$.
\end{lemma}
\fi

\iflong
\begin{lemma}
\label{lem:externalpaths}
Let $x, y$ be two vertices in an irreducible subgraph $G'$ of $G$, and let $f$ be a face in $G'$. Then there are at most two color-disjoint $x$-$y$ paths in $G'$ that contain only colors that appear on $f$.
\end{lemma}

\begin{proof}
Suppose, to get a contradiction, that there are three color-disjoint $x$-$y$ paths $P_1, P_2, P_3$ in $G'$ that contain only colors that appear on $f$.
We create a new vertex $z$ interior to $f$ and add edges between $z$ and each vertex incident to $f$. Note that we can extend the rotation system of $G'$ in a straightforward manner to obtain a rotation system for the plane graph resulting from adding $z$ and the edges incident to it to $G'$. Clearly, none of $P_1, P_2, P_3$ contains $z$. Because the paths $P_1, P_2, P_3$ are color-disjoint, both $x$ and $y$ must be empty vertices. Let $v_1, v_2, v_3$ be the neighbors of $x$ on $P_1, P_2, P_3$, respectively. Since $x$ is an empty vertex and $G'$ is irreducible, none of $v_1, v_2, v_3$ is an empty vertex, and hence each $v_i$, $i \in [3]$, must contain a color $c_i$ that appears on $f$. Since $P_1, P_2, P_3$ are pairwise color-disjoint, it follows that no vertex in $\{v_1, v_2, v_3\}\setminus\{v_i\}$ is contained in $P_i$, for $i\in [3]$.
By Lemma~\ref{lem:generalpaths}, there is a $v_i$, $i \in [3]$, such that the two paths in $\{P_1, P_2, P_3\} -P_i$ induce a Jordan curve in $G'+z$ separating $v_i$ and $z$, and hence separating $v_i$ from each vertex incident to $f$. Since $c_i$ appears on both $v_i$ and a vertex incident to $f$, by Observation~\ref{obs:colorconnectivity}, it follows that $c_i$ must appear on a vertex in $V(P_1) \cup V(P_2) \cup V(P_3) -V(P_i)$. This is a contradiction since $c_i$ appears on $P_i$ and the paths $P_1, P_2, P_3$ are pairwise color-disjoint.
\end{proof}
\fi
\iflong
\begin{figure}[htbp]
\begin{center}
\includegraphics[width=12cm]{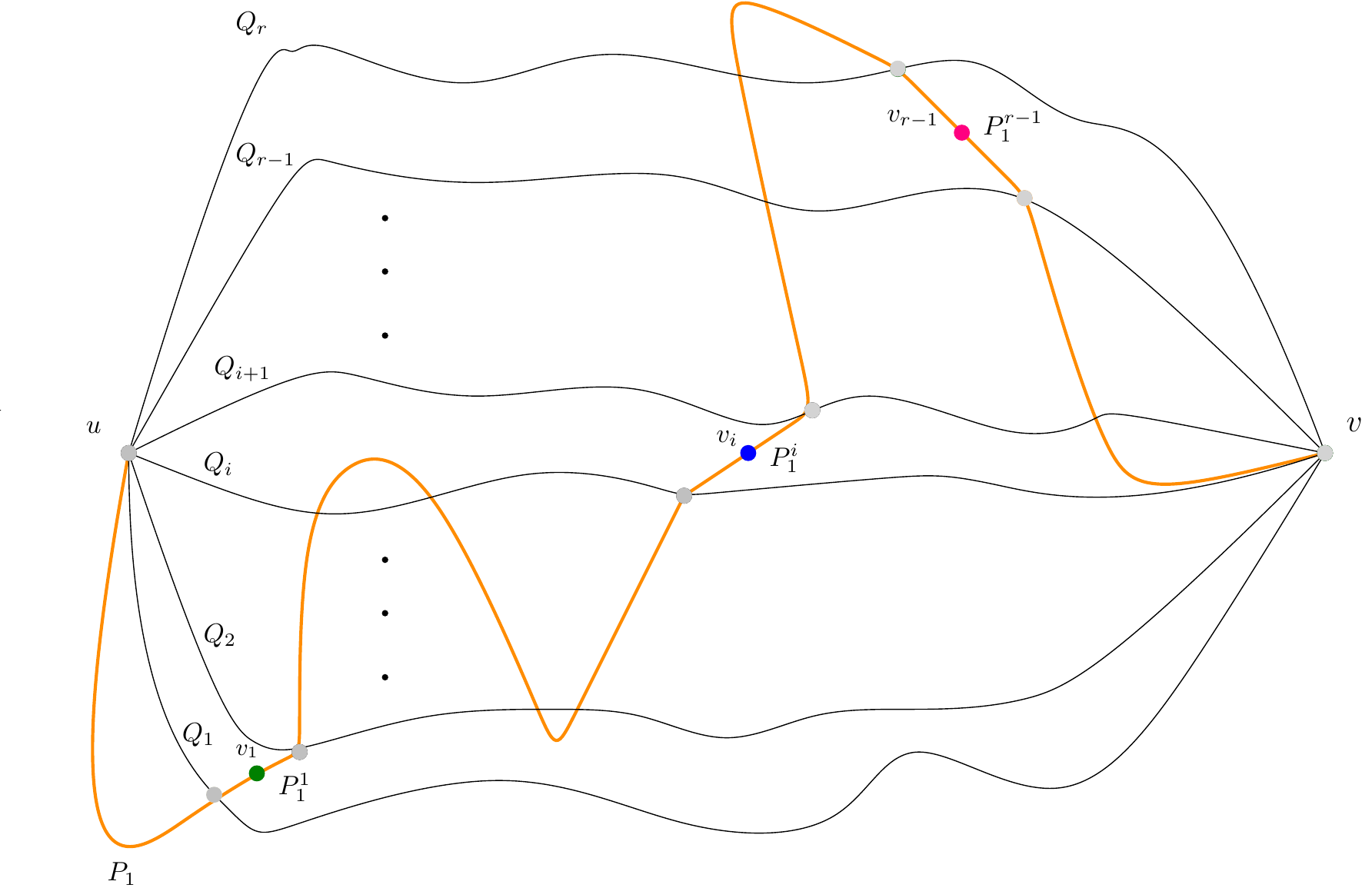}
\end{center}
\caption{Illustration for the proof of Lemma~\ref{lem:separator}.}
\label{fig:separator}
\end{figure}
\fi

\begin{lemma}
\label{lem:separator}
Suppose that $M$ is irreducible, then there exist paths $P_1,P_2,P_3\in \MMM$ such that $M-P_1-P_2-P_3$ has a $u$-$v$ vertex-separator of cardinality at most $2k+3$.
\end{lemma}

\begin{proof}
By Observation~\ref{obs:internalpath} and Observation~\ref{obs:externalsubgraphs}, $\MMM$ contains at most one path that contains only internal colors w.r.t.~$w$ in $M$. Therefore, it suffices to show that $\MMM$ contains two paths $P_1, P_2$ such that  $M-P_1-P_2$ has a $u$-$v$ vertex-separator of cardinality at most $2k+3$, assuming that every path in $\MMM$ contains an external color w.r.t.~$w$ in $M$.

By Observation~\ref{obs:externalcolor}, every path in $\MMM$ passes through an external vertex  w.r.t.~$w$ in $M$ that contains an external color to $w$ in $M$. Because the paths in $\MMM$ are pairwise color-disjoint and $u$ and $v$ are empty vertices, every path in $\MMM$ passes through a vertex on the external face of $w$ in $M$ that is different from $u$ and $v$. Let $u_1, \ldots, u_q$ be the neighbors of $u$ in $M$, and note that since $u$ is empty and $M$ is irreducible, each $u_i$, $i \in [q]$, contains a color. Let $P_1, \ldots, P_q$ be the paths in $\MMM$ containing $u_1, \ldots, u_q$, respectively, and note that since the paths in $\MMM$ are color-disjoint, no $P_i$ passes through $u_j$, for $j \neq i$. By Lemma~\ref{lem:paths}, there are two paths in $P_1, \ldots, P_q$, say $P_1, P_2$ without loss of generality, such that $V_{12}=V(P_1) \cup V(P_2) - \{v\}$ is a vertex-separator that separates $\{u_3, \ldots, u_q\}$ from $v$.


We proceed by contradiction and assume that $M^-=M-P_1-P_2$ does not have a $u$-$v$ vertex-separator of cardinality $2k+3$. By Menger's theorem~\cite{diestel}, there exists a set $\DDD$ of $r' \geq 2k+3$ vertex-disjoint $u$-$v$ paths in $M^-$. Since $V_{12}$ separates $\{u_3, \ldots, u_q\}$ from $v$ in $M$, every $u$-$v$ path in $M^-$
intersects at least one of $P_1$, $P_2$ at a vertex other than $v$.
It follows that there exists a path in $\{P_1, P_2\}$, say $P_1$, that intersects at least $k+2$ paths in $\DDD$ at vertices other than $v$. Since the paths in $\DDD$ are vertex-disjoint and incident to $u$, we can order the paths in $\DDD$ that intersect $P_1$ around $u$ (in counterclockwise order) as $\langle Q_1, \ldots, Q_r\rangle$, where $r \geq k+2$, and $Q_{i+1}$ is counterclockwise from $Q_{i}$, for $i \in [r-1]$. $P_1$ intersects each path $Q_i$, $i\in [r]$, possibly multiple times.  Moreover, since the paths in $\MMM$ are pairwise color-disjoint, each intersection between $P_1$ and a path $Q_i$, $i\in [r]$, must occur at an empty vertex. We choose $r-1$ subpaths, $P_1^{1}, \ldots, P_{1}^{r-1}$, of $P_1$
satisfying the property that the endpoints of $P_1^{i}$ are on $Q_i$ and $Q_{i+1}$, for $i =1, \ldots, r-1$, and the endpoints of $P_1^{i}$ are the only vertices on $P_1^{i}$ that appear on a path $Q_j$, for $j \in [r]$.
It is easy to verify that the subpaths $P_1^{1}, \ldots, P_{1}^{r-1}$ of $P_1$ can be formed by following the intersection of $P_1$ with the sequence of (ordered) paths $Q_1, \ldots, Q_r$. \iflong See Figure~\ref{fig:separator} for illustration. \fi \ifshort (\appno{} Refer to Figure~8 for illustration.)\fi

 Recall that the endpoints of $P_1^{1}, \ldots, P_{1}^{r-1}$ are empty vertices. Since $M$ is irreducible, no two empty vertices are adjacent, and hence, each subpath $P_1^{i}$ must contain an internal vertex $v_i$ that contains at least one color. We claim that no two vertices $v_i, v_j$, $1 \leq i < j \leq r-1$, contain the same color. Suppose not, and let $v_i, v_j$, $i < j$, be two vertices containing a color $c$. Since $v_i, v_j$ are internal to $P_1^{i}$ and $P_1^{j}$, respectively, $Q_1, \ldots, Q_r$ are vertex-disjoint $u$-$v$ paths, and by the choice of the subpaths $P_1^{1}, \ldots, P_{1}^{r-1}$, the paths $Q_{i}$ and $Q_{i+1}$ form a Jordan curve, and hence a vertex-separator in $G$, separating $v_i$ from $v_j$.
 By Observation~\ref{obs:colorconnectivity}, color $c$ must appear on a vertex in $Q_p$, $p\in \{i, i+1\}$, and this vertex is clearly not in $P_1$ since $P_1$ intersects $Q_p$ at empty vertices. Since every vertex in $M$ appears on a path in $\MMM$, and $c$ appears on $P_1 \in \MMM$ and on a vertex not in $P_1$, this contradicts that the paths in $\MMM$ are pairwise color-disjoint, and proves the claim.

 Since no two vertices $v_i, v_j$, $1 \leq i < j \leq r$, contain the same color, the number $r-1$ of subpaths $P_1^{1}, \ldots, P_{1}^{r-1}$ is upper bounded by the number of distinct colors that appear on $P_1$, which is at most $k$. It follows that $r$ is at most $k+1$, contradicting our assumption above and proving the lemma.
\end{proof}

\ifshort
\begin{lemma}[\appno{} Lemma~4.12]
\label{lem:contraction}
Let $S$ be a minimal $u$-$v$ vertex-separator in $M$. Let $M_u, M_v$ be a partition of $M -S$ containing $u$ and $v$, respectively, and such that there is no edge between $M_u$ and $M_v$. For any vertex $x \in S$, $M_u$ is contained in a single face of $M_v +x$.
\end{lemma}
\fi

\iflong
\begin{lemma}
\label{lem:contraction}
Let $S$ be a minimal $u$-$v$ vertex-separator in $M$. Let $M_u, M_v$ be a partition of $M -S$ containing $u$ and $v$, respectively, and such that there is no edge between $M_u$ and $M_v$. For any vertex $x \in S$, $M_u$ is contained in a single face of $M_v +x$.
\end{lemma}

\begin{proof}
Let $x \in S$. It suffices to show that the subgraph $F$ of $M$ induced by $V(M_u) \cup (S\setminus \{x\})$ is connected. This suffices because $V(F)$ and $V(M_v+x)$ are disjoint, and hence every face in $M_v+x$ separates the vertices in $V(F)$ inside the face from those outside of it. We will show that $F$ is connected by showing that there is a path in $F$ from each vertex in $F$ to $u \in V(F)$. Let $z \in V(F)$. If $z \in S$, then by minimality of $S$, there is a path from $u$ to $z$ whose internal vertices are all in $M_u$, and hence this path is in $F$. If $z \notin S$, let $P$ be a $u$-$v$ path containing $z$. If $P$ passes through $z$ before passing through any vertex in $S$, then clearly there is a path from $u$ to $z$ in $F$. Otherwise, $P$ passes through a vertex $y \in S$ before passing through $z$. In this case, there exists a vertex $y' \in S$, such that $y' \neq y$ and $P$ passes through $y'$ after passing through $z$. Either $y$ or $y'$, say $y'$, is different from $x$. From the above discussion, there is a path $P'$ from $u$ to $y'$ in $F$, which when combined with the subpath of $P$ between $y'$ and $z$ yields a path from $u$ to $z$ in $F$.
\end{proof}
\fi

\ifshort
\begin{lemma}
\label{lem:aux}
$|{\cal M}|\leq g(k)$, where $g(k)=\Oh(c^k k^{2k})$, for some constant $c > 1$.
\end{lemma}

\begin{proof}(sketch) [\appno{} see Lemma~4.13.]
 By Observation~\ref{obs:internalpath}, there is at most one path in $\MMM$ that contains only internal colors w.r.t.~$w$ in $G-w$. Therefore, it suffices to upper bound the number of paths in $\MMM$ that contain at least one external color to $w$ in $G-w$. By Observation~\ref{obs:externalsubgraphs}, every such path in $\MMM$ contains a color that is external to $w$ in $M$.

The proof is by induction on $k$, over every color-connected plane graph $G$, every triplet of vertices $u, v, w$ in $G$, and every minimal set $\MMM$ w.r.t.~$w$ of $k$-valid pairwise color-disjoint $u$-$v$ paths in $G-w$.  If $k=1$, any path in $\MMM$ contains exactly one external color w.r.t.~$w$ in $M$. By Lemma~\ref{lem:externalpaths}, at most two paths in $\MMM$ contain only external colors.
Assume by the inductive hypothesis that, for each $1 \leq i < k$, we have $|\MMM| \leq g(i)$. We can assume that $M$ is irreducible; otherwise, we apply the color contraction operation and replace $\MMM$ with a set of paths satisfying the same properties as $\MMM$ (\appno{} Definition~4.3 and Lemma~4.6).

By Lemma~\ref{lem:separator}, there are (at most) $3$ paths in $\MMM$, such that the subgraph of $M$ induced by the remaining paths of $\MMM$ has a $u$-$v$ vertex-separator $S$ satisfying $|S| \leq 2k+3$. Remove these $3$ paths from $\MMM$, and now we can assume that $M$ has a $u$-$v$ vertex-separator $S$ satisfying $|S| \leq 2k+3$; we will add 3 to the upper bound on $|\MMM|$ at the end.  We can assume that $S$ is minimal. $S$ separates $M$ into two subgraphs $M_u$ and $M_v$ such that $u \in V(M_u)$, $v \in V(M_v)$, and there is no edge between $M_u$ and $M_v$. We partition $\MMM$ into the following groups, where each group excludes the paths satisfying the properties of the groups defined before it: (1) The set of paths in $\MMM$ that contain a nonempty vertex in $S$; (2) the set of paths $\MMM_{u}^{k}$ consisting of each path $P$ in $\MMM$ such that all colors on $P$ appear on vertices in $M_u$ (these colors could still appear on vertices in $M_v$ as well); (3) the set of paths $\MMM_{v}^{k}$ consisting of each path $P$ in $\MMM$ such that all colors on $P$ appear on vertices in $M_v$; and (4) the set $\MMM^{< k}$ of remaining paths in $\MMM$, satisfying that each path contains a nonempty external vertex to $w$ in $M$ and contains less than $k$ colors from each of $M_u$ and $M_v$.

Since the paths in $\MMM$ are pairwise color-disjoint, no nonempty vertex in $S$ appears on two distinct paths from group (1). Therefore, the number of paths in group (1) is at most $|S| \leq 2k+3$.
Observe that the vertices in $S$ contained in any path from groups (2)-(4) are empty vertices. To upper bound the number of paths in group (2), for each path $P$, there is a last vertex $x_P$ (\ie, farthest from $u$) in $P$ that is in $S$. Fix a vertex $x \in S$, and let us upper bound the number of paths $P$ in group (2) for which $x=x_P$. Let $P_v$ be the subpath of $P$ from $x$ to $v$. Note that since $v$ is empty and all the vertices in $S$ that are contained in paths in group (2) are empty, and since $M$ is irreducible, $P_v$ must contain at least one color. Since all colors appearing on $P$ appear on vertices in $M_u$, all colors appearing on $P_v$ appear in $M_u$. By Lemma~\ref{lem:contraction}, $M_u$ is contained in a single face $f$ of $M_v+x$.
Since $f$ is a vertex-separator that separates $V(M_u)$ from $V(P_v)$ in $G$, by Observation~\ref{obs:colorconnectivity}, every color that appears on $P_v$ appears on $f$. By Lemma~\ref{lem:externalpaths}, there are at most two $x$-$v$ paths that contain only colors that appear on $f$. This shows that there are at most two paths in group (2) for which $x$ is the last vertex in $S$. Since $|S| \leq 2k+3$, this upper bounds the number of paths in group (2) by $2(2k+3)=4k+6$. By symmetry, the number of paths in group (3) is upper bounded by $4k+6$.

To upper bound the number of paths in group (4), let $S=\{s_2, \ldots, s_{r-1}\}$, and extend $S$ by adding the two vertices $s_1=u$ and $s_r=v$ to form the set
$A=\{s_1, s_2, \ldots, s_{r}\}$. For every two (distinct) vertices $s_j, s_{j'} \in A$, we define a set of paths $\PPP_{jj'}$ in $G-w$ whose endpoints are $s_j$ and $s_{j'}$ as follows.
For each path $P$ in group (4), partition (the edges in) $P$ into subpaths $P_1, \ldots, P_q$ satisfying the property that the endpoints of each $P_i$, $i \in [q]$, are in $A$, and no internal vertex to $P_i$ is in $A$.
For each $P_i$, $i \in [q]$, such that $P_i$ contains a vertex that contains an external color to $w$ in $G-w$, let $P'_i$ (possibly $P_i$) be a subpath in $G-w$ between the endpoints of $P_i$ satisfying that $\chi(P'_i) \subseteq \chi(P_i)$ and $\chi(P'_i)$ is minimal w.r.t.~containment. Since $P$ contains a vertex that contains an external color to $w$ in $G-w$, it is easy to see that there exists an $i \in [q]$ such that $P'_i$ contains a vertex containing an external color to $w$ in $G-w$~\app. Pick any $i \in [q]$ satisfying that $P'_i$ contains a vertex containing an external color to $w$ in $G-w$, associate $P$ with $P'_i$, and assign $P'_i$ to the set of paths $\PPP_{jj'}$ such that $s_j$ and $s_{j'}$ are the endpoints of $P'_i$. The map that takes each $P$ to its $P'_i$ is clearly a bijection~\app.

Therefore, it suffices to upper bound the number of paths assigned to the sets $\PPP_{jj'}$. Fix a set $\PPP_{jj'}$. The paths in $\PPP_{jj'}$ have $s_j, s_{j'}$ as endpoints, and are pairwise color-disjoint. It is not difficult to show that $\PPP_{jj'}$ is a minimal set of $(k-1)$-valid $s_j$-$s_{j'}$ paths in $G-w$ w.r.t.~$w$~\app. By the inductive hypothesis, we have $|\PPP_{jj'}| \leq g(k-1)$. Since the number of sets $\PPP_{jj'}$ is at most $\binom{2k+5}{2}$, the number of paths in group (4) is $\Oh(k^2) \cdot g(k-1)$.

It follows that $|\MMM| \leq g(k)$, where $g(k)$ satisfies:
\[g(k) \leq 3 + (2k+3) + 2(4k+6) + \Oh(k^2) \cdot g(k-1) = \Oh(k^2) \cdot g(k-1),\] where $3$ accounts for the $3$ paths removed from $\MMM$. Solving the aforementioned recurrence relation gives $g(k) =\Oh(c^k k^{2k})$, where $c > 1$ is a constant.
\end{proof}

Applying Lemma~\ref{lem:aux} to a maximal set $\MMM$ of color-disjoint paths in $\PPP$, and using an inductive proof, we can show the following theorem:

\begin{theorem}[\appno{} Theorem~4.14 and Corollary~4.15]
\label{thm:main}
Let $G$ be a plane color-connected graph, and let $w \in V(G)$. Let $G'$ be a subgraph of $G-w$, and let $u, v \in V(G')$. Every set ${\cal P}$ of minimal $k$-valid $u$-$v$ paths in $G'$ w.r.t.~$w$ satisfies $|\PPP|\leq h(k)$, where $h(k)=\Oh(c^{k^2} k^{2k^2+k})$, for some constant $c > 1$.
\end{theorem}

\fi

\iflong
\begin{lemma}
\label{lem:aux}
$|{\cal M}|\leq g(k)$, where $g(k)=\Oh(c^k k^{2k})$, for some constant $c > 1$.
\end{lemma}

\begin{proof}
 By Observation~\ref{obs:internalpath}, there can be at most one path in $\MMM$ that contains only internal colors w.r.t.~$w$ in $G-w$. Therefore, it suffices to upper bound the number of paths in $\MMM$ that contain at least one external color to $w$ in $G-w$. Without loss of generality, in the rest of the proof, we shall assume that $\MMM$ does not include a path that contains only internal colors w.r.t.~$w$ in $G-w$, and upper bound $|\MMM|$ by $g(k)$; adding 1 to $g(k)$ we obtain an upper bound on $|\MMM|$ with this assumption lifted. Note that by Observation~\ref{obs:externalsubgraphs}, the previous assumption implies that every path in $\MMM$ contains a color that is external to $w$ in $M$.

The proof is by induction on $k$, over every color-connected plane graph $G$, every triplet of vertices $u, v, w$ in $G$, and every minimal set $\MMM$ w.r.t.~$w$ of $k$-valid pairwise color-disjoint $u$-$v$ paths in $G-w$.  If $k=1$, then any path in $\MMM$ contains exactly one external color w.r.t.~$w$ in $M$. By Lemma~\ref{lem:externalpaths}, at most two paths in $\MMM$ contain only external colors. It follows that for $k=1$, $|\MMM| \leq 2 \leq g(1)$, if we choose the hidden constant in the $\Oh$ asymptotic notation to be at least $2$.

Suppose by the inductive hypothesis that for any $1 \leq i < k$, we have $|\MMM| \leq g(i)$. We can assume that $M$ is irreducible; otherwise, we apply the color contraction operation to any edge $xy$ in $\MMM$ to which the operation is applicable, and replace $\MMM$ with the set of paths $\Lambda_{xy}(\MMM)$, which is pairwise color-disjoint, contains the same number of paths as $\MMM$, and is minimal w.r.t.~$w$ by Lemma~\ref{lem:contractionminimal}.

By Lemma~\ref{lem:separator}, there are at most $3$ paths in $\MMM$, such that the subgraph of $M$ induced by the remaining paths of $\MMM$ has a $u$-$v$ vertex-separator $S$ satisfying $|S| \leq 2k+3$. To simplify the argument, in what follows, we assume that we already removed these $3$ paths from $\MMM$ and that $M$ already has a $u$-$v$ vertex-separator $S$ satisfying $|S| \leq 2k+3$. We will add 3 to the upper bound of $|\MMM|$ at the end to account for these removed paths.  We can assume, without loss of generality, that $S$ is minimal (w.r.t.~containment). $S$ separates $M$ into two subgraphs $M_u$ and $M_v$ such that $u \in V(M_u)$, $v \in V(M_v)$, and there is no edge between $M_u$ and $M_v$. We partition $\MMM$ into the following groups, where each group excludes the paths satisfying the properties of the groups defined before it: (1) The set of paths in $\MMM$ that contain a nonempty vertex in $S$; (2) the set of paths $\MMM_{u}^{k}$ consisting of each path $P$ in $\MMM$ such that all colors on $P$ appear on vertices in $M_u$ (these colors could still appear on vertices in $M_v$ as well); (3) the set of paths $\MMM_{v}^{k}$ consisting of each path $P$ in $\MMM$ such that all colors on $P$ appear on vertices in $M_v$; and (4) the set $\MMM^{< k}$ of remaining paths in $\MMM$, satisfying that each path contains a nonempty external vertex to $w$ in $M$ and contains less than $k$ colors from each of $M_u$ and $M_v$. Note that by Observation~\ref{obs:externalcolor}, each path in $\MMM$ belongs to one of the 4 groups above.

Since the paths in $\MMM$ are pairwise color-disjoint, no nonempty vertex in $S$ can appear on two distinct paths from group (1). Therefore, the number of paths in group (1) is at most $|S| \leq 2k+3$.
Observe, that the vertices in $S$ contained in any path in groups (2)-(4) are empty vertices.

To upper bound the number of paths in group (2), for each path $P$, there is a last vertex $x_P$ (\ie, farthest from $u$) in $P$ that is in $S$. Fix a vertex $x \in S$, and let us upper bound the number of paths $P$ in group (2) for which $x=x_P$. Let $P_v$ be the subpath of $P$ from $x$ to $v$. Note that since $v$ is empty and all the vertices in $S$ that are contained in paths in group (2) are empty, and since $M$ is irreducible, $P_v$ must contain at least one color. Since all colors appearing on $P$ appear on vertices in $M_u$, all colors appearing on $P_v$ appear in $M_u$. By Lemma~\ref{lem:contraction}, $M_u$ is contained in a single face $f$ of $M_v+x$.
Since $f$ is a vertex-separator that separates $V(M_u)$ from $V(P_v)$ in $G$, by Observation~\ref{obs:colorconnectivity}, every color that appears on $P_v$ appears on $f$. By Lemma~\ref{lem:externalpaths}, there are at most two $x$-$v$ paths that contain only colors that appear on $f$. This shows that there are at most two paths in group (2) for which $x$ is the last vertex in $S$. Since $|S| \leq 2k+3$, this upper bounds the number of paths in group (2) by $2(2k+3)=4k+6$. By symmetry, the number of paths in group (3) is upper bounded by $4k+6$.

Finally, we upper bound the number of paths in group (4). Let $S=\{s_2, \ldots, s_{r-1}\}$, where $r \leq 2k+5$, and extend $S$ by adding the two vertices $s_1=u$ and $s_r=v$ to form the set
$A=\{s_1, s_2, \ldots, s_{r}\}$. For every two (distinct) vertices $s_j, s_{j'} \in A$, $j, j' \in [r], j < j'$, we define a set of paths $\PPP_{jj'}$ in $G-w$ whose endpoints are $s_j$ and $s_{j'}$ as follows.
For each path $P$ in group (4), partition (the edges in) $P$ into subpaths $P_1, \ldots, P_q$ satisfying the property that the endpoints of each $P_i$, $i \in [q]$, are in $A$, and no internal vertex to $P_i$ is in $A$.
Since each $P$ is a $u$-$v$ path, clearly, $P$ can be partitioned as such. For each $P_i$, $i \in [q]$, such that $P_i$ contains a vertex that contains an external color to $w$ in $G-w$, let $P'_i$ (possibly equal to $P_i$) be a subpath in $G-w$ between the endpoints of $P_i$ satisfying that $\chi(P'_i) \subseteq \chi(P_i)$ and $\chi(P'_i)$ is minimal w.r.t.~containment (\ie, there does not exist a path $P''_i$ in $G-w$ between the endpoints of $P_i$ satisfying $\chi(P''_i) \subsetneq \chi(P'_i)$). Since $P$ contains a vertex that contains an external color to $w$ in $G-w$, there exists an $i \in [q]$ such that $P'_i$ contains a vertex that contains an external color to $w$ in $G-w$; otherwise, by concatenating (in the right sequence)
the $P_i$'s that do not contain an external color to $w$ (in $G-w$), with
the $P'_i$'s (instead of the $P_i$'s) of the $P_i$'s that contain an external color to $w$ (in $G-w$),
we would obtain a $u$-$v$ path $P'$ in $G-w$ satisfying $\chi(P') \subsetneq \chi(P)$ (since $\chi(P') \subseteq \chi(P)$ and $P$ contains an external color to $w$ and $P'$ does not), thus contradicting the minimality of $\MMM$. Pick any $i \in [q]$ satisfying that $P'_i$ contains a vertex that contains an external color to $w$ in $G-w$, associate $P$ with $P'_i$, and assign $P'_i$ to the set of paths $\PPP_{jj'}$ such that $s_j$ and $s_{j'}$ are the endpoints of $P'_i$. Since each $P'_i$ contains an external color that appears on $P$ and the paths in $\MMM$ are pairwise-color disjoint, it follows that the map that maps each $P$ to its $P'_i$ is a bijection.

Therefore, to upper bound the number of paths in group (4), it suffices to upper bound the number of paths assigned to the sets $\PPP_{jj'}$, where $j, j' \in [r], j < j'$. Fix a set $\PPP_{jj'}$. The paths in $\PPP_{jj'}$ have $s_j, s_{j'}$ as endpoints, and are pairwise color-disjoint. Moreover, each path in $\PPP_{jj'}$ contains a vertex that contains an external color to $w$ in $G-w$. It follows from the previous statements that $\PPP_{jj'}$ satisfies properties (i) and (ii) of Definition~\ref{def:minimalpaths} w.r.t.~$G$ and $w$. Moreover, from the definition of each path in $\PPP_{jj'}$, $\PPP_{jj'}$ satisfies property (iii) of Definition~\ref{def:minimalpaths} as well. Finally, observe that each path $P'_i \in \PPP_{jj'}$ was constructed based on a subpath $P_i$ of a path $P$ in group 4, and satisfying that $P_i$ has endpoints $s_j, s _{j'}$ and no internal vertex on $P_i$ is in $A$. Since $P$ is a $u$-$v$ path in $\MMM$ and $S$ is a vertex-separator of $M$, $V(P_i)$ is either contained in $V(M_u) \cup S$ or in $V(M_v) \cup S$.
Since $P$ is in group (4), $P$ contains at most $k-1$ colors from each of $M_u$ and $M_v$. Since the vertices in $S$ are empty, we deduce that $P_i$ contains at most $k-1$ colors. Since $\chi(P'_i) \subseteq \chi(P_i)$, $P'_i$ contains at most $k-1$ colors as well, and hence, every path in $\PPP_{jj'}$ contains at most $k-1$ colors. It follows that $\PPP_{jj'}$ is a minimal set of $(k-1)$-valid $s_j$-$s_{j'}$ paths in $G-w$ w.r.t.~$w$. By the inductive hypothesis, we have $|\PPP_{jj'}| \leq g(k-1)$. Since the number of sets $\PPP_{jj'}$ is at most $\binom{2k+5}{2}$, the number of paths in group (4) is $\Oh(k^2) \cdot g(k-1)$.

It follows from the above that $|\MMM| \leq g(k)$, where $g(k)$ satisfies the recurrence relation $g(k) \leq 3 + (2k+3) + 2(4k+6) + \Oh(k^2) \cdot g(k-1) = \Oh(k^2) \cdot g(k-1)$, where $3$ accounts for the $3$ paths we removed from $\MMM$ at the beginning of the proof to get a small $u$-$v$ vertex-separator. Solving the aforementioned recurrence relation gives $g(k) =\Oh(c^k k^{2k})$, where $c > 1$ is a constant. Adding 1 to $g(k)$ to account for the single path in $\MMM$ containing only internal colors w.r.t.~$w$ in $M$ yields the same asymptotic upper bound.
\end{proof}

\begin{theorem}
\label{thm:main}
Let $G$ be a plane color-connected graph, let $u, v, w \in V(G)$, and let ${\cal P}$ be a set of minimal $k$-valid $u$-$v$ paths w.r.t.~$w$ in $G-w$. Then $|\PPP|\leq h(k)$, where $h(k)=\Oh(c^{k^2} k^{2k^2+k})$, for some constant $c > 1$.
\end{theorem}

\begin{proof}
The proof is by induction on $k$. If $k=1$, then by minimality of $\PPP$, we have $\PPP=\MMM$. Lemma~\ref{lem:aux} gives an upper bound of $\Oh(c^k k^{2k})=\Oh(c^{k^2} k^{2k^2+k})$ on $|\PPP|$.

Assume by the inductive hypothesis that the statement of the lemma is true for $1 \leq i < k$. Let $\MMM$ be a maximal set of pairwise color-disjoint paths in $\PPP$. By Lemma~\ref{lem:aux}, $|\MMM| \leq g(k)=\Oh(c^k k^{2k})$. The number of colors contained in vertices of $\MMM$ is at most $r \leq k\cdot g(k)$. We group the paths in $\PPP$ into $r$ groups $\PPP_1, \ldots, \PPP_r$, such that all the paths in $\PPP_i$, $i \in [r]$, share the same color $c_i$, where $i \in [r]$, that is distinct from each color $c_j$ shared by the paths in $\PPP_j$, for $j \neq i$. We upper bound the number of paths in each $\PPP_i$, $i \in [r]$, to obtain an upper bound on $|\PPP|$.

Let $G_i$ be the graph obtained by removing color $c_i$ from each vertex in 
$G$ that $c_i$ appears on, and let $\PPP'_i$ be the set of paths obtained from $\PPP_i$ by removing color $c_i$ from each vertex in $\PPP_i$ that $c_i$ appears on. Clearly, every path in $\PPP'_i$ is a $(k-1)$-valid $u$-$v$ path. Moreover, it is easy to verify that $\PPP'_i$ satisfies properties (i)-(iii) in Definition~\ref{def:minimalpaths}, and hence, $\PPP'_i$ is minimal w.r.t.~$w$ in 
$G_i-w$. By the inductive hypothesis, we have $|\PPP'_i| \leq h(k-1)$.
It follows that the total number of paths in $\PPP$ is at most $h(k)$, where $h(k)$ satisfies the recurrence relation $h(k) \leq r \cdot h(k-1) \leq k \cdot g(k) \cdot h(k-1)$. Solving the aforementioned recurrence relations yields $h(k)=\Oh((k\cdot g(k))^k)=\Oh(c^{k^2} k^{2k^2+k})$.
\end{proof}

The result of Theorem~\ref{thm:main} will be employed in the next section in the form presented in the following corollary:

\begin{corollary}
\label{cor:main}
Let $G$ be a plane color-connected graph, and let $w \in V(G)$. Let $G'$ be a subgraph of $G-w$, and let $u, v \in V(G')$. Every set ${\cal P}$ of minimal $k$-valid $u$-$v$ paths in $G'$ w.r.t.~$w$ satisfies $|\PPP|\leq h(k)$, where $h(k)=\Oh(c^{k^2} k^{2k^2+k})$, for some constant $c > 1$.
\end{corollary}

\begin{proof}
Contract every connected component of $(G-w)-G'$ into a single vertex containing the union of the color-sets of the vertices in the component, and add $k+1$ new distinct colors to the resulting vertex. Denote the resulting graph by $G''$. Observe that the resulting graph is color-connected, and that every $k$-valid $u$-$v$ path in $G'$ w.r.t.~$w$ is a $k$-valid $u$-$v$ path in $G''$ w.r.t.~$w$, and vice versa. Therefore, every set ${\cal P}$ of minimal $k$-valid $u$-$v$ paths in $G'$ w.r.t.~$w$ is also a set of minimal $k$-valid $u$-$v$ paths in $G''$ w.r.t.~$w$. For any set ${\cal P}$ of minimal $k$-valid $u$-$v$ paths w.r.t.~$w$ in $G'$, by applying Theorem~\ref{thm:main} to ${\cal P}$ in $G''-w$, the corollary follows.
\end{proof}
\fi
\section{The Algorithm}
\label{sec:algo}
In this section, we \iflong present an~\fi \ifshort highlight how the \fi \FPT{} algorithm for \cmor{}, parameterized by both $k$ and the treewidth of the input graph\ifshort~works\fi. As pointed out in Section~\ref{sec:structural}, there can be too many (\ie, more than \FPT-many) subsets of colors that appear in a bag, and hence, that the algorithm may need to store/remember. To overcome this issue, we extend the notion of a minimal set of $k$-valid $u$-$v$ paths w.r.t.~a vertex---from the previous section---to a ``representative set'' of paths w.r.t.~a specific bag and a specific enumerated configuration for the bag. This allows us to upper bound the size of the table, in the dynamic programming algorithm, stored at a bag by a function of both $k$ and the treewidth of the graph.

Let $(G, C, \chi, s, t, k)$ be an instance of \cmor. \iflong The algorithm is a dynamic programming algorithm based on a tree decomposition of $G$. \fi Let $({\cal V}, {\cal T})$ be a nice tree decomposition of $G$. By Assumption~\ref{ass:sat}, we can assume that $s$ and $t$ are nonadjacent empty vertices. We add $s$ and $t$ to every bag in $\TTT$, and now we have $\{s,t\}\subseteq X_i$, for every bag $X_i \in \TTT$.
For a bag $X_i$, we say that $v\in X_i$ is \emph{useful} if $|\Col(v)|\le k$. Let $U_i$ be the set of all useful vertices in $X_i$ and let $\overline{U_i}=X_i\setminus U_i$.
We denote by $V_i$ the set of vertices in the bags of the subtree of $\mathcal{T}$ rooted at $X_i$. For any two vertices $u, v \in X_i$, let $G_{uv}^i=G[(V_i\setminus X_i) \cup \{u, v\}]$. We extend the notion of a minimal set of $k$-valid $u$-$v$ paths w.r.t.~a vertex, developed in the previous section, to the set of vertices in a bag of $\TTT$.

\begin{definition}
\label{def:minimalpathbag}
A set of $k$-valid $u$-$v$ paths $\PPP_{uv}$ in $G_{uv}^i$ is {\em minimal} w.r.t.~$X_i$ if it satisfies the following properties:

\begin{itemize}
\item[(i)] There does not exist two paths $P_1, P_2 \in \PPP_{uv}$ such that $\chi(P_1) \cap \chi(X_i) = \chi(P_2) \cap \chi(X_i)$;	
\item[(ii)] there does not exist two paths $P_1, P_2 \in \PPP_{uv}$ such that $\chi(P_1) \subseteq \chi(P_2)$; and
\item[(iii)] for any $P \in \PPP_{uv}$ there does not exist a $u$-$v$ path $P'$ in $G_{uv}^i$ such that $\chi(P') \subsetneq \chi(P)$.
\end{itemize}
\end{definition}

The following lemma uses the upper bound on the cardinality of a minimal set of $k$-valid $u$-$v$ paths w.r.t.~a vertex, derived in \iflong Corollary~\ref{cor:main}\fi \ifshort Theorem~\ref{thm:main}\fi~in the previous section, to obtain an upper bound on the cardinality of a minimal set of $k$-valid $u$-$v$ paths w.r.t.~a bag of $\TTT$:
\ifshort
\begin{lemma} [\appno{} Lemma~5.2]
	\label{lem:one_path_bound}
	Let $X_i$ be bag, $u, v \in X_i$, and $\PPP_{uv}$ a set of $k$-valid $u$-$v$ paths in $G_{uv}^i$ that is minimal w.r.t.~$X_i$. Then the number of paths in $\PPP_{uv}$ is at most $h(k)^{|X_i|}$, where $h(k) =\Oh(c^{k^2} k^{2k^2+k})$, for some constant $c >1$.
\end{lemma}
\fi

\iflong

\begin{lemma}
	\label{lem:one_path_bound}
	Let $X_i$ be bag, $u, v \in X_i$, and $\PPP_{uv}$ a set of $k$-valid $u$-$v$ paths in $G_{uv}^i$ that is minimal w.r.t.~$X_i$. Then the number of paths in $\PPP_{uv}$ is at most $h(k)^{|X_i|}$, where $h(k) =\Oh(c^{k^2} k^{2k^2+k})$, for some constant $c >1$.
\end{lemma}

\begin{proof}
	Let $X_i\setminus \{u, v\}=\{w_1, \ldots, w_r\}$, where $r=|X_i|-2$. For each $w_j \in X_i$, $j \in [r]$, let $\PPP_j$ be a minimal set of $k$-valid $u$-$v$ paths w.r.t.~$w_j$ in $G_{uv}^i$. Without loss of generality, we can pick $\PPP_j$ such that there is no $k$-valid $u$-$v$ path $P$ in $G_{uv}^i$ such that $\PPP_j \cup \{P\}$ is minimal. From \iflong Corollary~\ref{cor:main}\fi \ifshort Theorem~\ref{thm:main}\fi, we have $|\PPP_j| \leq h(k) =\Oh(c^{k^2} k^{2k^2+k})$, for some constant $c >1$. For each $P \in \PPP_{uv}$, and each $j \in [r]$, define $C_j = \chi(P) \cap \chi(w_j)$. Define the \emph{signature} of $P$ (w.r.t.~the colors of $w_1, \ldots, w_r$) to be the tuple $(C_1, \ldots, C_r)$.
	Observe that no two (distinct) paths $P_1, P_2 \in \PPP_{uv}$ have the same signature; otherwise, since $u$ and $v$ appear on both $P_1, P_2$, $\chi(P_1) \cap \chi(X_i)= \chi(P_2) \cap \chi(X_i)$, which contradicts condition (i) of the minimality of $\PPP_{uv}$. For each $P \in \PPP_{uv}$, and each $j \in [r]$, there is a path $P' \in \PPP_{j}$ such that $\chi(P') \cap \chi(w_j) =C_j$. If this were not true, then $P$ would have been added to $\PPP_{j}$ for the following reasons. Clearly, $P$ does not contradict conditions (i) and (iii) of the minimality of $\PPP_{j}$. It cannot contradict (ii) either, because otherwise, and since $P$ does not contradict (i), there would be a path $P'' \in \PPP_{j}$ such that $\chi(P'') \subsetneq \chi(P)$, contradicting the minimality of $\PPP_{uv}$. It follows that the number of signatures of paths in $\PPP_{uv}$ is at most
	$\prod_{j=1}^{r}|\PPP_j| \leq h(k)^{|X_i|}$. Since no two distinct paths in $\PPP_{uv}$ have the same signature, it follows that $|\PPP_{uv}| \leq h(k)^{|X_i|}$.
\end{proof}
\fi

\begin{definition}\rm
\label{def:pattern}
Let $X_i$ be a bag in $\TTT$. A {\em pattern} $\pi$ for $X_i$ is a sequence \\ $(v_1=s, \sigma_1, v_2, \sigma_2, \dots, \sigma_{r-1}, v_r=t)$, where $\sigma_i\in \{0,1\}$ and $v_i\in U_i$.
For a bag $X_i$, and a pattern $(v_1=s, \sigma_1, v_2, \sigma_2, \dots, \sigma_{r-1}, v_r=t)$ for $X_i$, we say that a sequence of paths $\mathcal{S} = (P_1,\dots, P_{r-1})$
\emph{conforms} to $(X_i, \pi)$ if:
\begin{itemize}
 	\item  For each $j \in [r-1]$, $\sigma_j=1$ implies that $P_j$ is an induced path from $v_j$ to $v_{j+1}$ whose internal vertices are contained in $V_i\setminus X_i$ and $P_j$ is empty otherwise; and
 	\item $|\Col(\mathcal{S})|=|\bigcup_{j \in [r-1]}\chi(P_j)|\le k$.
 \end{itemize}
\end{definition}

\begin{definition} \rm
\label{def:preceq}
Let $X_i$ be a bag, $\pi$ a pattern for $X_i$, and $\SSS_1, \SSS_2$ two sequences of paths that conform to $(X_i, \pi)$. We write $\SSS_1 \preceq_i \SSS_2$ if $|\chi(\SSS_1)\cup(\chi(\SSS_2)\cap \chi(X_i))|\le |\chi(\SSS_2)|$.
\end{definition}

\iflong
We note that at a certain point during the dynamic programming algorithm, we will have to deal for a short while with sequences of walks instead of sequences of paths (until we refine them), but the definition of
a sequence of paths conforming to a bag and a pattern, and the relation $\preceq$, extend seamlessly to sequences of walks.
\fi

\iflong
\begin{lemma}
\label{lem:transitive}
Let $X_i$ be a bag and $\pi$ a pattern for $X_i$. The relation $\preceq_i$ is a transitive relation on the set of all sequences of paths that conform to $(X_i, \pi)$.
\end{lemma}

\begin{proof}
Let $\SSS_1, \SSS_2, \SSS_3$ be three sequences that conform to $(X_i, \pi)$. Suppose that $\SSS_1 \preceq_i \SSS_2$ and $\SSS_2 \preceq_i \SSS_3$. We need to show that $\SSS_1 \preceq_i \SSS_3$.
To simplify the notation in the proof, let $A=\chi(S_1), B=\chi(S_2), C=\chi(S_3), X=X_i$. Since $\SSS_1 \preceq_i \SSS_2$, we have

 \begin{eqnarray}
 |A \cup B \cap X| & \leq & |B|\nonumber \\
  |A| + |B\cap X| -  |A \cap B \cap X| & \leq & |B|, \label{eq1}
 \end{eqnarray}

 and since $\SSS_2 \preceq_i \SSS_3$ we have:
 \begin{eqnarray}
  |B \cup C \cap X| & \leq & |C| \nonumber \\
  |B| + |C\cap X| -  |B \cap C \cap X| & \leq & |C|. \label{eq2}
 \end{eqnarray}

 From Inequalities (\ref{eq1}) and   (\ref{eq2}) we get:
{\small
 \begin{eqnarray}
 |A| + |C\cap X| + |B\cap X| -  |A \cap B \cap X| -  |B \cap C \cap X| & \leq & |C| \nonumber \\
 |A| + |C\cap X| + |B\cap X| -  (|A \cap B \cap X| +  |B \cap C \cap X|  - |A \cap B \cap C \cap X| + |A \cap B \cap C \cap X|) & \leq & |C| \nonumber \\
 |A| + |C\cap X| + |B\cap X| -  (|A \cap B \cap X \cup  B \cap C \cap X| + |A \cap B \cap C \cap X|) & \leq & |C| \nonumber \\
  |A| + |C\cap X| + |B\cap X| -  (|(A \cup C)\cap (B \cap X)| + |A \cap B \cap C \cap X|) & \leq &  |C| \nonumber \\
 |A| + |C\cap X| + |B\cap X| -  (|B \cap X| + |A \cap B \cap C \cap X|) & \leq & |C| \nonumber \\
  |A| + |C\cap X| -  |A \cap B \cap C \cap X|) & \leq & |C| \nonumber \\
 |A| + |C\cap X| -  |A  \cap C \cap X|) & \leq & |C|. \nonumber
  \end{eqnarray}
}

The last inequality proves that $\SSS_1 \preceq_i \SSS_3$.

\end{proof}
\fi

Using the relation $\preceq_i$ on the set of sequences that conform to $(X_i, \pi)$, we \iflong are now ready to \fi \ifshort can \fi define the key notion of representative sets that makes the dynamic programming approach work:

\begin{definition}\rm
\label{def:representativeset}
Let $X_i$ be a bag and $\pi=(v_1, \sigma_1, v_2 \ldots, \sigma_{r-1}, v_r)$ a pattern for $X_i$. A set $\RRR_\pi$ of sequences that conform to $(X_i, \pi)$ is a {\em representative} set for $(X_i, \pi)$ if:
\begin{itemize}
	\item[(i)]  For every sequence $\SSS_1\in \RRR_\pi$, and for every sequence $\SSS_2\neq \SSS_1$ that conforms to $(X_i, \pi)$, if $\SSS_1 \preceq_i \SSS_2$ then $\SSS_2 \notin \RRR_{\pi}$;
	\item[(ii)] for every sequence $\SSS\in\RRR_{\pi}$, and for every path $P \in \SSS$ between $v_j$ and $v_{j+1}$, $j\in [r-1]$, there does not exist a $v_j$-$v_{j+1}$ path $P'$ in $G_{v_jv_{j+1}}^i$ such that $\chi(P') \subsetneq \chi(P)$; and
	\item[(iii)] for every sequence $\SSS\notin\RRR_{\pi}$ that conforms to $(X_i, \pi)$ and satisfies that no two paths in $\SSS$ share a vertex that is not in $X_i$, there is a sequence $\WWW \in \RRR_{\pi}$ such that $\WWW \preceq_i \SSS$.
\end{itemize}
\end{definition}

\iflong
\begin{observation}
\label{obs:jtoi}
Let $X_i$ and $X_j$ be two bags such that $X_i \subseteq X_j$, let $\pi$ be a pattern for both $X_i$ and $X_j$, and let $\SSS, \SSS'$ be two sequences that conform to both $(X_i, \pi)$ and $(X_j, \pi)$. If $\SSS \preceq_j \SSS'$ then $\SSS \preceq_i \SSS'$.
\end{observation}

  \begin{proof}
  Since $X_i \subseteq X_j$, we have $|\chi(\SSS) \cup \chi(\SSS') \cap \chi(X_i)| \leq |\chi(\SSS) \cup \chi(\SSS') \cap \chi(X_j)|$.
  \end{proof}
\fi

\iflong
\begin{lemma}
\label{lem:join}
Let $X_i$ be a bag, $\pi$ a pattern for $X_i$, and $\SSS_1, \SSS'_1, \SSS_2, \SSS'_2, \SSS, \SSS'$ sequences that conform to $(X_i, \pi)$ and that satisfy the following: $\SSS'_1 \preceq_i \SSS_1$, $\SSS'_2 \preceq_i \SSS_2$,
$\chi(\SSS_1) \cup \chi(\SSS_2) = \chi(\SSS)$, $\chi(\SSS'_1) \cup \chi(\SSS'_2) = \chi(\SSS')$, and $\chi(\SSS_1) \cap \chi(\SSS_2) \subseteq \chi(X_i)$. Then $\SSS' \preceq_i \SSS$.
\end{lemma}

\begin{proof}
Let $A= \chi(\SSS_1), B= \chi(\SSS_2), C= \chi(\SSS), A'= \chi(\SSS'_1), B'= \chi(\SSS'_2), C'= \chi(\SSS')$, and $X=\chi(X_i)$.
Since $\SSS'_1 \preceq_i \SSS_1$ we have:

 \begin{eqnarray}
  |A' \cup A \cap X| & \leq & |A| \nonumber \\
  |A'| + |A\cap X| -  |A' \cap A \cap X| & \leq & |A|. \label{eq10}
 \end{eqnarray}

 Since $\SSS'_2 \preceq_i \SSS_2$ we have:

 \begin{eqnarray}
  |B' \cup B \cap X| & \leq & |B| \nonumber \\
  |B'| + |B\cap X| -  |B' \cap B \cap X| & \leq & |B|. \label{eq11}
 \end{eqnarray}

 Adding Inequality (\ref{eq10}) to (\ref{eq11}) and subtracting $|A\cap B|$ from each side of the resulting inequality, we obtain:

 \begin{eqnarray}
  |A'| + |B'| + |A \cap X| + |B \cap X| - |A' \cap A \cap X| - |B' \cap B \cap X| - |A \cap B| \leq |A \cup B|. \label{eq12}
 \end{eqnarray}

 Replacing in the last Inequality (\ref{eq12}) $|A'| + |B'|$ by $|A' \cup B'| + |A' \cap B'|$, and $|A \cap X| + |B \cap X|$ by $|(A \cup B) \cap X| + |A \cap B \cap X|$, observing that $A \cap B \cap X = A \cap B$ (because $A \cap B \subseteq X$), and simplifying, we get:
{
\footnotesize
  \begin{eqnarray*}
   |A' \cup B'| + |(A \cup B) \cap X| + |A' \cap B'| -|A' \cap A \cap X| - |B' \cap B \cap X|  \leq  |A \cup B| ~ \\
   |(A' \cup B') \cup (A \cup B) \cap X| + |(A' \cup B') \cap (A \cup B) \cap X| + |A' \cap B'| -|A' \cap A \cap X| - |B' \cap B \cap X|  \leq  |A \cup B|. \nonumber
 \end{eqnarray*}
}

 Replacing $-|A' \cap A \cap X| - |B' \cap B \cap X|$ in the last inequality with $-(|A' \cap A \cup B' \cap B) \cap X| + |A' \cap A \cap B' \cap B \cap X|)=-(|A' \cap A \cup B' \cap B) \cap X| + |A' \cap A \cap B' \cap B|)$ (because $A \cap B \subseteq X$), and observing that $|(A' \cap A \cup B' \cap B) \cap X| \leq |(A' \cup B') \cap (A \cup B) \cap X|$, and $|A' \cap A \cap B' \cap B| \leq |A' \cap B'|$, we conclude that:

\begin{eqnarray}
   |(A' \cup B') \cup (A \cup B) \cap X|  & \leq & |A \cup B|. \label{eq14}
\end{eqnarray}

 Inequality (\ref{eq14}) establishes that $\SSS' \preceq_i \SSS$.
\end{proof}
\fi

\ifshort
The following lemma uses the upper bound on the cardinality of a minimal set of $k$-valid $u$-$v$ paths w.r.t.~a bag $X_i$, derived in Lemma~\ref{lem:one_path_bound}, to obtain an upper bound on the cardinality of a representative set w.r.t.~a bag and a fixed pattern $(X_i, \pi)$:

\begin{lemma}[\appno{} Lemma~5.9]
	\label{lem:representativesbound}
	Let $X_i$ be bag, $\pi$ a pattern for $X_i$, and $\RRR_{\pi}$ be a representative set for $(X_i, \pi)$. Then the number of sequences in $\RRR_{\pi}$ is at most $h(k)^{|X_i|^2}$, where $h(k) =\Oh(c^{k^2} k^{2k^2+k})$, for some constant $c >1$.
\end{lemma}

For each bag $X_i$, we maintain a table $\Gamma_i$ that contains, for each pattern for $X_i$, a representative set of sequences $\RRR_{\pi}$ for $(X_i, \pi)$. The rest is a technical dynamic programming algorithm over $(\VVV, \TTT)$ that computes the table $\Gamma_i$ at a bag $X_i$ for each bag type (introduce, forget, join) in the nice tree decomposition. The algorithm and its analysis are given in the full paper \app. We conclude with the following theorem:

\begin{theorem}[\appno{} Theorem~5.12]
\label{thm:treewidth}
There is an algorithm that on input $(G, C, \chi, s, t, k)$ of \cmor{}, either outputs a $k$-valid $s$-$t$ path in $G$ or decides that no such path exists, in time $\Oh^{\star}(f(k)^{6\omega^2})$, where $\omega$ is the treewidth of $G$ and $f(k) =\Oh(c^{k^2} k^{2k^2+k})$, for some constant $c >1$. Therefore, \cmor{} parameterized by both $k$ and the treewidth $\omega$ of the input graph is \FPT.	
\end{theorem}

\section{Extensions and Applications}
\label{sec:extension}
In this section, we explain how to extend the \FPT{} result for \cmor{} w.r.t.~the parameterization by both $k$ and the treewidth of the graph, to the parameterization by both $k$ and the length $\ell$ of the sought path, and discuss important applications of this extended result. We formally define the problem w.r.t.~the parameterization by $k$ and $\ell$:

\paramproblem{\bcmor} {A planar graph $G$; a set of colors $C$; $\chi: V \longrightarrow 2^{C}$; two designated vertices $s, t \in V(G)$; and $k, \ell \in \nat$}{Does there exist a $k$-valid $s$-$t$ path of length at most $\ell$ in $G$?} \\

To extend Theorem~\ref{thm:treewidth}, we repeatedly contract every edge $uv$ incident to a vertex $v$ whose distance to $s$ is more than $\ell +1$; we assign the resulting vertex the color set $\chi(u) \cup \chi(v)$. (We do not delete such vertices in order to preserve the color-connectivity property.) Afterwards, the radius of $G$ is at most $\ell+1$, and hence $G$ has treewidth at most $3\cdot(\ell+1)+1=3\ell+4$~\cite{rs} (\appno{} Lemma~6.3).
Although the treewidth of $G$ is bounded by a function of $\ell$, we cannot use the \FPT{} algorithm for \cmor{}, parameterized by $k$ and the treewidth of $G$, to solve \bcmor{} because a $k$-valid path returned by the algorithm for \cmor{} may have length more than $\ell$. In fact, extending the \FPT{} results for \cmor{} to \bcmor{} turns out to be a nontrivial task. We state the result below, and refer to \app{} Section~6 for details.

\begin{theorem}[\appno{} Theorem~6.15]
\label{thm:ext_treewidth}
\bcmor{} parameterized by both $k$ and the length of the path is \FPT.	
\end{theorem}

We now describe applications of Theorem~\ref{thm:ext_treewidth}. The first application is a direct consequence of this theorem.

\begin{corollary}
\label{cor:boundedlengthaapp}
For any computable function $h$, the restriction of \cmor{} to instances in which the length of the path is at most $h(k)$ is \FPT{} parameterized by $k$.
\end{corollary}

We note that the above restriction of \cmor{} is \NP-hard, as a consequence of (the proof of) \app{} Corollary~3.3.

Corollary~\ref{cor:boundedlengthaapp} directly implies Kroman \etal's results~\cite{korman}, showing that \gmor{} is \FPT{} for unit-disk obstacles and for similar-size fat-region obstacles with constant overlapping number. Using Bereg and Kirkpatrick's result~\cite{kirkpatrick3}, the length of a shortest $k$-valid path for unit-disk obstacles is at most $3k$ (see also Lemma 3 in Korman \etal~\cite{korman}). By Corollary~2 in~\cite{korman},
the length of a shortest $k$-valid path for similar-size fat-region obstacles with constant overlapping number is linear in $k$. Corollary~\ref{cor:boundedlengthaapp} generalizes these \FPT{} results, which required quite some effort, and provides an explanation to why the problem is \FPT{} for such restrictions, namely because the path length is upper bounded by a function of $k$.

The second application is related to an open question posed in~\cite{lavalle,hauser}. For an instance $I=(G, C, \chi, s, t, k)$ of \cmor{}, and a color $c \in C$, define the \emph{intersection number} of $c$, denoted $\iota(c)$, to be the number of vertices in $G$ on which $c$ appears. Define the \emph{intersection number} of $G$, $\iota(G)$, as $\max\{\iota(c) \mid c \in C\}$.

\begin{corollary}[\appno{} Corollary~6.18]
\label{cor:boundedintersection}
For any computable function $h$, \cmor{} restricted to instances $(G, C, \chi, s, t, k)$ satisfying $\iota(G) \leq h(k)$ is \FPT{} parameterized by $k$.
\end{corollary}

Corollary~\ref{cor:boundedintersection} has applications pertaining to instances of \gcmor{} whose auxiliary graph has intersection number bounded by a function of $k$. An interesting case that was studied is when the obstacles are convex polygons, each intersecting at most a constant number of other polygons. The complexity of this problem was posed as an open question in~\cite{lavalle,hauser}, and remains unresolved. Corollary~\ref{cor:boundedintersection} implies that the problem is \FPT{}, even for the more general setting in which the obstacles are arbitrary convex regions satisfying that the number of regions intersected by any region is a constant (\appno{} Theorem~6.19). (Note that convexity is important here.)

%
\fi

\iflong
\begin{lemma}
	\label{lem:representativesbound}
	Let $X_i$ be bag, $\pi$ a pattern for $X_i$, and $\RRR_{\pi}$ be a representative set for $(X_i, \pi)$. Then the number of sequences in $\RRR_{\pi}$ is at most $h(k)^{|X_i|^2}$, where $h(k) =\Oh(c^{k^2} k^{2k^2+k})$, for some constant $c >1$.
\end{lemma}

\begin{proof}
	Let $\pi =(v_1=s, \sigma_1, v_2, \sigma_2, \dots, \sigma_{r-1}, v_r=t)$ and let $v_j$ and $v_{j+1}$ be two consecutive vertices in $\pi$ such that $\sigma_j=1$. For each $j\in [r-1]$ such that $\sigma_j=1$, let $\PPP_j$ be a minimal set of $k$-valid $v_j$-$v_{j+1}$ paths w.r.t.~$X_i$. Without loss of generality, we can pick $\PPP_j$ such that there is no $k$-valid $u$-$v$ path $P$ in $G_{v_jv_{j+1}}^i$ such that $\PPP_j \cup \{P\}$ is minimal w.r.t.~$X_i$. From Lemma~\ref{lem:one_path_bound}, it follows that $|\PPP_j|\le h(k)^{|X_i|}$, where $h(k) =\Oh(c^{k^2} k^{2k^2+k})$, for some constant $c >1$.
	For a sequence $\SSS= (P_1,\ldots, P_{r-1})$ in $\RRR_{\pi}$, we define the \emph{signature} of $\SSS$ (w.r.t.~$X_i$) to be the tuple $(\chi(P_1)\cap\chi(X_i), \ldots, \chi(P_{r-1})\cap\chi(X_i))$. Observe that if $\SSS_1$ and $\SSS_2$ have the same signature w.r.t.~$X_i$, then $\chi(\SSS_1)\cup (\chi(\SSS_2)\cap \chi(X_i))=\chi(\SSS_1)$ and $\chi(\SSS_2)\cup (\chi(\SSS_1)\cap \chi(X_i))=\chi(\SSS_2)$; hence, either $\SSS_1\preceq_i \SSS_2$ or $\SSS_2\preceq_i \SSS_1$. It follows from property (i) of representative sets that no two sequences in $\RRR_{\pi}$ have the same signature w.r.t.~$X_i$. Now let $\SSS = (P_1,\ldots, P_{r-1})$ be a sequence in $\RRR_{\pi}$ with a signature $(C_1,\ldots, C_{r-1})$. Note that if $C_j\neq\emptyset$, then $P_j$ is not the empty path, and hence $\sigma_j=1$. We show that for each $j\in [r-1]$ such that $C_j\neq\emptyset$, there is a path $P\in \PPP_j$ such that $\chi(P)\cap\chi(X_i) = C_j$. Suppose, for a contradiction, that this is not the case. Then for some $j\in [r-1]$ such that $C_j\neq \emptyset$, there is no path  $P\in \PPP_j$ such that $\chi(P)\cap\chi(X_i) = C_j$.  Clearly, $P_j\notin\PPP_j$, and therefore, by our choice of $\PPP_j$, the set $\PPP_j\cup \{P_j\}$ is not a minimal set w.r.t.~$X_i$. By assumption, $\PPP_j\cup \{P_j\}$ does not contradict property (i) in the definition of minimal set of paths w.r.t.~$X_i$. Moreover, since $\SSS\in \RRR_\pi$, it follows from property (ii) of representative sets that $P_j$, and hence $\PPP_j\cup \{P_j\}$, satisfies property (iii) of minimal set of paths w.r.t.~$X_i$. Therefore, $\PPP_j\cup \{P_j\}$ has to contradict property (ii) in the definition of minimal set of paths w.r.t.~$X_i$, and there are two paths $Q_1,Q_2\in\PPP_j\cup \{P_j\}$ such that $\chi(Q_1)\subseteq \chi(Q_2)$. However, if $\chi(Q_1) = \chi(Q_2)$, then $Q_1$ and $Q_2$ contradict property (i) of a minimal set of paths w.r.t.~$X_i$, and if $\chi(Q_1)\subsetneq \chi(Q_2)$, then $Q_2$ contradicts property (iii), and we already established that $\PPP_j\cup \{P_j\}$ satisfies properties (i) and (iii). Therefore, $\PPP_j\cup \{P_j\}$ is a set of minimal paths w.r.t.~$X_i$, which is a contradiction. We conclude that, for each $j\in [r-1]$ such that $C_j\neq\emptyset$, there is a path $P\in \PPP_j$ such that $\chi(P)\cap\chi(X_i) = C_j$. It follows that the number of signatures of paths in $\PPP_{uv}$ is at most
	$\prod_{j=1}^{r-1}|\PPP_j| \leq h(k)^{|X_i|^2}$. Since no two distinct sequences in $\RRR_{\pi}$ have the same signature, it follows that $|\RRR_{\pi}| \leq h(k)^{|X_i|^2}$.
\end{proof}

For each bag $X_i$, we maintain a table $\Gamma_i$ that contains, for each
pattern for $X_i$, a representative set of sequences $\RRR_{\pi}$ for $(X_i, \pi)$.
For two vertices $u, v \in X_i$ and two $u$-$v$ paths $P, P'$ in $G_{uv}^{i}$, we say that $P'$ \emph{refines} $P$ if $\chi(P') \subseteq \chi(P)$. For two sequences $\SSS=(P_1, \ldots, P_{r-1})$ and $\SSS'=(P'_1, \ldots, P'_{r-1})$ that conform to $(X_i, \pi)$, we say that $\SSS'$ \emph{refines} $\SSS$ if each path $P'_j$ refines $P_j$, for $j \in [r-1]$.

\begin{lemma}
\label{lem:refinetime}
Let $X_i$ be a bag, $\pi =(v_1=s, \sigma_1, v_2, \sigma_2, \dots, \sigma_{r-1}, v_r=t)$ a pattern for $X_i$, and $\WWW=(W_1, \ldots, W_{r-1})$ a sequence of walks, where each $W_j$ is a walk between vertices $v_j$ and $v_{j+1}$ in $G_{v_jv_{j+1}}^{i}$ satisfying $\chi(W_j) \leq k$. Then in time $\Oh^{*}(2^k)$ we can compute a sequence $\SSS=(P_1, \ldots, P_{r-1})$ of induced paths, where each $P_j$ is an induced path between vertices $v_j$ and $v_{j+1}$ in $G_{v_jv_{j+1}}^{i}$ such that $\chi(P_j) \subseteq \chi(W_j)$, for $j \in [r-1]$, and such that $\SSS$ satisfies property (ii) of representative sets.
\end{lemma}

\begin{proof}
For each walk $W_j$, $j \in [r-1]$, we do the following. For each subset $C' \subseteq \chi(W)$ considered in a nondecreasing order of cardinality, we form the subgraph $G'$ from $G_{v_jv_{j+1}}^{i}$ by removing every vertex $x$ in $G_{v_jv_{j+1}}^{i}$ that does not satisfy $\chi(x) \subseteq C'$. We then check if there is a $v_j$-$v_{j+1}$ induced path in $G'$, and set $P_j$ to this path if it exists. It is clear that the path $P_j$ satisfies $\chi(P_j) \subseteq \chi(W_j)$ and that the sequence  $\SSS'=(P_1, \ldots, P_{r-1})$ conforms to $\pi$ w.r.t.~$X_i$ and satisfies property (ii) of representative sets. Since each $W_j$ satisfies $\chi(W_j) \leq k$, we can enumerate all subsets of $\chi(W_j)$ in time $\Oh^{*}(2^k)$. Since checking if there is an induced $v_j$-$v_{j+1}$ path in $G'$ takes polynomial time, it follows that computing $P_j$ from $W_j$ takes $\Oh^{*}(2^k)$,  and so does the computation of $\SSS$.
\end{proof}

For a bag $X_i$, pattern $\pi$ for $X_i$, and a set of sequences $\RRR$ that conform to $(X_i, \pi)$, we define the procedure {\bf Refine()} that takes the set $\RRR$ and outputs a set $\RRR'$ of sequences that conform to $(X_i, \pi)$, and does not violate properties (i) and (ii). For each sequence $\SSS$ in $\RRR$, we compute a sequence $\SSS'$ that refines $\SSS$ and satisfies property (ii), and replace $\SSS$ with $\SSS'$ in $\RRR$. Afterwards, we initialize $\RRR' = \emptyset$, and order the sequences in $\RRR$ arbitrarily. We iterate through the sequences in $\RRR$ in order, and add a sequence $\SSS_p$ to $\RRR'$ if there is no sequence $\SSS$ already in $\RRR'$ such that $\SSS \preceq_i  \SSS_p$, and there is no sequence $\SSS_q \in \RRR$, $q > p$ (\ie, $\SSS_q$ comes after $\SSS_p$ in the order), such that $\SSS_q \preceq \SSS_p$.

\begin{lemma}
\label{lem:refine}
Let $X_i$ be a bag, $\pi =(v_1=s, \sigma_1, v_2, \sigma_2, \dots, \sigma_{r-1}, v_r=t)$ a pattern for $X_i$, and $\WWW=(W_1, \ldots, W_{r-1})$ a sequence of walks, where each $W_j$ is a walk between vertices $v_j$ and $v_{j+1}$ in $G_{v_jv_{j+1}}^{i}$ satisfying $\chi(W_j) \leq k$. The procedure {\bf Refine()} on input $\WWW$ produces a set of sequences $\RRR'$ that conforms to $(X_i, \pi)$ satisfying properties (i) and (ii), and such that for each sequence $\SSS \in \WWW$, there is a sequence $\SSS' \in \RRR'$ satisfying $\SSS' \preceq_i \SSS$. Moreover, the procedure runs in time $\Oh^{*}(2^k|\WWW|+|\WWW|^2)$.
\end{lemma}

\begin{proof}
By Lemma~\ref{lem:refinetime}, refining a sequence in $\WWW$ takes $\Oh^{*}(2^k)$ time, and hence, refining all sequences in $\WWW$ takes $\Oh^{*}(2^k|\WWW|)$ time. After refining $\WWW$, we initialize $\RRR'$ to the empty set, and iterate through the sequences in $\WWW$, adding a sequence $\SSS_p \in \WWW$ to $\RRR'$ if there is no sequence $\SSS$ already in $\RRR'$ such that $\SSS \preceq_i  \SSS_p$; clearly this takes $\Oh^{*}(|\WWW|^2)$ time, and the lemma follows.
\end{proof}

If a bag $X_i$ is a leaf in $\mathcal{T}$, then $X_i=V_i=\{s, t\}$, and there are only two patterns $(s,0,t)$ and $(s,1,t)$ for $X_i$. Clearly, the only sequence that conforms to $(s,0,t)$ is the sequence $(())$ containing exactly one empty path. Moreover, there is  no edge $st \in E(G)$. Therefore, there is no sequence that conforms to $(s,1,t)$, and the following claim holds:
\begin{claim}\label{claim:leaf}
	If a bag $X_i$ is a leaf in $\mathcal{T}$, then $\Gamma_i=\{((s,0,t), \{(())\} ), ((s,1,t), \emptyset)\}$ contains, for each
	pattern for $X_i$, a representative set for $(X_i, \pi)$.
\end{claim}	

 We describe next how to update the table stored at a bag $X_i$, based on the tables stored at its children in $\mathcal{T}$.
We distinguish the following cases based on the type of bag $X_i$.

\begin{enumerate}
	\item[Case 1.] $X_i$ is an introduce node with child $X_j$. Let $X_i= X_j \cup\{v\}$.
\end{enumerate}
Clearly, for every pattern $\pi$ for $X_i$ that does not contain $v$, we can set $\Gamma_i[\pi] = \Gamma_j[\pi]$. $\Gamma_j[\pi]$ is a representative set for $(X_i, \pi)$ for the following reasons: (i) follows because every color in $\chi(X_i)\setminus \chi(X_j)$ does not appear in $V_j$, since $X_i$ is a vertex-separator in $G$ separating $v$ and $V_j$ and colors are connected. Hence, if two sequences in $\Gamma_j[\pi]$ that conform to $(X_i, \pi)$ contradict (i), they contradict (i) w.r.t.~$(X_j, \pi)$ as well, but $\Gamma_j[\pi]$ is a representative set for $(X_j, \pi)$. For properties (ii) and (iii), it is easy to observe that $v$ does not appear on any path between two vertices in $\pi$ having internal vertices in $V_i\setminus X_i$, and hence, these properties are inherited from the child node $X_j$.

Now let $\pi=(v_1=s, \sigma_1, v_2, \sigma_2, \dots, \sigma_{r-1}, v_r=t)$ be a pattern such that $v_q=v$, $q \in \{2, \ldots, r-1\}$, and let $\pi'=(v_1, \sigma_1,\ldots v_{q-1}, 0, v_{q+1},\sigma_{q+1}, \ldots, \sigma_{r-1}, v_r)$.
Note that since $X_j$ is a separator between $v$ and $V_j$, the only possibility for a path from $v$ to a different vertex in $X_i$ to have all internal vertices in $V_i\setminus X_i$ is if it is a direct edge. Therefore, if $\sigma_{q-1}=1$ (resp. $\sigma_{q}=1$) then $v_{q-1}v$ (resp. $v_qv$) is an edge in $G$. Otherwise, there is no sequence that conforms to $(X_i, \pi)$.

We obtain $\Gamma_i[\pi]$ from $\Gamma_j[\pi']$ as follows. For every $\SSS' = (P'_1, P'_2, \ldots, P'_{r-2}) \in \Gamma_j[\pi']$, we replace the empty path corresponding to $0$ between $v_{q-1}$ and $v_{q+1}$ in $\pi'$ by two paths  $P_{q-1}, P_{q}$ such that $P_{q-1} = ()$ (resp. $P_{q} = ()$)  if $\sigma_{q-1}=0$ (resp. $\sigma_{q}=0$) and $P_{q-1} = (v_{q-1},v)$ (resp. $P_{q-1} = (v,v_{q})$) otherwise and we obtain $\SSS=(P'_1, \dots, P'_{q-2}, P_{q-1}, P_{q}, P'_q, \ldots, P'_{r-2})$. Denote by $\RRR_\pi$ the set of all formed sequences $\SSS$.  Finally, we set $\Gamma_i[\pi]= \mbox{\bf Refine}(\RRR_\pi)$. We claim that $\Gamma_i[\pi]$ is a representative set for $(X_i, \pi)$.

\begin{claim}\label{claim:introduce}
	If $X_i$ is an introduce node with child $X_j$, and $\Gamma_{j}$ contains for each
	pattern $\pi'$ for $X_j$ a representative set for $(X_j, \pi')$, then $\Gamma_i[\pi]$ defined above is a representative set for $(X_i, \pi)$.
\end{claim}
\begin{proof}

It is clear that from the application of {\bf Refine()}, $\Gamma_i[\pi]$ does not contradict properties (i)-(ii) of the definition of representative sets. Assume now that there exists a sequence $\SSS \notin \Gamma_i[\pi]$ that conforms to $(X_i, \pi)$ such that $\SSS$ violates property (iii). We define the sequence $\SSS'$ that conforms to $\pi'$, and is the same as $\SSS$ on all paths that $\pi$ and $\pi'$ share. Since no two paths in $\SSS$ share a vertex that is not in $X_i$ (since $\SSS$ violates (iii)), and all paths in $\SSS'$ are also in $\SSS$, it follows that no two paths in $\SSS'$ share a vertex that is not in $X_j$. Since $\Gamma_j[\pi']$ is a representative set for $(X_j, \pi')$, it follows that there exists $\SSS'_1 \in \Gamma_j[\pi']$ such that $\SSS'_1 \preceq_j \SSS'$. Let $\SSS_1$ be the sequence
obtained from $\SSS'_1$ and conforming to $(X_i, \pi)$. Then $\SSS_1 \in \RRR_\pi$, and hence by Lemma~\ref{lem:refine}, there is a sequence $\SSS_2 \in \Gamma_i[\pi]$ such that $\SSS_2 \preceq_i \SSS_1$. Since either both $\SSS_1$ and $\SSS$ contain $v$ or none of them does, we have $\SSS_1 \preceq_i \SSS$. By transitivity of $\preceq_i$ (Lemma~\ref{lem:transitive}), it follows that $\SSS_2 \preceq_i \SSS$. This contradicts the assumption that $\SSS$ violates property (iii).
\end{proof}

\begin{enumerate}
	\item[Case 2.] $X_i$ is a forget node with child $X_j$. Let $X_i= X_j\setminus\{v\}$.
\end{enumerate}

Let $\pi = (s=v_1,\sigma_1,\ldots,\sigma_{r-1}, v_r=t)$ be a pattern for the vertices in $X_i$. For $q\in[r-1]$, such that $\sigma_q=1$, we define  $\pi^q=(s=v'_1,\sigma'_1,\ldots, \sigma'_{r}, v'_{r+1}=t)$ to be the pattern obtained from $\pi$ by inserting $v$ between $v_q$ and $v_{q+1}$ and setting $\sigma'_q=\sigma'_{q+1}=1$. More precisely, we set $v'_p=v_p$ and $\sigma'_p=\sigma_p$  for $1\le p\le q$, $v'_{q+1}=v$ and $\sigma'_{q+1}=1$, and finally $v'_p=v_{p-1}$ and $\sigma'_p=\sigma_{p-1}$  for $q+2\le p\le r$.
We define $\RRR_{\pi}$ as follows:

{\small
\begin{equation*}
 \RRR_\pi  = \Gamma_j[\pi]  \cup  \{\SSS = (P_1, \ldots, P_{q-1}, P_q\circ P_{q+1}, P_{q+2} ,\ldots, P_{r})\mid (P_1, \ldots ,P_{r})\in \Gamma_j[\pi^q], q\in [r-1]\wedge \sigma_q=1 \}.
\end{equation*}
}

Finally, we set $\Gamma_i[\pi]= \mbox{\bf Refine}(\RRR_\pi)$ and we claim that $\Gamma_i[\pi]$ is a representative set for $(X_i, \pi)$.

\begin{claim}\label{claim:forget}
	If $X_i$ is a forget node with child $X_j$, and $\Gamma_{j}$ contains for each
	pattern $\pi'$ for $X_j$ a representative set for $(X_j, \pi')$, then $\Gamma_i[\pi]$ defined above is a representative set for $(X_i, \pi)$.
\end{claim}
\begin{proof}

It is straightforward to see that $\Gamma_i[\pi]$ satisfies properties (i) and (ii) due to the way procedure {\bf Refine()} works. Assume for a contradiction that there exists a sequence $\SSS$ that violates property (iii). We distinguish two cases.

First, suppose that no path in $\SSS$ contains the vertex $v$. Then this path conforms to the pattern $\pi$ in $X_j$. Since no two paths in $\SSS$ share a vertex that is not in $X_i$, and since $\Gamma_j[\pi]$ is a representative set, there exists $\SSS_1 \in \Gamma_j[\pi]$ such that $\SSS_1 \preceq_j \SSS$.
Then $\SSS_1\in \RRR_\pi$, and hence by Lemma~\ref{lem:refine}, $\Gamma_i[\pi]$ contains a sequence $\SSS_2$ such that
 $\SSS_2 \preceq_i \SSS_1$. Since $\SSS_1 \preceq_j \SSS$ and $X_i \subsetneq X_j$, it follows from Observation~\ref{obs:jtoi} that $\SSS_1 \preceq_i \SSS$. By transitivity of $\preceq_i$, it follows that $\SSS_2 \preceq_i \SSS$, which is a contradiction to the assumption that $\SSS$ violates property (iii).

Second, suppose that there is a path $P_q$ in $\SSS$ that contains $v$ on a path between $v_q$ and $v_{q+1}$. We form a sequence $\SSS'$ from $\SSS$ by keeping every path $P \neq P_q$ in $\SSS$, and replacing $P_q$ in the sequence by the two subpaths of $P_q$, $P'_q = (v_q, \ldots, v)$ and $P'_{q+1} = (v, \ldots, v_{q+1})$. The sequence $\SSS'$ conforms to $(X_j, \pi^{q})$, and since no two paths in $\SSS$ share a vertex that is not in $X_i$, no two paths in $\SSS'$ share a vertex that is not in $X_j$. Since $\Gamma_{j}[\pi^{q}]$ is a representative set for $(X_j, \pi^{q})$, it follows that there exists a sequence $\SSS'_1 \in \Gamma_{j}[\pi^{q}]$ such that $\SSS'_1 \preceq_j \SSS'$. Let $\SSS_1$ be the sequence conforming to $(X_i, \pi)$ obtained from $\SSS'_1$ by applying the operation $\circ$ to the two paths in $\SSS'_1$ that share $v$. Then $\SSS_1 \in \RRR_{\pi}$. Therefore, by Lemma~\ref{lem:refine}, $\Gamma_i[\pi]$ contains a sequence $\SSS_2$ such that $\SSS_2 \preceq_i \SSS_1$. Since $\SSS'_1 \preceq_j \SSS'$, $\chi(\SSS') = \chi(\SSS)$, $\chi(\SSS'_1) = \chi(\SSS_1)$, and $X_i \subsetneq X_j$, it follows that $\SSS_1 \preceq_i \SSS$. By transitivity of $\preceq_i$, it follows that $\SSS_2 \preceq_i \SSS$, which is a contradiction to the assumption that $\SSS$ violates property (iii).
\end{proof}

\begin{enumerate}
	\item[Case 3.] $X_i$ is a join node with children $X_j$, $X_{j'}$.
\end{enumerate}

Let $\pi = (s=v_1,\sigma_1, \dots, \sigma_{r-1}, v_r=t)$ be a pattern for $X_i$. Initialize $\RRR_{\pi}=\emptyset$. For every two patterns $\pi_1=(s=v_1,\tau_1, \dots, \tau_{r-1}, v_r=t)$ and $\pi_2=(s=v_1,\mu_1, \dots, \mu_{r-1}, v_r=t)$ such that $\sigma_q= \tau_q + \mu_q$, and for every two sequences $\SSS_1 =(P_{1}^{1}, \ldots, P_{1}^{r-1}) \in \Gamma_j[\pi_1]$ and $\SSS_2 =(P_{2}^{1}, \ldots, P_{2}^{r-1}) \in \Gamma_{j'}[\pi_2]$, we add the sequence $\SSS=(P_1, \ldots, P_{r-1})$ to $\RRR_{\pi}$, where $P_q = P_{1}^{q}$ if $P_{2}^{q}$ is the empty path, otherwise, $P_q = P_{2}^{q}$, for $q \in [r-1]$. We set $\Gamma_i[\pi]= \mbox{\bf Refine}(\RRR_\pi)$, and we claim that $\Gamma_i[\pi]$ is a representative set for $(X_i, \pi)$.

\begin{claim}\label{claim:join}
	If $X_i$ is a join node with children $X_j$, $X_{j'}$, and $\Gamma_{j}$ (resp.~$\Gamma_{j'}$) contains for each
	pattern $\pi'$ for $X_j=X_{j'}=X_i$ a representative set for $(X_j, \pi')$ (resp.~$(\pi', X_{j'})$), then $\Gamma_i[\pi]$ defined above is a representative set for $(X_i, \pi)$.
\end{claim}

\begin{proof}
Clearly $\Gamma_i[\pi]$ satisfies properties (i) and (ii) due to the application of the procedure {\bf Refine()}. To argue that $\Gamma_i[\pi]$ satisfies property (iii), suppose not, and let $\SSS =(P_1, \ldots, P_{r-1})$ be a sequence that violates property (iii). Notice that every path $P_q$, $q \in [r-1]$, is either an edge between two vertices in $X_i$, or a path between two vertices in $X_i$ such that its internal vertices are either all in $V_j \setminus X_i$ or in  $V_{j'} \setminus X_i$; this is true because $X_i$ is a vertex-separator separating $V_j \setminus X_i$ from $V_{j'} \setminus X_i$ in $G$. Define the two sequences $\SSS_1 =(P_{1}^{1}, \ldots, P_{1}^{r-1})$ and $\SSS_2 =(P_{2}^{1}, \ldots, P_{2}^{r-1})$ as follows. For $q \in [r-1]$, if $P_q$ is empty then set both $P_{1}^{q}$ and $P_{2}^{q}$ to the empty path; if $P_q$ is an edge then set $P_{1}^{q} =P_q$ and $P_{2}^{q}$ to the empty path. Otherwise, $P_q$ is either a path in $G[V_j]$ or in $G[V_{j'}]$; in the former case set $P_{1}^{q} =P_q$ and $P_{2}^{q}$ to the empty path, and in the latter case set $P_{2}^{q} =P_q$ and $P_{1}^{q}$ to the empty path. Since no two paths in $\SSS$ share a vertex that is not in $X_i$, and $X_i=X_j=X_{j'}$, no two paths in $\SSS_1$ (resp.~$\SSS_2$) share a vertex that is not in $X_j$ (resp.~$X_{j'}$). Let $\pi_1=(s=v_1,\tau_1, \dots, \tau_{r-1}, v_r=t)$ and $\pi_2=(s=v_1,\mu_1, \dots, \mu_{r-1}, v_r=t)$ be the two patterns that $\SSS_1$ and $\SSS_2$ conform to, respectively, and observe that, for every $q \in [r-1]$, we have $\sigma_q= \tau_q + \mu_q$. Since $\Gamma_j[\pi_1]$ and $\Gamma_{j'}[\pi_2]$ are representative sets, it follows that there exist
$\SSS'_1 = (Y'_{1}, \ldots, Y'_{r-1})$ in $\Gamma_j[\pi_1]$ and $\SSS'_2 = (Z'_{1}, \ldots, Z'_{r-1})$ in $\Gamma_{j'}[\pi_2]$ such that $\SSS'_1 \preceq_j \SSS_1$ and $\SSS'_2 \preceq_{j'} \SSS_2$. Let $\SSS'=(P'_1, \ldots, P'_{r-1})$, where $P'_q = Y'_{q}$ if $Z'_{q}$ is the empty path, otherwise, $P_q = Z'_{q}$, for $q \in [r-1]$. The sequence $\SSS'$ conforms to $\pi$ and is in $\RRR_\pi$. By Lemma~\ref{lem:refine}, $\Gamma_i[\pi]$ contains a sequence $\SSS''$ such that $\SSS'' \preceq_i \SSS'$. From Observation~\ref{obs:jtoi}, since $X_i=X_j=X_{j'}$, from $\SSS'_1 \preceq_j \SSS_1$ and $\SSS'_2 \preceq_{j'} \SSS_2$ it follows that $\SSS'_1 \preceq_i \SSS_1$ and $\SSS'_2 \preceq_i \SSS_2$. Since $\chi(\SSS_1) \cup \chi(\SSS_2) = \chi(\SSS)$ and $\chi(\SSS'_1) \cup \chi(\SSS'_2) = \chi(\SSS')$, and since $\chi(\SSS_1) \cap \chi(\SSS_2) \subseteq \chi(X_i)$, by Lemma~\ref{lem:join}, it follows that $\SSS' \preceq_i \SSS$. Since $\SSS'' \preceq_i \SSS'$, by transitivity of $\preceq_i$, it follows that $\SSS'' \preceq_i \SSS$, which concludes the proof.
\end{proof}

We can now conclude with the following theorem:

\begin{theorem}
\label{thm:treewidth}
There is an algorithm that on input $(G, C, \chi, s, t, k)$ of \cmor{}, either outputs a $k$-valid $s$-$t$ path in $G$ or decides that no such path exists, in time $\Oh^{\star}(f(k)^{6\omega^2})$, where $\omega$ is the treewidth of $G$ and $f(k) =\Oh(c^{k^2} k^{2k^2+k})$, for some constant $c >1$. Therefore, \cmor{} parameterized by both $k$ and the treewidth of the input graph is in \FPT.	
\end{theorem}

\begin{proof}
First, in time $\Oh(|V(G)|^4)$, we can compute a branch decomposition of $G$, and hence a tree decomposition, of width at most $3\omega/2$, where $\omega$ is the treewidth of $G$~\cite{rat1,rs,rat2}. From this tree decomposition, in polynomial time we can compute a nice tree decomposition $(\VVV, \TTT)$ of $G$ whose width is at most $3 \omega/2$ and satisfying $|\VVV| = \Oh(|V(G)|)$~\cite{k94}.
The algorithm starts by removing the colors of $s$ and $t$ from $G$, and decrements $k$ by $|\chi(s)\cup\chi(t)|$ (see Assumption~\ref{ass:sat}). Afterwards, if $k<0$, the algorithm concludes that there is no $k$-valid $s$-$t$ path in $G$. If $st\in E(G)$ and $k\ge 0$, the algorithm outputs the path $(s,t)$. Now we know that $s$ and $t$ are not adjacent, and that $\chi(s)=\chi(t)=\emptyset$. The algorithm then adds $s$ and $t$ to every bag in $\TTT$, and executes the dynamic programming algorithm based on $(\VVV, \TTT)$, described in this section, to compute a table $\Gamma_i$ that contains, for each
bag $X_i$ in $\TTT$ and each pattern $\pi$ for $X_i$, a representative set $\RRR_{\pi}$ for $(X_i, \pi)$.
	
From Claims~\ref{claim:leaf}, \ref{claim:introduce}, \ref{claim:forget}, \ref{claim:join}, it follows, by induction on the height of the tree-decomposition $(\mathcal{V}, \TTT)$ (the base case corresponds to the leaves), that the root node $X_r$ contains a representative set $\Gamma_r[\pi]$ for the sequence $\pi=(s,1,t)$. If $\Gamma_r[\pi]$ is empty, the algorithm concludes that there is no $k$-valid $s$-$t$ path in $G$. Otherwise, noting that there is only one sequence $\SSS$ in the representative set $\Gamma_r[\pi]$ since $X_r=\{s, t\}$ and $s$ and $t$ are empty, the algorithm outputs the $k$-valid $s$-$t$ path $P$ formed by $\SSS$. The correctness follows from the following argument, which shows that if there is a $k$-valid $s$-$t$ path in $G$, then the algorithm outputs such a path. Suppose that $P'$ is a $k$-valid induced $s$-$t$ path such that there does not exist an $s$-$t$ path $P''$ in $G$ satisfying $\chi(P'') \subsetneq \chi(P')$, and let $\SSS'=(P')$. Since $G_{st}^r=G$, it follows that $\SSS'$ conforms to $(X_r, \pi)$. Since $\SSS'$ contains exactly one path that is induced, no two paths in $\SSS'$ share a vertex. Therefore, by property (iii) of representative sets, there exists a sequence $\SSS$ in $\Gamma_r[\pi]$ satisfying $\SSS \preceq_r\SSS'$. Noting that a sequence in $\Gamma_r[\pi]$ must consist of a single $k$-valid $s$-$t$ path, it follows that the algorithm correctly outputs such a path.

Next, we analyze the running time of the algorithm. We observe that among the three types of bags in $\TTT$, the worst running time is for a join bag. Therefore, it suffices to upper bound the running time for a join bag, and since $|\VVV|=\Oh(n)$, the upper bound on the overall running time would follow.

Consider a join bag $X_i$ with children $X_j, X_{j'}$. Let $\omega'$ be the width of $\TTT$ plus 1, which serves as an upper bound on the bag size in $\TTT$, and note that $\omega'\leq 3\omega/2 +3$, where the (additional) plus 2 is to account for the vertices $s$ and $t$ that were added to each bag. The algorithm starts by enumerating each pattern $\pi$ for $X_i$. The number of such patterns is at most $2^{\omega'} \cdot \omega' \cdot \omega'!=\Oh^{*}(2^{\omega'} \cdot \omega'!)$, where $\omega' \cdot \omega'!$ is an upper bound on the number of ordered selections of a subset of vertices from the bag, and $2^{\omega'}$ is an upper bound on the number of combinations for the $\sigma_i$'s in the selected pattern. Fix a pattern $\pi$ for $X_i$. To compute $\Gamma_i[\pi]$, the algorithm enumerates all ways of partitioning $\pi$ into pairs of patterns $\pi_1, \pi_2$ for the children bags; there are $2^{\omega'}$ ways of partitioning $\pi$ into such pairs, because for each $\sigma_i=1$ in $\pi$, the path between $v_i$ and $v_{i+1}$ is either reflected in $\pi_1$ or in $\pi_2$. For a fixed pair $\pi_1, \pi_2$, the algorithm iterates through all pairs of sequences in the two tables $\Gamma_j[\pi_1]$ and $\Gamma_{j'}[\pi_2]$. Since each table contains a representative set, by Lemma~\ref{lem:representativesbound}, the size of each table is $\Oh(h_1(k)^{\omega'^2})$, where $h_1(k) =\Oh(c_{1}^{k^2} k^{2k^2+k})$, for some constant $c_1 >1$, and hence iterating over all pairs of sequences in the two tables can be done in $\Oh(h_1(k)^{2\omega'^2})$ time. From the above, it follows that the set $\RRR_{\pi}$ can be computed in time $2^{\omega'} \cdot \Oh(h_1(k)^{2\omega'^2})= \Oh(h_2(k)^{2\omega'^2})$, where $h_2(k) =\Oh(c_{2}^{k^2} k^{2k^2+k})$, for some constant $c_2 >1$, which is also an upper bound on the size of $\RRR_{\pi}$. By Lemma~\ref{lem:refine}, applying {\bf Refine()} to $\RRR_{\pi}$ takes time $\Oh^{*}(2^k h_2(k)^{2\omega'^2} + h_2(k)^{4\omega'^2}) = \Oh^{*}(h_3(k)^{4\omega'^2})$, where $h_3(k) =\Oh(c_{3}^{k^2} k^{2k^2+k})$, for some constant $c_3 >1$. It follows from all the above that the running time taken by the algorithm to compute $\Gamma_i$ is $\Oh^{*}(h_3(k)^{4\omega'^2} \cdot 2^{\omega'} \cdot \omega'!)=\Oh^{*}(h_4(k)^{4\omega'^2})$, where $h_4(k) =\Oh(c_{4}^{k^2} k^{2k^2+k})$, for some constant $c_4 >1$, and hence
the running time of the algorithm is $\Oh^{\star}(f(k)^{6\omega^2})$, where $f(k) =\Oh(c^{k^2} k^{2k^2+k})$, for some constant $c >1$.
\end{proof}

\section{Extensions and Applications}
\label{sec:extension}
In this section, we extend the \FPT{} results for \cmor{} w.r.t.~the combined parameters $k$ and $\omega$---the treewidth of the input graph, to show that \cmor{} parameterized by both $k$ and the length $\ell$ of the sought path is \FPT{}. We also show some applications of these \FPT{} results. We formally define the problem \bcmor{}:

\paramproblem{\bcmor} {A planar graph $G$; a set of colors $C$; $\chi: V \longrightarrow 2^{C}$; two designated vertices $s, t \in V(G)$; and $k, \ell \in \nat$}{Does there exist a $k$-valid $s$-$t$ path of length at most $\ell$ in $G$?} \\

We first start by showing that if we parameterize only by one of $\ell, k$ then the problem is \W[1]-hard.

\begin{theorem}\label{thm:bcmor_hard_k}
\bcmor{} is {\rm \W[1]}-hard parameterized by $k$.
\end{theorem}

\begin{proof}
We reduce from the \W[1]-hard problem \textsc{Clique}. The reduction is similar to that in the proof of Theorem~\ref{thm:pathwidthnphardness}. Let $(G,k)$ be an instance of \textsc{Clique}, where $V(G)=\{v_1,\ldots, v_n\}$ and $E(G)=\{e_1,\ldots, e_m\}$. We assume that the edges in $E(G)$ are incident to at least $k+3$ different vertices. This assumption is safe because \textsc{Clique} is \FPT{} for instances where the edges in $E(G)$ are incident to at most $k+2$ different vertices. Similarly to the proof of Theorem~\ref{thm:pathwidthnphardness}, we start by describing the instance $I$ of \gcmor{} whose associated graph is the desired instance of \bcmor.
	
The regions of $I$ are $O \cup \{Z_0, \ldots, Z_m\} \cup \bigcup_{i=1}^{m} \{O_{i}^{1}, O_{i}^{2}, O_{i}^{3}\}$, depicted in Figure~\ref{fig:bcmor_hard_k}. 
The obstacles of $I$ are defined as follows. For each vertex $v_j \in V(G)$, the obstacle $V_j$ corresponding to $v_j$ is the polygon whose boundary is the boundary of the region formed by the union of $O$, and each $O_{i}^{3}$ such that $v_j$ is incident to $e_i$. More formally, the obstacle corresponding to $v_j$ is $\partial(O \cup \bigcup_{v_j\in e_i} O_{i}^{3})$. Besides the obstacles corresponding to the vertices of $G$, there are two auxiliary obstacles $A_1=\partial(O \cup \bigcup_{i\in [m]} O_{i}^{1})$ and $A_2=\partial(O \cup \bigcup_{i\in [m]} O_{i}^{2})$. Finally, we place $s$ in $Z_0$, $t$ in $Z_m$, and ask whether there is a path from $s$ to $t$ that intersects at most $k+2$ obstacles and visits at most $3m-\binom{k}{2}+1$ regions (including $Z_0$ and $Z_m$).
	
Suppose that $H$ is a complete subgraph of $G$ with exactly $k$ vertices. We define an $s$-$t$ path $P$ as follows. $P$ starts at $s$ in $Z_0$. If for $i\in [m]$, $P$ enters $Z_{i-1}$ and $e_i\in E(H)$, then $P$ goes to $Z_i$ through $O_i^3$; otherwise, $P$ goes to $Z_i$ through $O_i^1$ and $O_i^2$. Since $P$ does not enter the region $O$, and enters region $O_i^3$ if and only if edge $e_i$ is part of the clique $H$, $P$ intersects an obstacle $V_j$ if and only if $v_j\in V(H)$. Hence, together with $A_1$ and $A_2$, $P$ intersects at most $k+2$ obstacles. Moreover, $P$ visits regions $Z_0,\ldots, Z_m$, regions $O_i^3$ for $e_i\in E(H)$ ($\binom{k}{2}$ times), and regions $O_i^1$, $O_i^2$ for $e_i\notin E(H)$ ($m-\binom{k}{2}$ times). Therefore, $P$ visits exactly $3m-\binom{k}{2}+1$ regions.
	
Conversely, suppose that there is an $s$-$t$ path $P$ that visits at most $3m-\binom{k}{2}+1$ regions and intersects at most $k+2$ obstacles. It is easy to see that $P$ does not visit region $O$. Furthermore, by our assumption, the edges of $G$ are incident to at least $k+3$ vertices, and hence $P$ cannot intersect all the $O_i^3$'s, for $i\in [m]$. Therefore, $P$ visits $O_i^1$ and $O_i^2$ for some $i\in[m]$, and hence intersects $A_1$ and $A_2$. Furthermore, since $P$ visits at most $3m-\binom{k}{2}+1$ regions, $P$ visits at least $\binom{k}{2}$ different $O_i^3$'s. Since each $O_i^3$ contains a pair of two obstacles $V_{j_1}$, $V_{j_2}$ such that $e_i=v_{j_1}v_{j_2}$, there are no multiple edges in $G$ (and hence, no two $O_i^3$'s have the same set of obstacles), and $P$ intersects at most $k$ obstacles that corresponds to vertices of $G$, it follows that $P$ intersects exactly $k$ obstacles that correspond to vertices, and visits exactly $\binom{k}{2}$ different $O_i^3$'s. Moreover, for each pair of obstacles $V_{j_1}$, $V_{j_2}$ that $P$ intersect, there has to be an $O_i^3$ that contains exactly these two obstacles. This means that there is an edge $v_{j_1}v_{j_2}$ in $G$ for each such pair, and hence a $k$-clique in $G$.
\end{proof}

\begin{figure}[htbp]
	\begin{center}
		\includegraphics[width=0.95\textwidth]{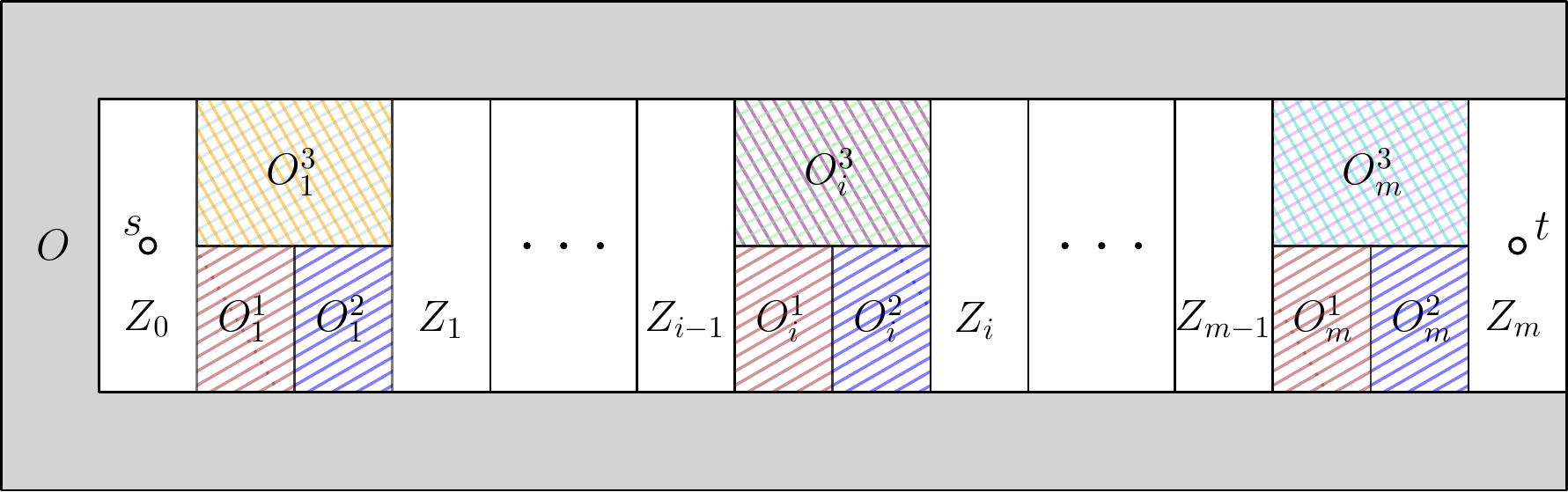}
	\end{center}
	\caption{Illustration for the proof of Theorem~\ref{thm:bcmor_hard_k}.}
	\label{fig:bcmor_hard_k}
\end{figure}

\begin{theorem}\label{thm:bcmor_hard_ell}
	\bcmor{} is {\rm \W[1]}-hard parameterized only by $\ell$.
\end{theorem}

\begin{proof}
	The proof is similar to the proof of Lemma~\ref{lem:w1hard}. We will describe the reduction from \textsc{Multi-Colored Clique} and point out the differences.
	
	Let $(G, k)$ be an instance of \textsc{Multi-Colored Clique}, where $V(G)$ is partitioned into the color classes $C_1, \ldots, C_k$.  We assume that all color classes have the same cardinality $N$. Let $C_j = \{u_{i}^{j} \mid i \in [N]\}$. We describe how to construct an instance $(G', C', \chi', s, t, k',\ell)$ of \bcmor{}. For an edge $e \in G$, associate a distinct color $c_e$, and define $C' = \{c_e \mid e \in E(G)\}$. Moreover, for convenience, we denote by $C'_{i,j}$ the set $\{c_e \mid e=uv \in E(G)\wedge u\in C_i\wedge v\in C_j \}$.
	To simplify the description of the construction, we start by defining a gadget that will serve as a building block for this construction.
	
	For a vertex $u_{i}^{j}$ in color class $C_j$, we define the gadget $G_{i,j}$, which is very similar to the one in Lemma~\ref{lem:w1hard}, 
	as follows. First, we create a copy of each color class $C_{j'}$, $j' \neq j$.
	Let the resulting copies of the color classes be $C'_1, \ldots, C'_{k-1}$. We define the color of a copy $v'$ of a vertex $u_{i'}^{j'}$ as $\chi'(v') = C'_{j,j'}\setminus \{c_e\}$ if there is an edge $e=u_{i}^{j}u_{i'}^{j'}$, and $\chi'(v') = C'_{j,j'}$ otherwise. Moreover, for each $i'\in [N-1]$ we connect the copies of vertices $u_{i'}^{j'}$ and $u_{i'+1}^{j'}$ by an edge.

	Next, we introduce $k-2$ empty vertices $y_{r}$, $r \in [k-2]$. For $r \in [k-2]$, we connect all vertices in $C'_r$ to $y_r$, and connect $y_r$ to all vertices in $C'_{r+1}$. This completes the construction of gadget $G_{i,j}$; we refer to $C'_1$ and $C'_{k-1}$ as the first and last color classes in gadget $G_{i,j}$, respectively.	Observe that each color $c_e$ for an edge $e=u_i^ju_{i'}^{j'}$ appears exactly on
	\begin{itemize}
		\item all vertices in the copies of $C_{j'}$ in the gadgets $G_{i^*,j}$ for $i^*\neq i$,
		\item all vertices in the copies of $C_{j}$ in the gadgets $G_{i^*,j'}$ for $i^*\neq i'$,
		\item all vertices but the copy of $u_{i'}^{j'}$ in the copies of $C_{j'}$ in the gadget $G_{i,j}$, and
		\item all vertices but the copy of $u_{i}^{j}$ in the copies of $C_{j}$ in the gadget $G_{i',j'}$.
	\end{itemize}
	 Furthermore, every path from a vertex in $C'_1$ to a vertex in $C'_{k-1}$ of length at most $2k-2$ contains exactly one vertex from each $C'_r$, $r \in [k-1]$, and contains all vertices $y_{r}$, $r \in [k-2]$. Moreover, each such path contains all but at most one color from each $C_{j,j'}$, for all $j'$ such that $j'\neq j$ and if the path 
	 does not contain a color $c_e\in C_{j,j'}$, then it contains a copy of a vertex $u_{i'}^{j'}$ such that there is an edge $e=u_{i}^{j}u_{i'}^{j'}$ in $G$.
	
	We continue the construction similarly as in Lemma~\ref{lem:w1hard} by introducing $k+1$ new empty vertices $z_0, \ldots, z_{k}$, and connecting them as follows. For each color class $C_j$, $j \in [k]$, and each vertex $u_{i}^{j} \in C_j$, we create the gadget $G_{i,j}$, connect $z_{j-1}$ to each vertex in the first color class of $G_{i,j}$, and connect each vertex in the last color class of $G_{i,j}$ to $z_j$. Let $G'$ be the resulting graph.  We now set $s=z_0$, $t=z_k$, and $k'=|E(G)|-\binom{k}{2}$ and $\ell=2k^2$. 
	We are nearly finished with the reduction. However, we need to make the colors in $G'$ connected. To achieve this, we first introduce a new vertex $o$ such that $\chi(o)=C'$. Now for each $j\in [k]$, each $j'\neq j$, and each $i\in [N-1]$ we connect the copy of $u_{N}^{j'}$ in $G_{i,j}$ with the copy of $u_{1}^{j'}$ in $G_{i+1,j}$, and we connect the vertex $o$ with the copies of $u_{1}^{j'}$ in $G_{1,j}$ and of $u_{N}^{j'}$ in $G_{N,j}$, respectively. It is not hard to see that every color is now connected.  See Figure~\ref{fig:w1_kl_hardness_l}
	for an illustration of the construction, and Figure~\ref{fig:w1_kl_hardness_l_1color}, which highlights the subgraph induced by one color.
	
	This completes the construction of the instance $(G', C', \chi', s, t, k',\ell)$ of \bcmor{}. 
	
	Clearly, the reduction that takes an instance $(G, k)$ of \textsc{Multi-Colored Clique} and produces the instance $(G', C', \chi', s, t, k', \ell)$ of \bcmor{} is computable in \FPT-time. To show its correctness, suppose that
	$(G, k)$ is a yes-instance of \textsc{Multi-Colored Clique}, and let $Q$ be a $k$-clique in $G$. Then $Q$ contains a vertex from each $C_j$, for $j \in [k]$. For a vertex $u_{i}^{j} \in Q$, let $G_{i,j}$ be its gadget, and define the path $P_j$ as follows. In each color class in $G_{i,j}$, pick the unique vertex that is a copy of a neighbor of $u_{i}^{j}$ in $Q$; define $P_j$ to be the path in $G_{i,j}$ induced by the picked vertices, plus the empty vertices $y_{r}$, $r \in [k-2]$, that appear in $G_{i,j}$. Finally, define $P$ to be the $s$-$t$ path in $G'$ whose edges are: the (unique) edge between $z_{r-1}$ and an endpoint of $P_r$, $P_r$, and the (unique) edge between an endpoint of $P_r$ and $z_r$, for $r \in [k]$. Clearly, the length of $P$ is exactly $(2k-2)k+2k=2k^2=\ell$. To show that $P$ is $k'$-valid, observe that none of the nonempty vertices in $P$ contains a color of an edge between two vertices in $Q$. This shows that the number of colors that appear on $P$ is at most $k'=|E(G)|-\binom{k}{2}$, and hence, $P$ is $k'$-valid.
	 It follows that $(G', C', \chi', s, t, k')$ is a yes-instance of \bcmor{}.

	Conversely, suppose that $P$ is a $k'$-valid $s$-$t$ path in $G'$ of length at most $2k^2$. It is easy to see that $P$ does not contain $o$, because $|\chi(o)|=|E(G)|$. Moreover, notice that the shortest path from $s$ to $t$ in $G-o$ has length $2k^2$ and each $s$-$t$ path of length $2k^2$ must start at $s$, visit the gadgets of exactly $k$ vertices $u_{i_j}^{j} \in C_j$, for $j \in [k], i_j \in [N]$, and end at $t$. Furthermore, the subpath of the path in each gadget $G_{i,j}$ has length exactly $2k-2$.
	We claim that
	$Q= \{u_{i_j}^{j} \mid j \in [k]\}$ is a clique in $G$.
	 Recall that the subpath of $P$ that traverses a gadget $G_{i, j}$ contains all but at most one color from each $C_{j,j'}$, for all $j'$ such that $j'\neq j$ and if such path 
	does not contain a color $c_e\in C_{j,j'}$, then it contains a copy of a vertex $u_{i'}^{j'}$ such that $e=u_{i}^{j}u_{i'}^{j'}$ is in $G$.	
	It follows that if $P$ does not contain the color $c_e$, for an edge $e=u_{i}^{j}u_{i'}^{j'}$, then it has to traverse the gadgets for $u_{i}^{j}$ and the gadget for $u_{i'}^{j'}$. Since $P$ traverses at most $k$ gadgets and it does not contain $\binom{k}{2}$ colors, it follows that there has to be an edge between every pair of vertices in $G$ for which $P$ traverses the corresponding gadgets, and $Q$ is a clique in $G$.
	Since $|Q|=k$, it follows that $Q$ is a $k$-clique in $G$, and that $(G, k)$ is a yes-instance of \textsc{Multi-Colored Clique}.
\end{proof}

	
	\begin{figure}[htbp]
		\begin{center}
			\includegraphics[width=\textwidth,page=1]{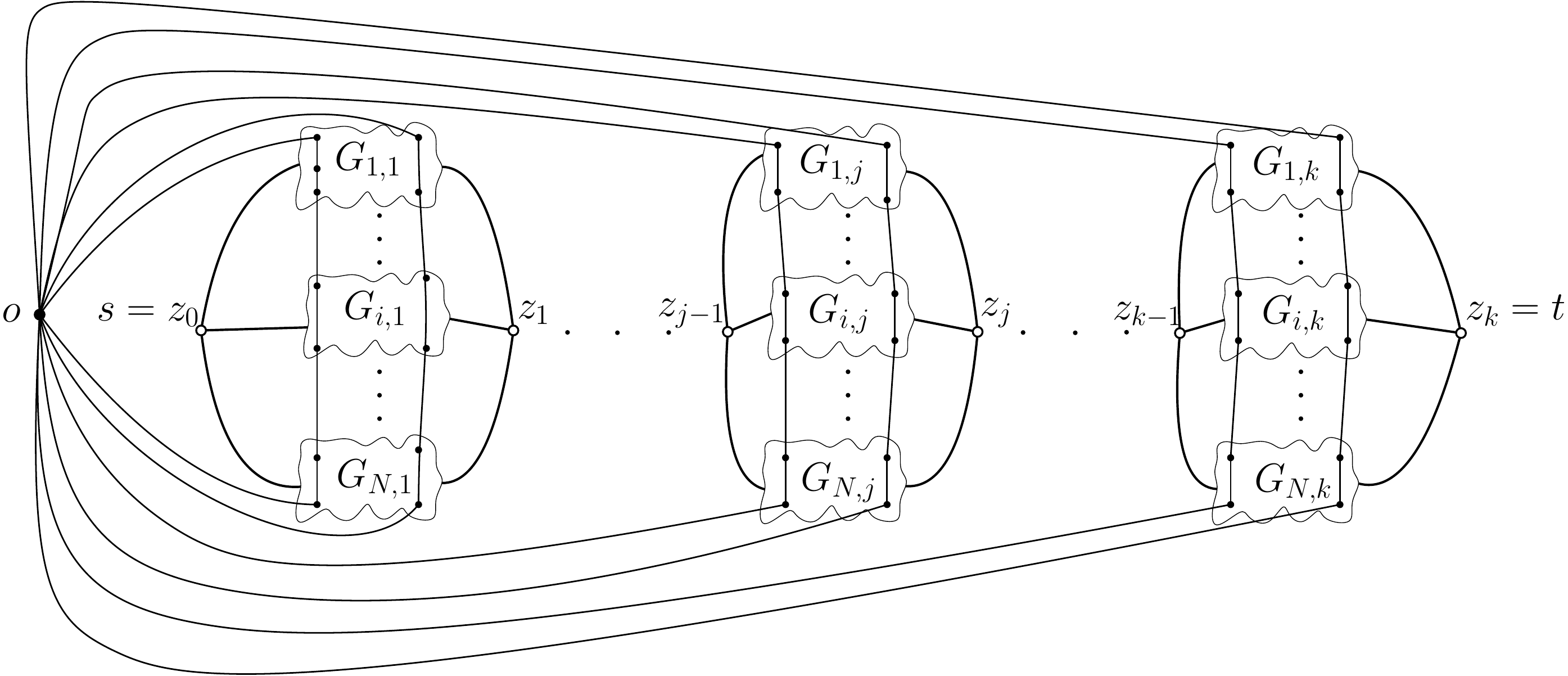}
		\end{center}
	\caption{Illustration of the reduced instance for the proof of Theorem~\ref{thm:bcmor_hard_ell}.}\label{fig:w1_kl_hardness_l}
	\end{figure}
	\begin{figure}[htbp]
		\begin{center}
			\includegraphics[width=\textwidth,page=2]{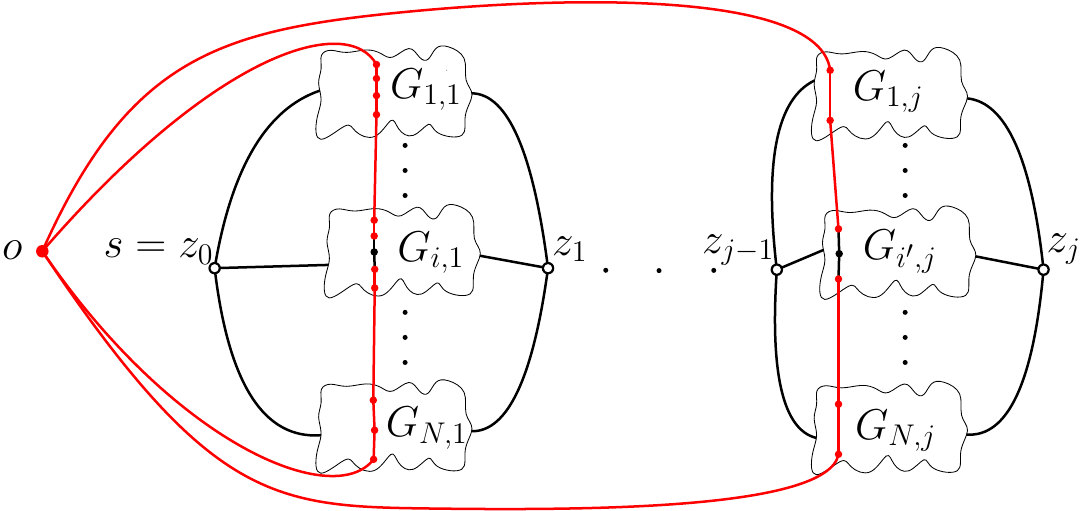}
		\end{center}
	\caption{Illustration for the proof of Theorem~\ref{thm:bcmor_hard_ell}. The red highlights a color representing the edge between vertices $u_i^1$ in $C_1$ and $u_{i'}^j$ in $C_j$.}\label{fig:w1_kl_hardness_l_1color}
	\end{figure}
	

We now switch our attention to showing that \bcmor{} parameterized by both $k$ and $\ell$ is \FPT{}. We start with the following lemma that enables us to upper bound the treewidth of the input graph by a function of the parameter $\ell$:

\begin{lemma}\label{lem:ext_distantvertex}
	Let $(G, C, \chi, s, t, k, \ell)$ be an instance of \bcmor, and let $v$ be a vertex in $G$ such that $d_G(s, v) > \ell +1$. Let $G'$ be the graph obtained from $G$ by contracting any edge $uv$ that is incident to $v$, and let $\chi'(x) = \chi(u) \cup \chi(v)$, where $x$ is the new vertex resulting from contracting $uv$, and $\chi'(w) = \chi(w)$ for any $w \in V(G) \setminus \{u, v\}$. Then $(G', C, \chi', s, t, k, \ell)$
	is a yes-instance of \bcmor{} if and only if $(G, C, \chi, s, t, k, \ell)$ is.
\end{lemma}

\begin{proof}
	Since $G$ is color-connected and $\chi(x) = \chi(u) \cup \chi(v)$, it is easy to see that $G'$ is color-connected as well. Because $d_G(s, v) > \ell+1$, any solution to $(G, C, \chi, s, t, k, \ell)$ does not contain any of $u, v$, and hence, is a solution to $(G', C, \chi', s, t, k, \ell)$. Conversely, because $d_G(s, v) \geq \ell+1$, any solution to $(G', C, \chi', s, t, k, \ell)$ does not contain $x$, and hence is a solution to $(G, C, \chi, s, t, k, \ell)$.
\end{proof}

By Lemma~\ref{lem:ext_distantvertex}, we may assume w.l.o.g.~that in an instance $(G, C, \chi, s, t, k, \ell)$ of \bcmor{}, every vertex $v \in V(G)$ satisfies $d_G(s,v) \leq \ell+1$. Therefore, we may assume that $G$ has radius at most $\ell+1$, and hence $G$ has treewidth at most $3\cdot(\ell+1)+1=3\ell+4$~\cite{rs}.

At this point we draw the following observation. Although the treewidth of $G$ is bounded by a function of $\ell$, we cannot use the \FPT{} algorithm for \mor{}, parameterized by $k$ and the treewidth of $G$, to solve \bcmor{} because the $k$-valid path returned by the algorithm for \cmor{} may have length exceeding the desired upper bound $\ell$. In fact, extending the \FPT{} results for \mor{} to \bcmor{} turns out to be a nontrivial task, that necessitates a nontrivial extension of the structural results in Section~\ref{sec:structural}, as well as the dynamic programming algorithm in Section~\ref{sec:algo}. In particular, the color contraction operation, on which the structural results developed in Section~\ref{sec:structural} hinge, is no longer applicable since contracting an edge may decrease the distance between $s$ and $t$ in the resulting instance, and hence, may not result in an equivalent instance of the problem. However, we will show in the next subsection that we can extend the notion of a minimal set of $k$-valid paths between two vertices to incorporate the length of these paths, while still be able to upper bound the size of such a set by a function of both $k$ and the length of these paths.

\subsection{Extended Structural Results}
\label{subsec:structural}
We start with the following definition:

\begin{definition}\rm
\label{def:minimalpathsextended}
Let $u, v, w \in V(G)$, and let $\lambda \in [\ell]$. Let ${\cal P}$ be a set of $k$-valid $u$-$v$ paths in $G-w$, each of length $\lambda$. The set $\PPP$ is said to be {\em $\lambda$-minimal} w.r.t.~$w$ if there does not exist two paths $P_1, P_2 \in {\cal P}$ such that $\chi(P_1) \cap \chi(w) = \chi(P_2) \cap \chi(w)$.
\end{definition}

Let $u, v, w \in V(G)$, $\lambda \in [\ell]$, and let ${\cal P}$ be a set of $\lambda$-minimal $k$-valid $u$-$v$ paths in $G-w$.  Let ${\cal M}$ be a set of $\lambda$-minimal $k$-valid color-disjoint $u$-$v$ paths in $G-w$. Let $H$ be the subgraph of $G-w$ induced by the edges of the paths in ${\cal P}$, and let $M$ be that induced by the edges of the paths in ${\cal M}$.

\begin{lemma}
\label{lem:separatorextended}
$M$ has a $u$-$v$ vertex-separator of cardinality at most $2(\lambda+1)$.
\end{lemma}

\begin{proof}
We proceed by contradiction, and assume that $M$ does not have a $u$-$v$ vertex-separator of cardinality at most $2(\lambda+1)$. By Menger's theorem~\cite{diestel}, there exists a set $\DDD=\{P_1, \ldots, P_{r}\}$, where $r \geq 2\lambda+3$, of vertex-disjoint $u$-$v$ paths in $M$. Let $u_1, \ldots, u_{r}$ be the neighbors of $u$ in counterclockwise order such that $P_i$ contains $u_i$, $i \in [r]$, and let $Q_i$ be a path in $\MMM$ containing $u_i$.

Since all paths in $\MMM$ have the same length $\lambda$, Definition~\ref{def:minimalpathsextended} implies that Observation~\ref{obs:internalpath} holds. Therefore, at most one path in $\MMM$ contains only internal colors w.r.t.~$w$ in $M$. By Observation~\ref{obs:externalcolor}, any vertex on a path in $\MMM$ such that the vertex contains an external color w.r.t.~$w$ in $M$ must be incident to the external face to $w$ in $M$. Choose $r' \in [r]$ such that $|r'-\lfloor \frac{r}{2}\rfloor|$ is minimum and $Q_{r'}$ contains an external color w.r.t.~$w$ in $M$. Since $Q_{r'}$ contains an external color w.r.t.~$w$ in $M$, $Q_{r'}$ contains a vertex incident to the external face to $w$ in $M$. Since all paths in $\DDD$ are $u$-$v$ vertex-disjoint paths, $Q_{r'}$ contains vertices other than $u$ and $v$ from at least $\lfloor r/2\rfloor -1$ distinct paths (including itself) in $\DDD$. Since the paths in $\DDD$ are all vertex disjoint, it follows that $|Q_{r'}| \geq r/2 -1$, and hence $\lambda \geq r/2-1$, which implies that $r \leq 2(\lambda+1)$. This contradicts our assumption that $r \geq 2\lambda+3$.
\end{proof}

\begin{lemma}
\label{lem:auxextended}
$|{\cal M}|\leq g(\lambda)$, where $g(\lambda)=\Oh(c^{\lambda} \lambda^{3\lambda})$, for some constant $c > 1$.
\end{lemma}

\begin{proof}
 As in the proof of Lemma~\ref{lem:separatorextended}, Definition~\ref{def:minimalpathsextended} implies that Observation~\ref{obs:internalpath} holds, and hence, at most one path in $\MMM$ contains only internal colors w.r.t.~$w$ in $G-w$. Therefore, we upper bound the number of paths in $\MMM$ that each contains at least one external color to $w$ in $G-w$, and add 1 to $g(\lambda)$ at the end. Henceforth, we shall assume that every path in $\MMM$ contains a color that is external to $w$ in $M$.

 The proof is by induction on $\lambda$, over every color-connected plane graph $G$, every triplet of vertices $u, v, w$ in $G$, and every $\lambda$-minimal set $\MMM$ w.r.t.~$w$ in $G-w$ of $k$-valid pairwise color-disjoint $u$-$v$ paths.  If $\lambda=1$, then $|\MMM| \leq 1 \leq g(1)$, if we choose $g(1)$ to be at least $1$.

Suppose, by the inductive hypothesis, that for any $1 \leq i < \lambda$, we have $|\MMM| \leq g(i)$.  By Lemma~\ref{lem:separatorextended}, $M$ has a $u$-$v$ vertex-separator $S$ satisfying $|S| \leq 2\lambda+2$. $S$ separates $M$ into two subgraphs $M_u$ and $M_v$ such that $u \in V(M_u)$, $v \in V(M_v)$, and there is no edge between $M_u$ and $M_v$. We partition $\MMM$ into two groups: (1) The set of paths in $\MMM$ that each contains a nonempty vertex in $S$; and (2) the set of remaining paths $\MMM_{\emptyset}$, which contains each path in $\MMM$ whose intersection with $S$ consists of only empty vertices.  Since the paths in $\MMM$ are pairwise color-disjoint, no nonempty vertex in $S$ can appear on two distinct paths from group (1). Therefore, the number of paths in group (1) is at most $|S| \leq 2\lambda+2$.

To upper bound the number of paths in group (2), suppose that $S=\{s_2, \ldots, s_{r-1}\}$, where $r \leq 2\lambda+4$, and extend $S$ by adding the two vertices $s_1=u$ and $s_r=v$ to form the set
$A=\{s_1, s_2, \ldots, s_{r}\}$. For every two (distinct) vertices $s_j, s_{j'} \in A$, $j, j' \in [r], j < j'$, we define a set of paths $\PPP_{jj'}$ in $G-w$ whose endpoints are $s_j$ and $s_{j'}$ as follows.
For each path $P$ in group (2), partition $P$ into subpaths $P_1, \ldots, P_q$ satisfying the property that the endpoints of each $P_i$, $i \in [q]$, are in $A$, and no internal vertex to $P_i$ is in $A$.
Since $P$ contains a vertex that contains an external color to $w$ in $G-w$, there exists an $i \in [q]$ such that $P_i$ contains a vertex that contains an external color to $w$ in $G-w$; pick any such $i \in [q]$, and assign $P_i$ to the set of paths $\PPP_{jj'}$ such that $s_j$ and $s_{j'}$ are the endpoints of $P_i$. Since each $P_i$ contains an external color that appears on $P$ and the paths in $\MMM$ are pairwise-color disjoint, it follows that the map that maps each $P$ to its $P_i$ is a bijection. Moreover, since each path $P$ must intersect $S \setminus \{u, v\}$, the length of each path in $\PPP_{jj'}$, for $j, j' \in [r], j < j'$, is strictly smaller than $\lambda$.

Fix a set $\PPP_{jj'}$.
For any fixed length $i' \in [\lambda-1]$, the subset of paths in $\PPP_{jj'}$ of length $i'$, $\PPP_{jj'}^{i'}$, have $s_j, s_{j'}$ as endpoints, and are pairwise color-disjoint. Moreover, each path in $\PPP_{jj'}^{i'}$ contains a vertex that contains an external color to $w$ in $G-w$. It follows from the previous statements that $\PPP_{jj'}^{i'}$ satisfies Definition~\ref{def:minimalpathsextended} w.r.t.~$G$ and $w$, and
hence $\PPP_{jj'}^{i'}$, $i' \in [\lambda-1]$, is an $i'$-minimal set of $k$-valid $s_j$-$s_{j'}$ paths in $G$ w.r.t.~$w$. By the inductive hypothesis, we have $|\PPP_{jj'}^{i'}| \leq g(i')$. Since the number of sets $\PPP_{jj'}$ is at most $\binom{2\lambda+4}{2}$, $i'\leq \lambda-1$, and noting that $g$ is an increasing function, the number of paths in group (2) is $\Oh(\lambda^2) \cdot (\lambda-1)\cdot g(\lambda-1)=\Oh(\lambda^3) \cdot g(\lambda-1)$.

It follows from the above that $|\MMM| \leq g(\lambda)$, where $g(\lambda)$ satisfies the recurrence relation $g(\lambda) \leq (2\lambda+2) + \Oh(\lambda^3) \cdot g(\lambda-1) = \Oh(\lambda^3) \cdot g(\lambda-1)$. Solving the aforementioned recurrence relation gives $g(\lambda) =\Oh(c^{\lambda} \lambda^{3\lambda})$, where $c > 1$ is a constant. Adding 1 to $g(\lambda)$ to account for the single path in $\MMM$ containing only internal colors w.r.t.~$w$ in $M$ yields the same asymptotic upper bound.
\end{proof}

\begin{theorem}
\label{thm:ext_main}
Let $G$ be a plane color-connected graph, let $u, v, w \in V(G)$, let $\lambda \in [\ell]$, and let ${\cal P}$ be a set of $\lambda$-minimal $k$-valid $u$-$v$ paths w.r.t.~$w$ in $G-w$. Then $|\PPP|\leq h(k,\lambda)$, where $h(k,\lambda)=\Oh(c^{\lambda k} \cdot k^k \cdot \lambda^{3\lambda k})$, for some constant $c > 1$.
\end{theorem}

\begin{proof}
The proof is by induction on $k$. If $k=0$, then by minimality of $\PPP$, there can be at most one path in $\PPP$, namely the path consisting of empty vertices. If $k=1$, then by minimality of $\PPP$, we have $\PPP=\MMM$, and by Lemma~\ref{lem:auxextended}, $|\PPP|=\Oh(c^{\lambda} \lambda^{3\lambda})=\Oh(c^{\lambda k} \cdot k^k \cdot \lambda^{3\lambda k})$.

Assume by the inductive hypothesis that the statement of the lemma is true for $1 \leq i < k$. Let $\MMM$ be a maximal set of pairwise color-disjoint paths in $\PPP$. By Lemma~\ref{lem:aux}, $|\MMM| \leq g(\lambda)=\Oh(c^{\lambda} \lambda^{3\lambda})$. The number of colors contained in the vertices of $\MMM$ is at most $r \leq k\cdot g(\lambda)$. We group the paths in $\PPP$ into $r$ groups $\PPP_1, \ldots, \PPP_r$, such that all the paths in $\PPP_i$, $i \in [r]$, share the same color $c_i$, where $i \in [r]$, that is distinct from each color $c_j$ shared by the paths in $\PPP_j$, for $j \neq i$. We upper bound the number of paths in each $\PPP_i$, $i \in [r]$, to obtain an upper bound on $|\PPP|$.

Let $G_i$ be the graph obtained by removing color $c_i$ from each vertex in 
$G$ that $c_i$ appears on, and let $\PPP'_i$ be the set of paths obtained from $\PPP_i$ by removing color $c_i$ from each vertex in $\PPP_i$ that $c_i$ appears on. Clearly, every path in $\PPP'_i$ is a $(k-1)$-valid $u$-$v$ path of length $\lambda$. Moreover, it is easy to verify that $\PPP'_i$ satisfies Definition~\ref{def:minimalpathsextended}, and hence, $\PPP'_i$ is $\lambda$-minimal w.r.t.~$w$ in 
$G_i-w$. By the inductive hypothesis, we have $|\PPP'_i| \leq h(k-1, \lambda)$.
It follows that the total number of paths in $\PPP$ is at most $h(k, \lambda)$, where $h(k, \lambda)$ satisfies the recurrence relation $h(k, \lambda) \leq r \cdot h(k-1, \lambda) \leq k \cdot g(\lambda) \cdot h(k-1, \lambda)$. Solving the aforementioned recurrence relations yields $h(k, \lambda)=\Oh((k\cdot g(\lambda))^k)=\Oh(c^{\lambda k} \cdot k^k \cdot \lambda^{3\lambda k})$.
\end{proof}

The result of Theorem~\ref{thm:ext_main} will be employed in the next section in the form presented in the following corollary:

\begin{corollary}
\label{cor:ext_main}
Let $G$ be a plane color-connected graph, let $w \in V(G)$, and let $\lambda \in [\ell]$. Let $G'$ be a subgraph of $G-w$, and let $u, v \in V(G')$. Every set ${\cal P}$ of $\lambda$-minimal $k$-valid $u$-$v$ paths in $G'$ w.r.t.~$w$ satisfies $|\PPP|\leq h(k,\lambda)$, where $h(k, \lambda)=\Oh(c^{\lambda k} \cdot k^k \cdot \lambda^{3\lambda k})$, for some constant $c > 1$.
\end{corollary}

\begin{proof}
Contract every connected component of $(G-w)-G'$ into a single vertex containing the union of the color-sets of the vertices in the component, and add $k+1$ new distinct colors to the resulting vertex. Denote the resulting graph by $G''$. Observe that the resulting graph is color-connected, and that every $k$-valid $u$-$v$ path of length $\lambda$ in $G'$ w.r.t.~$w$ is a $k$-valid $u$-$v$ path of length $\lambda$ in $G''$ w.r.t.~$w$, and vice versa. Therefore, every set ${\cal P}$ of $\lambda$-minimal $k$-valid $u$-$v$ paths in $G'$ w.r.t.~$w$ is also a set of $\lambda$-minimal $k$-valid $u$-$v$ paths in $G''$ w.r.t.~$w$. For any set ${\cal P}$ of $\lambda$-minimal $k$-valid $u$-$v$ paths w.r.t.~$w$ in $G'$, by applying Theorem~\ref{thm:ext_main} to ${\cal P}$ in $G''-w$, the corollary follows.
\end{proof}

\subsection{The Extended Algorithm}
\label{subsec:extendedalgo}

Let $(G, C, \chi, s, t, k, \ell)$ be an instance of \bcmor. The algorithm is a dynamic programming algorithm based on a tree decomposition of $G$. Let $({\cal V}, {\cal T})$ be a nice tree decomposition of $G$. By Assumption~\ref{ass:sat}, we can assume that $s$ and $t$ are nonadjacent empty vertices. We add $s$ and $t$ to every bag in $\TTT$, and from now on, we assume that $\{s,t\}\subseteq X_i$, for every bag $X_i \in \TTT$.
For a bag $X_i$, we say that $v\in X_i$ is \emph{useful} if $|\Col(v)|\le k$. Let $U_i$ be the set of all useful vertices in $X_i$ and let $\overline{U_i}=X_i\setminus U_i$.
We denote by $V_i$ the set of vertices in the bags of the subtree of $\mathcal{T}$ rooted at $X_i$. For any two vertices $u, v \in X_i$, let $G_{uv}^i=G[(V_i\setminus X_i) \cup \{u, v\}]$. We extend the notion of a $\lambda$-minimal set of $k$-valid $u$-$v$ paths w.r.t.~a vertex, developed in the previous section, to the set of vertices in a bag of $\TTT$.

\begin{definition}
\label{def:ext_minimalpathbag}
Let $\lambda \in [\ell]$. A set of $k$-valid $u$-$v$ paths $\PPP_{uv}$ in $G_{uv}^i$ is {\em $\lambda$-minimal} w.r.t.~$X_i$ if each path in  $\PPP_{uv}$ has length exactly $\lambda$ and there does not exist two paths $P_1, P_2 \in \PPP_{uv}$ such that $\chi(P_1) \cap \chi(X_i) = \chi(P_2) \cap \chi(X_i)$.
\end{definition}

\begin{lemma}
	\label{lem:ext_one_path_bound}
	Let $X_i$ be bag, $u, v \in X_i$, $\lambda \in [\ell]$ and $\PPP_{uv}$ a $\lambda$-minimal set of $k$-valid $u$-$v$ paths w.r.t.~$X_i$ in $G_{uv}^i$. Then the number of paths in $\PPP_{uv}$ is at most $h(k,\lambda)^{|X_i|}$, where $h(k,\lambda)=\Oh(c^{\lambda k} \cdot k^k \cdot \lambda^{3\lambda k})$, for some constant $c >1$.
\end{lemma}

\begin{proof}
	Let $X_i\setminus \{u, v\}=\{w_1, \ldots, w_r\}$, where $r=|X_i|-2$. For each $w_j \in X_i$, $j \in [r]$, let $\PPP_j$ be a $\lambda$-minimal set of $k$-valid $u$-$v$ paths w.r.t.~$w_j$ in $G_{uv}^i$. Without loss of generality, we can pick $\PPP_j$ such that there is no $k$-valid $u$-$v$ path $P$ in $G_{uv}^i$ such that $\PPP_j \cup \{P\}$ is $\lambda$-minimal. From Corollary~\ref{cor:ext_main}, we have $|\PPP_j| \leq h(k,\lambda) =\Oh(c^{\lambda k} \cdot k^k \cdot \lambda^{3\lambda k})$, for some constant $c >1$. For each $P \in \PPP_{uv}$, and each $j \in [r]$, define $C_j = \chi(P) \cap \chi(w_j)$. Define the \emph{signature} of $P$ (w.r.t.~the colors of $w_1, \ldots, w_r$) to be the tuple $(C_1, \ldots, C_r)$.
	Observe that no two (distinct) paths $P_1, P_2 \in \PPP_{uv}$ have the same signature; otherwise, since $u$ and $v$ appear on both $P_1, P_2$, $\chi(P_1) \cap \chi(X_i)= \chi(P_2) \cap \chi(X_i)$, which contradicts the definition of the $\lambda$-minimality of $\PPP_{uv}$. For each $P \in \PPP_{uv}$, and each $j \in [r]$, there is a path $P' \in \PPP_{j}$ such that $\chi(P') \cap \chi(w_j) =C_j$. Otherwise,  $\PPP_{j}\cup\{P\}$ satisfy Definition~\ref{def:minimalpathsextended}, which contradicts our assumption  that there is no $k$-valid $u$-$v$ path $P$ in $G_{uv}^i$ such that $\PPP_j \cup \{P\}$ is $\lambda$-minimal. It follows that the number of signatures of paths in $\PPP_{uv}$ is at most
	$\prod_{j=1}^{r}|\PPP_j| \leq h(k, \lambda)^{|X_i|}$. Since no two distinct paths in $\PPP_{uv}$ have the same signature, it follows that $|\PPP_{uv}| \leq h(k, \lambda)^{|X_i|}$.
\end{proof}

We define the length of a sequence of paths (walks)  $\SSS$, denoted by $|\SSS|$, to be the sum of the lengths of the paths in $\SSS$.

\begin{definition}\rm
\label{def:representivesetextended}
Let $X_i$ be a bag and $\pi=(v_1, \sigma_1, v_2 \ldots, \sigma_{r-1}, v_r)$ a pattern for $X_i$. A set $\RRR_\pi$ of sequences of length at most $\ell$ that conform to $(X_i, \pi)$ is a {\em representative} set for $(X_i, \pi)$ if:
\begin{itemize}
	\item[(i)]  For every sequence $\SSS_1\in \RRR_\pi$, and for every sequence $\SSS_2\neq \SSS_1$ that conforms to $(X_i, \pi)$, if $\SSS_1 \preceq_i \SSS_2$ and $|\SSS_1| \leq |\SSS_2|$ then $\SSS_2 \notin \RRR_{\pi}$; and
	\item[(ii)] for every sequence $\SSS\notin\RRR_{\pi}$, $|\SSS|\le \ell$, that conforms to $(X_i, \pi)$ and satisfies that no two paths in $\SSS$ share a vertex that is not in $X_i$, there is a sequence $\WWW \in \RRR_{\pi}$ such that $\WWW \preceq_i \SSS$ and $|\WWW| \leq |\SSS|$.
\end{itemize}
\end{definition}

We mention that Lemma~\ref{lem:transitive} and Observation~\ref{obs:jtoi} extend as they are to the current setting.

\begin{lemma}
	\label{lem:representativesboundextended}
	Let $X_i$ be bag, $\pi$ a pattern for $X_i$, 
	and $\RRR_{\pi}$ a representative set for $(X_i, \pi)$. Then the number of sequences in $\RRR_{\pi}$ is at most $h(k,\ell)^{|X_i|^2}$, where $h(k,\ell) =\Oh(c^{\ell k} \cdot k^k \cdot \ell^{3\ell k})$, for some constant $c >1$.
\end{lemma}

\begin{proof}
Let $\pi =(v_1=s, \sigma_1, v_2, \sigma_2, \dots, \sigma_{r-1}, v_r=t)$, and let
$\Lambda = (\lambda_1, \ldots, \lambda_{r-1})$ be such that, for each $j \in [r-1]$: (1) $\lambda_j=0$ if $\sigma_j=0$ and $\lambda_j\in [\ell]$ otherwise, and (2) $\sum_{j=1}^{r-1}\lambda_j\le \ell$.
For each $\lambda \in [\ell]$, the number of tuples $\Lambda = (\lambda_1, \ldots, \lambda_{r-1})$ satisfying $\sum_{j=1}^{r-1}\lambda_j = \lambda$ is the number of {\em weak compositions} of
$\lambda$ into $r-1$ parts, which is $\binom{\lambda+r-2}{r-2}$. It follows that the number of tuples $\Lambda = (\lambda_1, \ldots, \lambda_{r-1})$ satisfying $\sum_{j=1}^{r-1}\lambda_j \leq \ell$ is upper bounded by
$\binom{\ell+r-1}{r-2} \leq \binom{|X_i|+\ell}{r-2}\leq 2^{|X_i|+\ell}$. Therefore, if we upper bound the number of sequences in $\RRR_{\pi}$ corresponding to some fixed tuple $\Lambda$ by $h_1(k,\ell)^{|X_i|^2} = \Oh(c_1^{\ell k} \cdot k^k \cdot \ell^{3\ell k})$ for some constant $c_1 >1$, then we obtain $\RRR_{\pi}\le  2^{|X_i|+\ell}\cdot h_1(k,\ell)^{|X_i|^2} \le h(k,\ell)^{|X_i|^2}$, where $h(k,\ell) =\Oh(c^{\ell k} \cdot k^k \cdot \ell^{3\ell k})$, for some constant $c >1$. Therefore, for the rest of the proof, we fix $\Lambda = (\lambda_1, \ldots, \lambda_{r-1})$ and we let $\RRR_{\pi}^\Lambda$ be the subset of $\RRR_\pi$ such that for each sequence $\SSS = (P_1,\ldots, P_{r-1})$ in $\RRR_{\pi}^\Lambda$ it holds that the length of $P_j$ is $\lambda_j$ for each $j\in[r-1]$ such that $\sigma_j=1$.

For each $j\in [r-1]$ such that $\sigma_j=1$, let $\PPP_j$ be a $\lambda_j$-minimal set of $k$-valid $v_j$-$v_{j+1}$ paths w.r.t.~$X_i$. Without loss of generality, we can pick $\PPP_j$ such that there is no $k$-valid $u$-$v$ path $P$ of length $\lambda_j$ in $G_{v_jv_{j+1}}^i$ such that $\PPP_j \cup \{P\}$ is $\lambda_j$-minimal w.r.t.~$X_i$. From Lemma~\ref{lem:ext_one_path_bound}, it follows that $|\PPP_j|\le h_1(k,\lambda_j)^{|X_i|}$, where $h_1(k,\lambda_j)=\Oh(c_1^{\lambda_j k} \cdot k^k \cdot \lambda_j^{3\lambda_j k})$, for some constant $c_1 >1$.

	For a sequence $\SSS= (P_1,\ldots, P_{r-1})$ in $\RRR_{\pi}$ we define the \emph{signature} of $\SSS$ (w.r.t.~$X_i$) to be the tuple $(\chi(P_1)\cap\chi(X_i), \ldots, \chi(P_{r-1})\cap\chi(X_i))$. Observe that if $\SSS_1$ and $\SSS_2$ have the same signature w.r.t.~$X_i$, then $\chi(\SSS_1)\cup (\chi(\SSS_2)\cap \chi(X_i))=\chi(\SSS_1)$ and $\chi(\SSS_2)\cup (\chi(\SSS_1)\cap \chi(X_i))=\chi(\SSS_2)$; hence, either $\SSS_1\preceq_i \SSS_2$ or $\SSS_2\preceq_i \SSS_1$. Since all sequences in $\RRR_{\pi}^\Lambda$ have the same length, it follows from property (i) of representative sets that no two sequences in $\RRR_{\pi}$ have the same signature w.r.t.~$X_i$. Now let $\SSS = (P_1,\ldots, P_{r-1})$ be a sequence in $\RRR_{\pi}$ with a signature $(C_1,\ldots, C_{r-1})$. Note that if $C_j\neq\emptyset$, then $P_j$ is not the empty path, and hence $\sigma_j=1$ and the length of $P_j$ is $\lambda_j$. For each $j\in [r-1]$ such that $C_j\neq\emptyset$, there is a path $P\in \PPP_j$ such that $\chi(P)\cap\chi(X_i) = C_j$; otherwise, since $\chi(P_j)\cap\chi(X_i)= C_j$ and the length of $P_j$ is $\lambda_j$, $\PPP_j\cup \{P_j\}$ would also be a $\lambda_j$-minimal set of paths w.r.t.~$X_i$, which contradicts our choice of $\PPP_j$.
	 It follows that the number of signatures of sequences in $\RRR_{\pi}^\Lambda$ is at most
	$\prod_{j=1}^{r-1}|\PPP_j|\le \prod_{j=1}^{r-1}h_1(k,\lambda_j) \leq h_1(k,\ell)^{|X_i|^2}$. Since no two distinct sequences in $\RRR_{\pi}^\Lambda$ have the same signature, it follows that $|\RRR_{\pi}^\Lambda| \leq h_1(k,\ell)^{|X_i|^2}$ and $|\RRR_{\pi}|\le 2^{|X_i|+\ell} \cdot h_1(k,\ell)^{|X_i|^2}\le h(k,\ell)^{|X_i|^2}$, where $h(k,\ell) =\Oh(c^{\ell k} \cdot k^k \cdot \ell^{3\ell k})$, for some constant $c >1$.
	\end{proof}

For each bag $X_i$, we maintain a table $\Gamma_i$ that contains, for each
pattern for $X_i$, a representative set of sequences $\RRR_{\pi}$ for $(X_i, \pi)$.
For two vertices $u, v \in X_i$ and two $u$-$v$ paths $P, P'$ in $G_{uv}^{i}$, we say that $P'$ \emph{refines} $P$ if $\chi(P') \subseteq \chi(P)$. For two sequences $\SSS=(P_1, \ldots, P_{r-1})$ and $\SSS'=(P'_1, \ldots, P'_{r-1})$ that conform to $(X_i, \pi)$, we say that $\SSS'$ \emph{refines} $\SSS$ if each path $P'_j$ refines $P_j$, for $j \in [r-1]$.

\begin{lemma}
\label{lem:ext_refinetime}
Let $X_i$ be a bag, $\pi =(v_1=s, \sigma_1, v_2, \sigma_2, \dots, \sigma_{r-1}, v_r=t)$ a pattern for $X_i$, and $\WWW=(W_1, \ldots, W_{r-1})$ a sequence such that each $W_j$ is a walk between vertices $v_j$ and $v_{j+1}$ in $G_{v_jv_{j+1}}^{i}$ satisfying $\chi(W_j) \leq k$. Then in time $\Oh(r\cdot(|V(G)|+|V(E)|))$ we can compute a sequence $\SSS=(P_1, \ldots, P_{r-1})$, where for each $j \in [r-1]$, $P_j$ is an induced path between $v_j$ and $v_{j+1}$ in $G_{v_jv_{j+1}}^{i}$ such that $\chi(P_j) \subseteq \chi(W_j)$ and the length of $P_j$ is at most the length of $W_j$.
\end{lemma}

\begin{proof}
For each walk $W_j$, $j \in [r-1]$, we do the following.
We form the subgraph $G'$ from $G_{v_jv_{j+1}}^{i}$ by removing every vertex $x$ in $G_{v_jv_{j+1}}^{i}$ that does not satisfy $\chi(x) \subseteq \chi(W_j)$. Clearly, $W_j$ is a subgraph of $G'$, and hence there exists a $v_j$-$v_{j+1}$ path of length at most the length of $W_j$ in $G'$. We find a shortest $v_j$-$v_{j+1}$ path in $G'$ in time $\Oh(|V(G)|+|E(G)|)$ and set $P_j$ to this path. Clearly, the computation of $\SSS$ takes time $\Oh(r\cdot (|V(G)|+|E(G)|))$.
\end{proof}

For a bag $X_i$, pattern $\pi$ for $X_i$, and a set of sequences (of walks) $\RRR$ that conform to $(X_i, \pi)$, we define the procedure {\bf Refine()} that takes the set $\RRR$ and outputs a set $\RRR'$ of sequences of length at most $\ell$ that conform to $(X_i, \pi)$, and does not violate property (i) of Definition~\ref{def:representivesetextended}. First, for each sequence $\SSS$ in $\RRR$, we compute a sequence $\SSS'$
that refines $\SSS$ and has length at most the length of $\SSS$, and replace $\SSS$ with $\SSS'$ in $\RRR$. Afterwards, we initialize $\RRR' = \emptyset$, and order the sequences in $\RRR$ arbitrarily. We iterate through the sequences in $\RRR$ in order, and add a sequence $\SSS_p$ to $\RRR'$ if $|\SSS_p|\le \ell$, there is no sequence $\SSS$ already in $\RRR'$ such that $\SSS \preceq_i  \SSS_p$ and $|\SSS|\le |\SSS_p|$, and there is no sequence $\SSS_q \in \RRR$, $q > p$ (\ie, $\SSS_q$ comes after $\SSS_p$ in the order), such that $\SSS_q \preceq \SSS_p$ and $|\SSS_q| \le |\SSS_p|$.

\begin{lemma}
\label{lem:ext_refine}
Let $X_i$ be a bag, $\pi$ 
a pattern for $X_i$, and $\WWW$ be a set of sequences of walks that conforms to $(X_i, \pi)$.
The procedure {\bf Refine()}, on input $\WWW$, produces a set of sequences of induced paths $\RRR'$ that conform to $(X_i, \pi)$ and satisfy property (i) of Definition~\ref{def:representivesetextended}, and such that for each sequence $\SSS \in \WWW$ with $|\SSS|\le \ell$, there is a sequence $\SSS' \in \RRR'$ satisfying $\SSS' \preceq_i \SSS$ and $|\SSS'|\le |\SSS|$. Moreover, the procedure runs in time $\Oh^{*}(|\WWW|^2)$.
\end{lemma}

\begin{proof}
By Lemma~\ref{lem:ext_refinetime}, refining a sequence in $\WWW$ takes $\Oh(|V(G)|+|E(G)|)$ time, and hence, refining all sequences in $\WWW$ takes $\Oh^{*}(|\WWW|)$ time. After refining $\WWW$, we initialize $\RRR'$ to the empty set, and iterate through the sequences in $\WWW$, adding a sequence $\SSS_p \in \WWW$ to $\RRR'$ if: $|\SSS_p|\le \ell$, there is no sequence $\SSS$ already in $\RRR'$ such that $\SSS \preceq_i  \SSS_p$, and $|\SSS| \le  |\SSS_p|$. Clearly, this takes $\Oh^{*}(|\WWW|^2)$ time. Moreover, for a sequence $\SSS \in \WWW$ with $|\SSS|\le \ell$, the refined sequence $\SSS'$ we obtained from the application of Lemma~\ref{lem:ext_refinetime} to $\SSS$ satisfies $|\SSS'|\le \ell$. The lemma follows.
\end{proof}

If a bag $X_i$ is a leaf in $\mathcal{T}$, then $X_i=V_i=\{s, t\}$, and there are only two patterns $(s,0,t)$ and $(s,1,t)$ for $X_i$. Clearly, the only sequence that conforms to $(s,0,t)$ is the sequence $(())$ containing exactly one empty path. Moreover, there is  no edge $st \in E(G)$. Therefore, there is no sequence that conforms to $(s,1,t)$, and the following claim holds:
\begin{claim}\label{claim:ext_leaf}
	If a bag $X_i$ is a leaf in $\mathcal{T}$, then $\Gamma_i=\{((s,0,t), \{(())\} ), ((s,1,t), \emptyset)\}$ contains, for each
	pattern for $X_i$, a representative set for $(X_i, \pi)$.
\end{claim}	

 We describe next how to update the table stored at a bag $X_i$, based on the tables stored at its children in $\mathcal{T}$.
We distinguish the following cases based on the type of bag $X_i$.

\begin{enumerate}
	\item[Case 1.] $X_i$ is an introduce node with child $X_j$. Let $X_i= X_j \cup\{v\}$.
\end{enumerate}
Clearly, for every pattern $\pi$ for $X_i$ that does not contain $v$, we can set $\Gamma_i[\pi] = \Gamma_j[\pi]$. $\Gamma_j[\pi]$ is a representative set for $(X_i, \pi)$ for the following reasons: (i) follows because every color in $\chi(X_i)\setminus \chi(X_j)$ does not appear in $V_j$, since $X_i$ is a vertex-separator in $G$ separating $v$ and $V_j$ and colors are connected. Hence, if two sequences in $\Gamma_j[\pi]$ that conform to $(X_i, \pi)$ contradict (i), they contradict (i) w.r.t.~$(X_j, \pi)$ as well, thus contradicting that $\Gamma_j[\pi]$ is a representative set for $(X_j, \pi)$. For property (ii), it is easy to observe that $v$ does not appear on any path between two vertices in $\pi$ having internal vertices in $V_i\setminus X_i$, and hence, this property is inherited from the child node $X_j$.

Now let $\pi=(v_1=s, \sigma_1, v_2, \sigma_2, \dots, \sigma_{r-1}, v_r=t)$ be a pattern such that $v_q=v$, $q\in \{2, \ldots,r-1\}$, and let $\pi'=(v_1, \sigma_1,\ldots v_{q-1}, 0, v_{q+1},\sigma_{q+1}, \ldots, \sigma_{r-1}, v_r)$.
Note that since $X_j$ is a separator between $v$ and $V_j$, the only possibility for a path from $v$ to a different vertex in $X_i$ to have all internal vertices in $V_i\setminus X_i$ is if it is a direct edge. Therefore, if $\sigma_{q-1}=1$ (resp. $\sigma_{q}=1$) then $v_{q-1}v$ (resp. $v_qv$) is an edge in $G$. Otherwise, there is no sequence that conforms to $(X_i, \pi)$.

We obtain $\Gamma_i[\pi]$ from $\Gamma_j[\pi']$ as follows. For every $\SSS' = (P'_1, P'_2, \ldots, P'_{r-2}) \in \Gamma_j[\pi']$, we replace the empty path corresponding to $0$ between $v_{q-1}$ and $v_{q+1}$ in $\pi'$ by two paths  $P_{q-1}, P_{q}$ such that $P_{q-1} = ()$ (resp. $P_{q} = ()$)  if $\sigma_{q-1}=0$ (resp. $\sigma_{q}=0$) and $P_{q-1} = (v_{q-1},v)$ (resp. $P_{q-1} = (v,v_{q})$) otherwise and we obtain $\SSS=(P'_1, \dots, P'_{q-2}, P_{q-1}, P_{q}, P'_q, \ldots, P'_{r-2})$. Denote by $\RRR_\pi$ the set of all formed sequences $\SSS$.  Finally, we set $\Gamma_i[\pi]= \mbox{\bf Refine}(\RRR_\pi)$. We claim that $\Gamma_i[\pi]$ is a representative set for $(X_i, \pi)$.

\begin{claim}\label{claim:ext_introduce}
	If $X_i$ is an introduce node with child $X_j$, and $\Gamma_{j}$ contains for each
	pattern $\pi'$ for $X_j$ a representative set for $(X_j, \pi')$, then $\Gamma_i[\pi]$ defined above is a representative set for $(X_i, \pi)$.
\end{claim}
\begin{proof}

From the application of {\bf Refine()}, it is clear that $\Gamma_i[\pi]$ does not violate property (i) of the definition of representative sets. Assume now that there exists a sequence $\SSS \notin \Gamma_i[\pi]$ of length at most $\ell$ that conforms to $(X_i, \pi)$ and violates property (ii) of Definition~\ref{def:representivesetextended}. We define the sequence $\SSS'$ that conforms to $\pi'$, and is the same as $\SSS$ on all paths that $\pi$ and $\pi'$ share. Since no two paths in $\SSS$ share a vertex that is not in $X_i$ (since $\SSS$ violates (ii)), and all paths in $\SSS'$ are also in $\SSS$, it follows that no two paths in $\SSS'$ share a vertex that is not in $X_j$. Moreover, $|\SSS'|\le |\SSS|\le \ell$. Since $\Gamma_j[\pi']$ is a representative set for $(X_j, \pi')$, it follows that there exists $\SSS'_1 \in \Gamma_j[\pi']$ such that $\SSS'_1 \preceq_j \SSS'$ and $|\SSS'_1|\le |\SSS'|$. Let $\SSS_1$ be the sequence
obtained from $\SSS'_1$ and conforming to $(X_i, \pi)$. Then $\SSS_1 \in \RRR_\pi$ and it is easy to verify that $|\SSS_1|\le \ell$. Hence by Lemma~\ref{lem:ext_refine}, there is a sequence $\SSS_2 \in \Gamma_i[\pi]$ such that $\SSS_2 \preceq_i \SSS_1$ and $|\SSS_2|\le |\SSS_1|$. Since either both $\SSS_1$ and $\SSS$ contain $v$ or none of them does, we have $\SSS_1 \preceq_i \SSS$ and $|\SSS_1|\le |\SSS|$. By transitivity of $\preceq_i$ (Lemma~\ref{lem:transitive}) and of $\le$, it follows that $\SSS_2 \preceq_i \SSS$ and $|\SSS_2|\le |\SSS|$. This contradicts the assumption that $\SSS$ violates property (ii).
\end{proof}

\begin{enumerate}
	\item[Case 2.] $X_i$ is a forget node with child $X_j$. Let $X_i= X_j\setminus\{v\}$.
\end{enumerate}

Let $\pi = (s=v_1,\sigma_1,\ldots,\sigma_{r-1}, v_r=t)$ be a pattern for $X_i$. For $q\in[r-1]$ such that $\sigma_q=1$, we define  $\pi^q=(s=v'_1,\sigma'_1,\ldots, \sigma'_{r}, v'_{r+1}=t)$ to be the pattern obtained from $\pi$ by inserting $v$ between $v_q$ and $v_{q+1}$ and setting $\sigma'_q=\sigma'_{q+1}=1$. More precisely, we set $v'_p=v_p$ and $\sigma'_p=\sigma_p$  for $1\le p\le q$, $v'_{q+1}=v$ and $\sigma'_{q+1}=1$, and finally $v'_p=v_{p-1}$ and $\sigma'_p=\sigma_{p-1}$  for $q+2\le p\le r$.
We define $\RRR_{\pi}$ as follows:

{\small
\begin{equation*}
 \RRR_\pi  = \Gamma_j[\pi]  \cup  \{\SSS = (P_1, \ldots, P_{q-1}, P_q\circ P_{q+1}, P_{q+2} ,\ldots, P_{r})\mid (P_1, \ldots ,P_{r})\in \Gamma_j[\pi^q], q\in [r-1]\wedge \sigma_q=1 \}.
\end{equation*}
}

Finally, we set $\Gamma_i[\pi]= \mbox{\bf Refine}(\RRR_\pi)$ and we claim that $\Gamma_i[\pi]$ is a representative set for $(X_i, \pi)$.

\begin{claim}\label{claim:ext_forget}
	If $X_i$ is a forget node with child $X_j$, and $\Gamma_{j}$ contains for each
	pattern $\pi'$ for $X_j$ a representative set for $(X_j, \pi')$, then $\Gamma_i[\pi]$ defined above is a representative set for $(X_i, \pi)$.
\end{claim}
\begin{proof}

It is straightforward to see that $\Gamma_i[\pi]$ satisfies property (i) due to the way procedure {\bf Refine()} works. Assume for a contradiction that there exists a sequence $\SSS$ that violates property (ii). We distinguish two cases.

First, suppose that no path in $\SSS$ contains $v$. Then $\SSS$ conforms to $(X_j, \pi)$. Since no two paths in $\SSS$ share a vertex that is not in $X_i$, and since $\Gamma_j[\pi]$ is a representative set, there exists $\SSS_1 \in \Gamma_j[\pi]$ such that $\SSS_1 \preceq_j \SSS$ and $|\SSS_1| \le |\SSS|\le \ell$.
Then $\SSS_1\in \RRR_\pi$, and hence by Lemma~\ref{lem:ext_refine}, $\Gamma_i[\pi]$ contains a sequence $\SSS_2$ such that
 $\SSS_2 \preceq_i \SSS_1$ and $|\SSS_2|\le |\SSS_1|$. Since $\SSS_1 \preceq_j \SSS$ and $X_i \subsetneq X_j$, it follows from Observation~\ref{obs:jtoi} that $\SSS_1 \preceq_i \SSS$. By transitivity of $\preceq_i$, it follows that $\SSS_2 \preceq_i \SSS$. Moreover, $|\SSS_2|\le |\SSS|$, which is a contradiction to the assumption that $\SSS$ violates property (ii).

Second, suppose that there is a path $P_q$ in $\SSS$ between $v_q$ and $v_{q+1}$ that contains $v$. We form a sequence $\SSS'$ from $\SSS$ by keeping every path $P \neq P_q$ in $\SSS$, and replacing $P_q$ in the sequence by the two subpaths of $P_q$, $P'_q = (v_q, \ldots, v)$ and $P'_{q+1} = (v, \ldots, v_{q+1})$. The sequence $\SSS'$ conforms to $(X_j, \pi^{q})$, and since no two paths in $\SSS$ share a vertex that is not in $X_i$, no two paths in $\SSS'$ share a vertex that is not in $X_j$. Moreover, it is straightforward that $|\SSS'|=|\SSS|$. Since $\Gamma_{j}[\pi^{q}]$ is a representative set for $(X_j, \pi^{q})$, it follows that there exists a sequence $\SSS'_1 \in \Gamma_{j}[\pi^{q}]$ such that $\SSS'_1 \preceq_j \SSS'$ and $|\SSS'_1|\le |\SSS'|$. Let $\SSS_1$ be the sequence conforming to $(X_i, \pi)$ obtained from $\SSS'_1$ by applying the operation $\circ$ to the two paths in $\SSS'_1$ that share $v$. Then $\SSS_1 \in \RRR_{\pi}$ and $|\SSS_1|=|\SSS'_1|\le \ell$. Therefore, by Lemma~\ref{lem:refine}, $\Gamma_i[\pi]$ contains a sequence $\SSS_2$ such that $\SSS_2 \preceq_i \SSS_1$. Since $\SSS'_1 \preceq_j \SSS'$, $\chi(\SSS') = \chi(\SSS)$, $\chi(\SSS'_1) = \chi(\SSS_1)$, and $X_i \subsetneq X_j$, it follows that $\SSS_1 \preceq_i \SSS$. By transitivity of $\preceq_i$, it follows that $\SSS_2 \preceq_i \SSS$. Moreover $|\SSS_2|\le |\SSS|$, which is a contradiction to the assumption that $\SSS$ violates property (ii).
\end{proof}

\begin{enumerate}
	\item[Case 3.] $X_i$ is a join node with children $X_j$, $X_{j'}$.
\end{enumerate}

Let $\pi = (s=v_1,\sigma_1, \dots, \sigma_{r-1}, v_r=t)$ be a pattern for $X_i$. Initialize $\RRR_{\pi}=\emptyset$. For every two patterns $\pi_1=(s=v_1,\tau_1, \dots, \tau_{r-1}, v_r=t)$ and $\pi_2=(s=v_1,\mu_1, \dots, \mu_{r-1}, v_r=t)$ such that $\sigma_q= \tau_q + \mu_q$, and for every two sequences $\SSS_1 =(P_{1}^{1}, \ldots, P_{1}^{r-1}) \in \Gamma_j[\pi_1]$ and $\SSS_2 =(P_{2}^{1}, \ldots, P_{2}^{r-1}) \in \Gamma_{j'}[\pi_2]$, we add the sequence $\SSS=(P_1, \ldots, P_{r-1})$ to $\RRR_{\pi}$, where $P_q = P_{1}^{q}$ if $P_{2}^{q}$ is the empty path, otherwise, $P_q = P_{2}^{q}$, for $q \in [r-1]$. We set $\Gamma_i[\pi]= \mbox{\bf Refine}(\RRR_\pi)$, and we claim that $\Gamma_i[\pi]$ is a representative set for $(X_i, \pi)$.

\begin{claim}\label{claim:ext_join}
	If $X_i$ is a join node with children $X_j$, $X_{j'}$, and $\Gamma_{j}$ (resp.~$\Gamma_{j'}$) contains for each
	pattern $\pi'$ for $X_j=X_{j'}=X_i$ a representative set for $(X_j, \pi')$ (resp.~$(\pi', X_{j'})$), then $\Gamma_i[\pi]$ defined above is a representative set for $(X_i, \pi)$.
\end{claim}

\begin{proof}
Clearly $\Gamma_i[\pi]$ satisfies property (i) due to the application of the procedure {\bf Refine()}. To argue that $\Gamma_i[\pi]$ satisfies properties (ii), suppose not, and let $\SSS =(P_1, \ldots, P_{r-1})$ be a sequence that violates property (ii). Notice that every path $P_q$, $q \in [r-1]$ is either an edge between two vertices in $X_i$, or is a path between two vertices in $X_i$ such that its internal vertices are either all in $V_j \setminus X_i$ or in  $V_{j'} \setminus X_i$; this is true because $X_i$ is a vertex-separator separating $V_j \setminus X_i$ from $V_{j'} \setminus X_i$ in $G$. Define the two sequences $\SSS_1 =(P_{1}^{1}, \ldots, P_{1}^{r-1})$ and $\SSS_2 =(P_{2}^{1}, \ldots, P_{2}^{r-1})$ as follows. For $q \in [r-1]$, if $P_q$ is empty then set both $P_{1}^{q}$ and $P_{2}^{q}$ to the empty path; if $P_q$ is an edge then set $P_{1}^{q} =P_q$ and $P_{2}^{q}$ to the empty path. Otherwise, $P_q$ is either a path in $G[V_j]$ or in $G[V_{j'}]$; in the former case set $P_{1}^{q} =P_q$ and $P_{2}^{q}$ to the empty path, and in the latter case set $P_{2}^{q} =P_q$ and $P_{1}^{q}$ to the empty path. Since no two paths in $\SSS$ share a vertex that is not in $X_i$, and $X_i=X_j=X_{j'}$, no two paths in $\SSS_1$ (resp.~$\SSS_2$) share a vertex that is not in $X_j$ (resp.~$X_{j'}$). Moreover, it is easy to see that $|\SSS_1|+|\SSS_2|=|\SSS|\le \ell$. Let $\pi_1=(s=v_1,\tau_1, \dots, \tau_{r-1}, v_r=t)$ and $\pi_2=(s=v_1,\mu_1, \dots, \mu_{r-1}, v_r=t)$ be the two patterns that $\SSS_1$ and $\SSS_2$ conform to, respectively, and observe that, for every $q \in [r-1]$, we have $\sigma_q= \tau_q + \mu_q$. Since $\Gamma_j[\pi_1]$ and $\Gamma_{j'}[\pi_2]$ are representative sets, it follows that there exist
$\SSS'_1 = (Y'_{1}, \ldots, Y'_{r-1})$ in $\Gamma_j[\pi_1]$ and $\SSS'_2 = (Z'_{1}, \ldots, Z'_{r-1})$ in $\Gamma_{j'}[\pi_2]$ such that $\SSS'_1 \preceq_j \SSS_1$, $|\SSS'_1| \le |\SSS_1|$ and $\SSS'_2 \preceq_{j'} \SSS_2$, $|\SSS'_2| \le |\SSS_2|$. Let $\SSS'=(P'_1, \ldots, P'_{r-1})$, where $P'_q = Y'_{q}$ if $Z'_{q}$ is the empty path, otherwise, $P_q = Z'_{q}$, for $q \in [r-1]$. The sequence $\SSS'$ conforms to $\pi$, is in $\RRR_\pi$,  and $|\SSS'|=|\SSS'_1|+|\SSS'_2|\le |\SSS|$. By Lemma~\ref{lem:ext_refine}, $\Gamma_i[\pi]$ contains a sequence $\SSS''$ such that $\SSS'' \preceq_i \SSS'$ and $|\SSS''| \le |\SSS'|$. From Observation~\ref{obs:jtoi}, since $X_i=X_j=X_{j'}$, from $\SSS'_1 \preceq_j \SSS_1$ and $\SSS'_2 \preceq_{j'} \SSS_2$ it follows that $\SSS'_1 \preceq_i \SSS_1$ and $\SSS'_2 \preceq_i \SSS_2$. Since $\chi(\SSS_1) \cup \chi(\SSS_2) = \chi(\SSS)$ and $\chi(\SSS'_1) \cup \chi(\SSS'_2) = \chi(\SSS')$, and since $\chi(\SSS_1) \cap \chi(\SSS_2) \subseteq \chi(X_i)$, by Lemma~\ref{lem:join}, it follows that $\SSS' \preceq_i \SSS$. Since $\SSS'' \preceq_i \SSS'$, by transitivity of $\preceq_i$, it follows that $\SSS'' \preceq_i \SSS$. Moreover $|\SSS''|\le |\SSS|$, which concludes the proof.
\end{proof}

We can now conclude with the following theorem:

\begin{theorem}
\label{thm:ext_treewidth}
There is an algorithm that on input $(G, C, \chi, s, t, k, \ell)$ of \bcmor{}, either outputs a $k$-valid $s$-$t$ path of length at most $\ell$ in $G$, or decides that no such path exists, in time $\Oh^{\star}(f(k,\ell)^{37\ell^2})$, where $f(k,\ell) =\Oh(c^{\ell k} \cdot k^k \cdot \ell^{3\ell k})$, for some constant $c >1$. Therefore, \bcmor{} parameterized by both $k$ and $\ell$ is \FPT.	
\end{theorem}

\begin{proof}
If $d_G(s,t)>\ell$, then, by definition, there is no $s$-$t$ path of length at most $\ell$. Hence, we assume that $d_G(s,t)\le \ell$. By Lemma~\ref{lem:ext_distantvertex}, if there exists a vertex $v$ such that $d_G(s, v)>\ell+1$, we can contract any edge incident to $v$ and obtain an equivalent instance. The contraction of a single edge can be done in time polynomial in the size of the instance, and after applying Lemma~\ref{lem:ext_distantvertex} $|E|$ times, we would get a trivial instance. Moreover, from the proof of Lemma~\ref{lem:ext_distantvertex}, it follows that we can obtain a solution in the original instance from a solution in the contracted instance in polynomial time. Therefore, we can assume for the rest of the proof that we applied Lemma~\ref{lem:ext_distantvertex} exhaustively, and hence $G$ has radius at most $\ell+1$ and treewidth $\omega$ that is at most $3\ell +4$~\cite{rs}. Moreover, a tree decomposition of $G$ of width $\omega$ can be computed in (polynomial) time $\Oh(\ell \cdot n)$~\cite{ioannis}. From such a tree decomposition, in polynomial time we can compute a nice tree decomposition $(\VVV, \TTT)$ of $G$ whose width is at most $\omega \leq 3\ell+4$ and satisfying $|\VVV| = \Oh(|V(G)|)$~\cite{k94}.

The algorithm starts by removing the colors of $s$ and $t$ from $G$, and decrements $k$ by $|\chi(s)\cup\chi(t)|$ (see Assumption~\ref{ass:sat}). Afterwards, if $k<0$, the algorithm concludes that there is no $k$-valid $s$-$t$ path in $G$. If $st\in E(G)$ and $k\ge 0$, the algorithm outputs the path $(s,t)$. Now we know that $s$ and $t$ are not adjacent, and that $\chi(s)=\chi(t)=\emptyset$. The algorithm then adds $s$ and $t$ to every bag in $\TTT$, and executes the dynamic programming algorithm based on $(\VVV, \TTT)$, described in this section, to compute a table $\Gamma_i$ that contains, for each
bag $X_i$ in $\TTT$ and each pattern $\pi$ for $X_i$, a representative set $\RRR_{\pi}$ for $(X_i, \pi)$.
	
From Claims~\ref{claim:ext_leaf}, \ref{claim:ext_introduce}, \ref{claim:ext_forget}, \ref{claim:ext_join}, it follows, by induction on the height of the tree-decomposition $(\mathcal{V}, \TTT)$ (the base case corresponds to the leaves), that the root node $X_r$ contains a representative set $\Gamma_r[\pi]$ for the sequence $\pi=(s,1,t)$. If $\Gamma_r[\pi]$ is empty, the algorithm concludes that there is no $k$-valid $s$-$t$ path of length at most $\ell$ in $G$. Otherwise, noting that there is only one sequence $\SSS$ in the representative set $\Gamma_r[\pi]$ since $X_r=\{s, t\}$ and $s$ and $t$ are empty, the algorithm outputs the $k$-valid $s$-$t$ path $P$ of length at most $\ell$ formed by $\SSS$. The correctness follows from the following argument, which shows that if there is a $k$-valid $s$-$t$ path of length at most $\ell$ in $G$, then the algorithm outputs such a path. Suppose that $P'$ is a $k$-valid induced $s$-$t$ path of length at most $\ell$ and let $\SSS'=(P')$. Since $G_{st}^r=G$, it follows that $\SSS'$ conforms to $(X_r, \pi)$. Moreover $|\SSS'|\le \ell$ and  $\SSS'$ contains exactly one path that is induced, hence no two paths in $\SSS'$ share a vertex. Therefore, by property (ii) of representative sets, there exists a sequence $\SSS$ in $\Gamma_r[\pi]$ satisfying $\SSS \preceq_r\SSS'$ and $|\SSS|\le |\SSS'|$. Noting that a sequence in $\Gamma_r[\pi]$ must consist of a single $k$-valid $s$-$t$ path of length at most $\ell$, it follows that the algorithm correctly outputs such a path.

Next, we analyze the running time of the algorithm. We observe that among the three types of bags in $\TTT$, the worst running time is for a join bag. Therefore, it suffices to upper bound the running time for a join bag, and since $|\VVV|=\Oh(n)$, the upper bound on the overall running time would follow.

Consider a join bag $X_i$ with children $X_j, X_{j'}$. Let $\omega'$ be the width of $\TTT$ plus 1, which serves as an upper bound on the bag size in $\TTT$. Therefore, we have $\omega' \leq 3\ell+7$, where the (additional) plus 2 is to account for the vertices $s$ and $t$ that were added to each bag.  The algorithm starts by enumerating each pattern $\pi$ for $X_i$. The number of such patterns is at most $2^{\omega'} \cdot \omega' \cdot \omega'!=\Oh^{*}(2^{\omega'} \cdot \omega'!)$, where $\omega' \cdot \omega'!$ is an upper bound on the number of ordered selections of a subset of vertices from the bag, and $2^{\omega'}$ is an upper bound on the number of combinations for the $\sigma_i$'s in the selected pattern. Fix a pattern $\pi$ for $X_i$. To compute $\Gamma_i[\pi]$, the algorithm enumerates all ways of partitioning $\pi$ into pairs of patterns $\pi_1, \pi_2$ for the children bags; there are $2^{\omega'}$ ways of partitioning $\pi$ into such pairs, because for each $\sigma_i=1$ in $\pi$, the path between $v_i$ and $v_{i+1}$ is either reflected in $\pi_1$ or in $\pi_2$. For a fixed pair $\pi_1, \pi_2$, the algorithm iterates through all pairs of sequences in the two tables $\Gamma_j[\pi_1]$ and $\Gamma_{j'}[\pi_2]$. Since each table contains a representative set, by Lemma~\ref{lem:representativesboundextended}, the size of each table is $h_1(k,\ell)^{\omega'^2}$, where $h_1(k,\ell) =\Oh(c_1^{\ell k} \cdot k^k \cdot \ell^{3\ell k})$ for some constant $c_1>1$,
 and hence iterating over all pairs of sequences in the two tables can be done in $\Oh(h_1(k,\ell)^{2\omega'^2})$ time. From the above, it follows that the set $\RRR_{\pi}$ can be computed in time $2^{\omega'} \cdot \Oh(h_1(k,\ell)^{2\omega'^2})= \Oh(h_2(k,\ell)^{2\omega'^2})$, where $h_2(k,\ell) =\Oh(c_2^{\ell k} \cdot k^k \cdot \ell^{3\ell k})$, for some constant $c_2 >1$, which is also an upper bound on the size of $\RRR_{\pi}$. By Lemma~\ref{lem:ext_refine}, applying {\bf Refine()} to $\RRR_{\pi}$ takes time $\Oh^{*}(h_2(k,\ell)^{4\omega'^2})$. It follows from all the above that the running time taken by the algorithm to compute $\Gamma_i$ is $\Oh^{*}(h_2(k,\ell)^{4\omega'^2} \cdot 2^{\omega'} \cdot \omega'!)=\Oh^{*}(h_3(k,\ell)^{4\omega'^2})$, where $h_3(k,\ell) =\Oh(c_3^{\ell k} \cdot k^k \cdot \ell^{3\ell k})$, for some constant $c_3 >1$, and hence
the running time of the algorithm is $\Oh^{\star}(f(k,\ell)^{37\ell^2})$, where $f(k,\ell) =\Oh(c^{\ell k} \cdot k^k \cdot \ell^{3\ell k})$, for some constant $c >1$.
\end{proof}

\subsection{Applications}
\label{subsec:applications}
In this subsection, we describe some applications of Theorem~\ref{thm:ext_treewidth}. The first result is a direct consequence of Theorem~\ref{thm:ext_treewidth}.

\begin{corollary}
\label{cor:boundedlengthaapp}
For any computable function $h$, the restriction of \cmor{} to instances in which the length of the sought path is at most $h(k)$ is \FPT{} parameterized by $k$.
\end{corollary}

We note that the above restriction of \cmor{} is \NP-hard, as a consequence of (the proof of) Corollary~\ref{cor:pathwidthnphardness}.

Corollary~\ref{cor:boundedlengthaapp} directly implies Kroman \etal's results~\cite{korman}, showing that \gmor{} parameterized by $k$ is \FPT{} for unit-disk obstacles and for similar-size fat regions with constant overlapping number. Using Bereg and Kirkpatrick's result~\cite{kirkpatrick3}, the length of a shortest $k$-valid path for unit-disk obstacles is at most $3k$ (see also Lemma 3 in Korman \etal~\cite{korman}). By Corollary~2 in~\cite{korman},
the length of a shortest $k$-valid path for similar-size fat-region obstacles with constant overlapping number is linear in $k$. Corollary~\ref{cor:boundedlengthaapp} generalizes these \FPT{} results, which required quite some effort, and provides an explanation to why the problem is \FPT{} for such restrictions, namely because the path has length upper bounded by a function of the parameter. In particular, one may allow the connected obstacles to be of various shapes and sizes, provided that the length of the path is upper bounded by a function of the parameter.

The second application we describe is related to an open question posed in~\cite{lavalle}. For an instance $I=(G, C, \chi, s, t, k)$ of \cmor{}, and a color $c \in C$, define the \emph{intersection number} of $c$, denoted $\iota(c)$, to be the number of vertices in $G$ on which $c$ appears. Define the \emph{intersection number} of $G$, $\iota(G)$, as $\max\{\iota(c) \mid c \in C\}$.
Consider the following problem:

\paramproblem{\bicmor} {A planar graph $G$ such that $\iota(G) \leq i$; a set of colors $C$; $\chi: V \longrightarrow 2^{C}$; two designated vertices $s, t \in V(G)$; and $k, i \in \nat$}{Does there exist a $k$-valid $s$-$t$ path in $G$?} \\

Again, the above problem is \NP-hard, as a consequence of (the proof of) Corollary~\ref{cor:pathwidthnphardness}.

\begin{corollary}
\label{thm:boundedintersection}
\bicmor{} is \FPT{} parameterized by both $k$ and $i$.
\end{corollary}

\begin{proof}
Since the number of vertices in $G$ on which any color $c \in C$ appears is at most $\iota(G)$, the length of any $k$-valid $s$-$t$ path is $\Oh(k\cdot i)$. The result now follows from Theorem~\ref{thm:ext_treewidth}.
\end{proof}

The following corollary is a direct consequence of Corollary~\ref{thm:boundedintersection}:

\begin{corollary}
\label{cor:boundedintersection}
For any computable function $h$, \bicmor{} restricted to instances $(G, C, \chi, s, t, k)$ satisfying $\iota(G) \leq h(k)$ is \FPT{} parameterized by $k$.
\end{corollary}

Corollary~\ref{thm:boundedintersection} has applications pertaining to instances of \gcmor{} whose auxiliary graph is an instance of \bicmor{}. In particular, an interesting case that was studied corresponds to the case in which the obstacles are convex polygons, each intersecting at most a constant number of other polygons. The complexity of this problem was left as an open question in~\cite{lavalle,hauser}, and remains unresolved. The result in Corollary~\ref{cor:boundedintersection} subsumes this case, and even the more general case in which the obstacles are arbitrary convex obstacles satisfying that each obstacle intersects a constant number of other obstacles, as it is easy to see that the auxiliary graph of such instances will have a constant intersection number\footnote{Note that convexity is essential here, as otherwise, the intersection number of the auxiliary graph may be unbounded.}. In fact, we can even allow the intersection number to be any (computable) function of the parameter:

\begin{theorem}
\label{thm:boundedintersectiongeometric}
Let $h$ be a computable function. The restriction of \gcmor{} to any set of convex obstacles in the plane satisfying that each obstacle intersects at most $h(k)$ other obstacles, is \FPT{} parameterized by $k$.
\end{theorem}

Whereas the complexity of the problem in the above theorem is open, the theorem settles its parameterized complexity by showing it to be in \FPT{}.
\fi

\ifshort \newpage \fi

\bibliography{ref}

\end{document}